\documentclass[pdflatex,sn-mathphys-num]{sn-jnl}
\usepackage{graphicx}%
\usepackage{multirow}%
\usepackage{amsmath,amssymb,amsfonts}%
\usepackage{amsthm}%
\usepackage{mathrsfs}%
\usepackage[title]{appendix}%
\usepackage{xcolor}%
\usepackage{textcomp}%
\usepackage{manyfoot}%
\usepackage{booktabs}%
\usepackage{algorithm}%
\usepackage{algorithmicx}%
\usepackage{algpseudocode}%
\usepackage{listings}%

\usepackage{mathtools}
\usepackage{thmtools, thm-restate}
\usepackage{bussproofs}
\usepackage{bbding}
\usepackage{hyperref}
\usepackage{cleveref}

\theoremstyle{thmstyleone}
\newtheorem{theorem}{Theorem}
\newtheorem{lemma}[theorem]{Lemma}
\newtheorem{proposition}[theorem]{Proposition}
\newtheorem{definition}[theorem]{Definition}
\newtheorem{conjecture}[theorem]{Conjecture}

\theoremstyle{thmstyletwo}
\newtheorem{example}[theorem]{Example}

\newcommand{\abs}[1]{\left|#1\right| }
\newcommand{\suchthat}{\ |\ }

\newcommand{\mset}[1]{\left\{#1\right\}}
\newcommand{\union}{\cup}

\newcommand{\Union}{\bigcup}
\newcommand{\seq}[1]{\overline{#1}}

\renewcommand{\iff}{\equiv}
\newcommand{\imp}{\Rrightarrow}
\newcommand{\rimp}{\Lleftarrow}
\newcommand{\limp}{\rightarrow}
\newcommand{\liff}{\leftrightarrow}
\newcommand{\Lor}{\bigvee}
\newcommand{\Land}{\bigwedge}
\newcommand{\oeq}{\simeq}
\newcommand{\noeq}{\not\oeq}

\newcommand{\resclosure}[2]{\mathrm{Res}_{#1}^{<\omega}({#2})}

\newcommand{\dom}{\mathrm{dom}}

\newcommand{\N}{\mathbb{N}}

\newcommand{\Exists}[1]{\exists #1 \, }
\newcommand{\Forall}[1]{\forall #1 \, }

\newcommand{\FEQ}{\ensuremath{\mathrm{FEQ}}}

\newcommand{\WSOQE}{\ensuremath{\mathrm{WSOQE}}}

\newcommand{\witness}{\alpha}
\newcommand{\velim}{\rightarrow_\mathrm{ve}}
\newcommand{\velimtrans}{\velim^\ast}
\newcommand{\subsumes}{\trianglelefteq}
\newcommand{\subsumesL}[1]{\subsumes_{#1}}
\newcommand{\subsumesLvelim}[1]{\subsumes_{#1,\mathrm{ve}}}
\newcommand{\subsumesLvelimtrans}[1]{\subsumesLvelim{#1}^\ast}

\usepackage{todonotes}
\newboolean{todoactive}
\setboolean{todoactive}{false}
\let\todonote=\todo
\newcommand{\todocolor}{red}
\renewcommand{\todo}[1]{%
  \ifthenelse{\boolean{todoactive}}%
  {{\color{\todocolor}\footnote{\todonote[inline,size=\footnotesize,textcolor=\todocolor,bordercolor=white,backgroundcolor=white,noinlinepar]{TODO #1}}}}%
  {}%
}

\AtBeginEnvironment{proof}{
  \makeatletter
  \let\oldenumerate\enumerate
  \renewcommand{\enumerate}{\oldenumerate\small\normalfont}
  \makeatother
}
\AtBeginEnvironment{theorem}{
  \makeatletter
  \let\oldenumerate\enumerate
  \renewcommand{\enumerate}{\oldenumerate\small\normalfont}
  \makeatother
}
\AtBeginEnvironment{lemma}{
  \makeatletter
  \let\oldenumerate\enumerate
  \renewcommand{\enumerate}{\oldenumerate\small\normalfont}
  \makeatother
}
\AtBeginEnvironment{proposition}{
  \makeatletter
  \let\oldenumerate\enumerate
  \renewcommand{\enumerate}{\oldenumerate\small\normalfont}
  \makeatother
}
\AtBeginEnvironment{definition}{
  \makeatletter
  \let\oldenumerate\enumerate
  \renewcommand{\enumerate}{\oldenumerate\small\normalfont}
  \makeatother
}
\AtBeginEnvironment{conjecture}{
  \makeatletter
  \let\oldenumerate\enumerate
  \renewcommand{\enumerate}{\oldenumerate\small\normalfont}
  \makeatother
}
\AtBeginEnvironment{example}{
  \makeatletter
  \let\oldenumerate\enumerate
  \renewcommand{\enumerate}{\oldenumerate\small\normalfont}
  \makeatother
}
\AtBeginEnvironment{problem}{
  \makeatletter
  \let\oldenumerate\enumerate
  \renewcommand{\enumerate}{\oldenumerate\small\normalfont}
  \makeatother
}
\AtBeginEnvironment{remark}{
  \makeatletter
  \let\oldenumerate\enumerate
  \renewcommand{\enumerate}{\oldenumerate\small\normalfont}
  \makeatother
}

\newcommand{\orcidID}[1]{\href{https://orcid.org/#1}{#1}}

\begin{document}

\title{Computing Witnesses Using the SCAN Algorithm}

\author*[1]{\fnm{Fabian} \sur{Achammer} \orcidID{0009-0002-3799-6393} \email{fabian.achammer@tuwien.ac.at}}

\author[1]{\fnm{Stefan} \sur{Hetzl} \orcidID{0000-0002-6461-5982} \email{stefan.hetzl@tuwien.ac.at}}

\author[2]{\fnm{Renate} A.~\sur{Schmidt} \orcidID{0000-0002-6673-3333} \email{renate.schmidt@manchester.ac.uk}}

\affil*[1]{\orgdiv{Institute for Discrete Mathematics and Geometry}, \orgname{TU Wien}, \orgaddress{\city{Vienna}, \country{Austria}}}

\affil*[2]{\orgdiv{Department of Computer Science}, \orgname{University of Manchester}, \orgaddress{\city{Manchester}, \country{United Kingdom}}}

\abstract{
  Second-order quantifier elimination is the problem of finding, given a formula with second-order quantifiers, a logically equivalent first-order formula.
  While such formulas are not computable in general, there are practical algorithms and subclasses with applications throughout computational logic.
  One of the most prominent algorithms for second-order quantifier elimination is the saturation-based SCAN algorithm.
  In this paper we show how the SCAN algorithm on clause sets can be extended to solve a more general problem: namely, finding a witness for the second-order quantifiers that results in a logically equivalent first-order formula.
  In addition, we provide a prototype implementation of the proposed method.
}


\maketitle

\section*{Acknowledgements}
This research was funded in part by the Austrian Science Fund (FWF) grant number 10.55776/P35787.

The first author is grateful for the support in facilitating two research visits to the University of Manchester, which greatly benefited this work.

We thank Luke Sanderson who dockerfied the original version of SCAN with help from the University of Manchester Open Source Software Club.

\section{Introduction}
Given a formula with second-order quantifiers
$\Exists{\seq{X}} \varphi$, where $\varphi$ is first-order,
\emph{second-order quantifier elimination (SOQE)}~\cite{Gabbay92Quantifier,Gabbay08Second} is the problem of finding a first-order formula~$\psi$ such that
$$
  \Exists{\seq{X}} \varphi \iff \psi.
$$
For example, consider $\varPhi = \Exists{X} (X(a) \land \Forall{u} (X(u) \limp B(u)))$.
Then $\varPhi$ implies the first-order formula $B(a)$.
Furthermore, given a model $\mathcal{M}$ of $B(a)$ one can show~\mbox{$\mathcal{M} \models\Exists{X} (X(a) \land \Forall{u} (X(u) \limp B(u)))$} by setting $X = \mset{a^\mathcal{M}}$.
Thus $\varPhi$ is equivalent to the first-order formula $B(a)$.
An algorithm for eliminating existential quantifiers can
be used to eliminate universal quantifiers by writing
$\Forall{\seq{X}}$ as $\neg \exists \seq{X} \neg$.

SOQE is an important topic in logic, automated reasoning and artificial
intelligence that has real-world applications in diverse areas.
It has been used to automate correspondence theory for modal
logic~\cite{Gabbay92Quantifier}, which has led
to new attention and results in correspondence theory of various
algebras and logics~\cite{BrinkGabbayOhlbach95,Goranko03Scan,ConradieGorankoVakarelov06a,ConradieGhilardiPalmigiano14}.
SOQE can be used to transform a knowledge base (or set of formulas, given by~$\varphi$) into
a restricted view $\psi$ of the knowledge base in which all occurrences
of the predicate symbols $\overline X$ have been eliminated.
This reduction has also been referred to as
forgetting~\cite{LinReiter94} or projection and studied in the form
of uniform interpolation or strongest necessary condition in knowledge
representation~\cite{KonevWaltherWolter09,KoopmannSchmidt13a,KoopmannSchmidt13c,Delgrande17,FereeVanDerGiessenEtAl24,DohertyLukaszewiczSzalas01}
and answer set programming~\cite{EiterKernIsberner19,GoncalvesKnorrLeite23}.
Forgetting offers solutions to several important applications and manipulations of
ontologies: ontology extraction, ontology creation, reuse and comparison,
and abductive reasoning~\cite{ChenAlghamdiEtAl19a,LudwigKonev14,LiuLuEtAl21,KoopmannDelPintoEtAl20}.
In automated reasoning an application of SOQE is (predicate) symbol
elimination~\cite{Gabbay92Quantifier,HoderKovacsVoronkov10,KhasidashviliKorovin16,PeuterSofronieStokkermans21}.
Another very promising application domain is agent communication,
requiring knowledge sharing among agents~\cite{DohertySzalas04,ToluhiSchmidtParsia21}.

Two prominent approaches to computing SOQE-solutions
are the SCAN algorithm introduced
in~\cite{Gabbay92Quantifier} and the DLS algorithm
introduced in~\cite{Doherty97Computing} and extended
in~\cite{Doherty98General,Nonnengart98Fixpoint}.
SCAN is based on the idea of computing all logical
consequences of $\Exists{\seq{X}} \varphi$ and omitting those formulas that contain
predicate variables from $\seq{X}$.
SCAN transforms the first-order part $\varphi$ into clausal
normal form and computes the closure under a constraint resolution calculus,
performing inferences only on $\seq{X}$-literals. During this process
clauses with $\seq{X}$-literals for which sufficiently many inferences
have been performed are deleted~\cite{Engel96Quantifier,Ohlbach96SCAN}.
If the algorithm terminates, it returns a set of first-order consequences
which is logically equivalent to $\Exists{\seq{X}} \varphi$.
Reverse Skolemization might be applied,
if during the clause form transformation Skolemization was performed.
SCAN has been shown complete for computing frame correspondence properties
for the class of modal Sahlqvist axioms~\cite{Goranko03Scan}.
The SCAN algorithm is closely related to hierarchic
superposition~\cite{Bachmair94Refutational,BaumgartnerWaldmann19}.

The other prominent approach, the DLS algorithm, is based on a result of
Ackermann~\cite{Ackermann35Untersuchungen}.
Known as Ackermann's Lemma, it
states a condition under which a predicate variable is eliminable from a second-order formula $\exists X \varphi$.
The idea of the DLS algorithm is to use equivalence preserving
operations to bring the problem into a form
where Ackermann's Lemma is applicable.

In this paper we are interested in a more general problem which we call
\emph{witnessed second-order quantifier elimination (WSOQE)}: Given a formula $\Exists{\seq{X}} \varphi$ with first-order~$\varphi$, find first-order predicates
(formulas) $\seq{\witness}$ satisfying this \emph{WSOQE-condition}
\begin{equation*}
  \label{wsoqe-condition}
  \Exists{\seq{X}} \varphi \iff \varphi[\seq{X} \leftarrow \seq{\witness}] \tag{$\ast$}.
\end{equation*}
That is, the goal is to find a first-order instantiation $\seq{\witness}$ (in the same language as $\varphi$) of variables $\seq X$
such that an equivalent first-order formula is obtained.
We call $\seq{\witness}$ a \emph{WSOQE-witness for $\Exists{\seq{X}} \varphi$}.

Recall that the second-order formula $\varPhi$ from the beginning of the introduction is equivalent to $B(a)$.
One can check that $\alpha = \lambda u. u \oeq a$ and $\beta = \lambda u. B(u)$ are WSOQE-witnesses for $\varPhi$.

The contribution of this paper is an algorithm to solve the WSOQE problem and compute such witnesses.
Our algorithm, called \emph{WSCAN}, is an extension of the SCAN algorithm with
a post-processing step which extracts a witness from the derivation
of a terminating SCAN run.
The witness construction proceeds bottom-up, iteratively computing
a witness for derivation step~$i-1$ from a witness at derivation step~$i$.
A prototype has been implemented in version 2.19.0 of the GAPT software package~\cite{GAPT}.
In addition to solving the more general WSOQE problem, the concepts
introduced in this paper provide a new correctness proof for the
SCAN algorithm.

This paper is an extension of \cite{AchammerHetzlSchmidt25}.
We provide a fixpoint construction of witnesses and give a method that produces first-order witnesses for a larger class of derivations than the previous method.
We also outline an approach to integrating theory reasoning, in particular equality reasoning, into our method.

The WSOQE problem bridges the gap between SOQE and another important problem,
namely, \emph{solving formula equations (FEQ)}: Given an input formula $\Exists{\seq{X}} \varphi$, where~$\varphi$ is first-order, find a tuple $\seq{\witness}$ of first-order predicates such that
$$\models \varphi[\seq{X} \leftarrow \seq{\witness}].$$
One can check that the first-order predicates $\alpha = \lambda u. u \oeq a$ and $\beta = \lambda u. B(u)$ are FEQ-solutions for the formula $\varPhi' = \Exists{X}(B(a) \limp (X(a) \land \Forall{u} (X(u) \limp B(u))))$.
Also note that the earlier example $\varPhi = \Exists{X} (X(a) \land \Forall{u}(X(u) \limp B(u)))$ is not valid and therefore does not have an FEQ-solution.

Going back to studies in the 19th century by Schröder in \cite{Schroeder90Vorlesungen}, solving formula equations is one of the oldest and most central problems of logic.
The book~\cite{Rudeanu74Boolean} comprehensively considers the problem in the setting of Boolean algebra.
Solving Boolean equations is closely related to Boolean unification, a subject of thorough study in
computer science, see, e.g.,~\cite{Martin89Boolean} for a survey.
The generalization of this problem from propositional to first-order logic has been made explicit
as early as~\cite{Behmann50Aufloesungsproblem,Behmann51Aufloesungsproblem}.

Today, FEQ is the central problem underlying several areas of computational logic, even though
this is often not made explicit.
For example, solving constrained Horn clauses~\cite{Bjorner15Horn}, a popular formalism in verification,
is a restriction of FEQ, with the existentially quantified variables representing
unknown loop invariants in applications of software verification.
More generally, the problem of inductive theorem proving, e.g., in Peano arithmetic, can be
formulated as a restriction of FEQ, where the predicate variables stand for the unknown
induction formulas.
One can naively solve FEQ for~$\Exists{\seq{X}} \varphi$ by enumerating all first-order predicates $\seq{\witness}$ and checking whether $\varphi[\seq{X} \leftarrow \seq{\witness}]$ is valid by using, e.g., a first-order theorem prover.

An algorithm for WSOQE, like the one we introduce in this paper, can be used to solve both SOQE and FEQ:
A WSOQE-witness $\seq{\witness}$ satisfies \eqref{wsoqe-condition} and therefore
$\varphi[\seq{X} \leftarrow \seq{\witness}]$ solves SOQE for $\Exists{\seq{X}} \varphi$.
However, note that WSOQE is not complete for SOQE as there are formulas that have a SOQE-solution, but no WSOQE-witness (see \Cref{sec.limitations-for-finding-witnesses} for such an example).

Apart from solving SOQE, if one has a WSOQE-witness $\seq{\witness}$ for $\Exists{\seq{X}} \varphi$ one can determine whether it is an FEQ-solution for $\Exists{\seq{X}} \varphi$ by checking whether $\varphi[\seq{X} \leftarrow \seq{\witness}]$ is valid which can be done with a first-order theorem prover.

The use of solutions to WSOQE for solving SOQE or FEQ is very common in the literature:
It is the basis of Ackermann's Lemma~\cite{Ackermann35Untersuchungen}, which provides the foundation for the DLS
algorithm~\cite{Doherty97Computing} for SOQE.
It is central for the decidability of quantifier-free Boolean unification with predicates~\cite{Eberhard17Boolean}
and for the fixed-point theorem for Horn formula equations of~\cite{Hetzl21Fixed,Hetzl21Abstract,Hetzl25Abstract}.
Solutions to WSOQE are called ELIM-witnesses in~\cite{Wernhard17Boolean}.

The paper is organized as follows.
In \Cref{sec.preliminaries} we introduce some necessary notation.
\Cref{sec.scan} describes the SCAN algorithm.
We present our method of extracting witnesses from SCAN derivations in \Cref{sec.resolution-witnesses}.
In \Cref{sec.fixpoint-witnesses} we show that the witnesses produced in \Cref{sec.resolution-witnesses} have a fixpoint representation.
\Cref{sec.first-order-witnesses} provides conditions on SCAN derivations that guarantee first-order witnesses.
\Cref{sec.first-order-background-theories} shows how one can deal with first-order background theories in SCAN.
The prototype implementation of our method is outlined in \Cref{sec.implementation}.
\Cref{sec.discussion} discusses limitations of our method and outlines avenues for future work.

\section{Preliminaries}\label{sec.preliminaries}

We assume familiarity with classical first-order and second-order logic with equality, the usual Tarskian semantics and first-order clause logic.
We fix a language with equality $\oeq$ and countably many
\emph{first-order variables} $u,v,w,\dots$, \emph{predicate variables}~$X, Y, Z, \dots$, \emph{individual constant symbols} $a, b,
  c, \dots$, \emph{function symbols} $f, g, h, \dots$ and \emph{predicate symbols}
$A, B, \dots$ all possibly with subscripts.
Function symbols, predicate symbols and predicate variables have an associated arity.

\emph{Terms} are built from first-order variables, individual constant symbols and function symbols and are denoted by $t, r, s$, possibly with subscripts.
\emph{Atoms} are expressions of the form $A(t_1, \dots, t_k)$ where $A$ is a predicate symbol or predicate variable of arity $k$ and~$t_1, \dots, t_k$ are terms.
\emph{Formulas} are built from atoms, Boolean connectives~\mbox{$\neg$,$\land$,$\lor$,$\limp$} and $\liff$ and quantifiers $\Forall{u}$, $\Exists{u}$, $\Forall{X}$ and $\Exists{X}$ and are denoted by~\mbox{$\varphi$,$\psi$}, possibly with subscripts.
A formula is \emph{first-order} if all its quantifiers range over individual variables.
A formula is \emph{closed} if it does not have any free variables.

We write $\models$ for semantic entailment.
For formulas $\varphi$ and $\psi$ we write $\varphi \imp \psi$ as a shorthand for $\models \varphi \limp \psi$, i.e., the validity of $\varphi \limp \psi$.
We write $\varphi \iff \psi$ as a shorthand for $\models \varphi \liff \psi$, i.e., that $\varphi$ and $\psi$ are logically equivalent.
Note that $\varphi \iff \psi$ if and only if $\varphi \imp \psi$ and $\psi \imp \varphi$.

Let $X$ be a predicate symbol or predicate variable.
An occurrence of $X$ in $\varphi$ is \emph{positive} or \emph{has $+$-polarity} if it is in the scope of an even number of negation signs.
Otherwise it is \emph{negative} or \emph{has $-$-polarity}.
We say \emph{$X$ is positive (negative) in $\varphi$} if all free occurrences of $X$ in $\varphi$ are positive (negative).
In that case we also say \emph{$X$ has~$+$-polarity ($-$-polarity) in $\varphi$}.

A \emph{literal $L$} is a formula either of the form $A(t_1, \dots, t_k)$ or $\neg A(t_1, \dots, t_k)$ where~$A(t_1, \dots, t_k)$ is an atom.
In this case the predicate symbol or predicate variable~$A$ is called the \emph{predicate symbol of $L$} and $t_1, \dots, t_k$ are called the \emph{terms of $L$}.
If~$\seq{X}$ is a tuple of predicate variables then a literal $L$ is called an \emph{$\seq{X}$-literal} if the predicate symbol of $L$ is among the predicate variables in $\seq{X}$.
The dualization of a literal $L$ is denoted by $L^{\perp}$ and defined as $\neg A(t_1, \dots, t_k)$ if~$L = A(t_1, \dots, t_k)$ and $A(t_1, \dots, t_k)$, if~$L = \neg A(t_1, \dots, t_k)$.
If $L$ is an atom we call $L$ a \emph{positive literal}, otherwise a \emph{negative literal}.
We say \emph{$L'$ is an $L$-literal} if $L$ and $L'$ have the same predicate symbol and the same polarity.

A \emph{clause} is a finite set of literals, denoted by $C, C'$, possibly with subscripts.
We identify two clauses, if one results from the other by a renaming of variables.
We often write clauses as a disjunction of its literals, i.e., $L_1 \lor \dots \lor L_n$ and as a disjunction of subclauses, i.e., $C = C_1 \lor \dots \lor C_n$, where $C_1, \dots, C_n$ are subclauses of $C$.
A \emph{clause set} is a (potentially infinite) set of clauses, denoted by $N, N'$, possibly with subscripts.
When a clause set is used in the context of a formula we mean the
formula $\Land_{C \in N} \forall^\ast (\Lor_{L \in C} L)$ where $\forall^\ast$ denotes the universal closure of the formula to the right.
We will sometimes use infinite conjunctions and disjunctions, denoted by $\Land_{i \in I} \varphi_i$ and $\Lor_{i \in I} \varphi_i$ where $I$ is a corresponding infinite index set.

The set of \emph{basic expressions} is the union of the set of terms, formulas, clauses and clause sets.
The set of \emph{expressions} is the smallest set containing the basic expressions that is closed under
$\lambda$-abstraction, i.e., if $E$ is an expression and $u$ is a first-order variable, then $\lambda u. E$ is an expression.

Given expressions $E_1, \dots, E_n$ we write $\seq{E}$ as a shorthand for the tuple $(E_1, \dots, E_n)$.
If~\mbox{$\seq{u} = (u_1, \dots, u_n)$} is a tuple of first-order variables and $E$ an expression, we write~\mbox{$\lambda \seq{u}. E$} for the expression $\lambda u_1. \dots \lambda u_n. E$.
For tuples of terms $\seq{t} = (t_1, \dots, t_n)$ and~\mbox{$\seq{s} = (s_1, \dots, s_n)$} we write $\seq{t} \oeq \seq{s}$  for $t_1 \oeq s_1 \land \dots \land t_n \oeq s_n$ and $\seq{t} \noeq \seq{s}$ for~$t_1 \noeq s_1 \lor \dots \lor t_n \noeq s_n$.

If we consider an expression $E(\seq{u})$,
where $\seq{u}$ is a tuple of first-order variables (or constants),
and $\seq{t}$ is a tuple of terms with the same length as $\seq{u}$,
then $E(\seq{t})$ denotes the simultaneous substitution of the free variables (or constants) $\seq{u}$ in $E$ by $\seq{t}$.
We do the same for tuples of terms, i.e., if $\seq{E}(\seq{u})$ is a tuple of expressions, then $\seq{E}(\seq{t})$ denotes the simultaneous substitution of the free variables (or constants) $\seq{u}$ in all expressions of $\seq{E}$ by $\seq{t}$.
Additionally, using $E(\seq{u})$ implies that all free variables of~$E$ are among those in $\seq{u}$.

A \emph{predicate expression} $\alpha$ is an expression of the form~$\lambda \seq{u}. \varphi$ where $\varphi$ is a formula.
The \emph{arity} of $\alpha$ is the size of the tuple $\seq{u}$.
If $\varphi$ is first-order, we call $\alpha$ a \emph{first-order predicate}.
For predicate expressions $\alpha = \lambda \seq{u}. \varphi(\seq{u})$ and $\beta = \lambda \seq{u}. \psi(\seq{u})$ of the same arity we write $\alpha \imp \beta$ if $\models \Forall{\seq{X}}\Forall{\seq{u}}(\varphi(\seq{u}) \limp \psi(\seq{u}))$ where $\seq{X}$ is a tuple composed of all free predicate variables in $\varphi$ and $\psi$.
We say $\alpha$ and $\beta$ are \emph{equivalent (in symbols $\alpha \iff \beta$)} if~$\alpha \imp \beta$ and $\beta \imp \alpha$.
We also use boolean connectives $\neg, \land, \lor, \limp$ to compose predicate expression of the same arity point-wise, e.g., if $\alpha = \lambda \seq{u}. \varphi(\seq{u})$ and $\beta = \lambda \seq{u}. \psi(\seq{u})$ are first-order predicates then $\alpha \land \beta := \lambda \seq{u}. \alpha(\seq{u}) \land \beta(\seq{u})$.
If $\alpha = \lambda \seq{u}. \varphi(\seq{u})$ is a predicate expression of arity $k$ and $\seq{t}$ is a tuple of $k$ terms, then $\alpha(\seq{t}) := \varphi(\seq{t})$.

Let $X_1, \dots, X_n$ be predicate variables and let $\alpha_1, \dots, \alpha_n$ be predicate expressions such that $X_i$ and $\alpha_i$ have the same arity for all $i \in \mset{1, \dots, n}$.
Then denote by~\mbox{$[X_1 \leftarrow \alpha_1, \dots, X_n \leftarrow \alpha_n]$} (abbreviated as $[\seq{X} \leftarrow \seq{\alpha}]$) the substitution which simultaneously replaces all instances of $X_i$ by the predicate expression $\alpha_i$.
If $E$ is an expression and $\sigma$ is a substitution, then we denote by $E\sigma$ the result of applying the substitution $\sigma$ to $E$.
The identity substitution is denoted by $\mathrm{id}$.
The composition of two substitutions~$\sigma$ and $\tau$ is denoted by concatenating them, i.e., $\sigma\tau$, which denotes the substitution which first applies $\sigma$ and then $\tau$.
The \emph{domain} of a substitution $\sigma$ is the set of individual and predicate variables that $\sigma$ does not map to themselves.
We denote the domain of a substitution $\sigma$ by $\mathrm{dom}(\sigma)$.
We call $\sigma$ a \emph{predicate substitution} if the domain of $\sigma$ only consists of predicate variables and thus the range only consists of predicate expressions.
If $\sigma$ and $\tau$ are predicate substitutions we write $\sigma \imp \tau$ if and only if $\mathrm{dom}(\sigma) = \mathrm{dom}(\tau)$ and for all $X \in \mathrm{dom}(\sigma)$ we have $\sigma(X) \imp \tau(X)$.
We write~$\sigma \iff \tau$ if $\sigma \imp \tau$ and $\tau \imp \sigma$.

Let $\seq{X} = (X_1, \dots, X_n)$ be a tuple of predicate variables and let $Y$ be a predicate variable.
We write $Y \in \seq{X}$ if $Y \in \mset{X_1, \dots, X_n}$.
Furthermore for a set of predicate variables $A$ we write $A \subseteq \seq{X}$ if $Y \in \seq{X}$ for all $Y \in A$.

For terms $\seq{t} = (t_1, \dots, t_n)$ and $\seq{s} = (s_1, \dots, s_n)$ a substitution $\sigma$ is called a \emph{unifier} of $\seq{t}$ and $\seq{s}$, if $t_i\sigma = s_i\sigma$ for all $i \in \mset{1, \dots, n}$.
A unifier $\sigma$ of $\seq{t}$ and $\seq{s}$ is called a \emph{most general unifier} of $\seq{t}$ and $\seq{s}$, if for all unifiers $\tau$ of $\seq{t}$ and $\seq{s}$ there exists a substitution $\rho$ such that $\tau = \sigma\rho$.

Let $\mathcal{M}$ be a structure with domain $M$.
An environment $\theta$ is a function that assigns to every individual variable $u$ an element $m \in M$ and to every $k$-ary predicate variable a relation $R \subseteq M^k$.
Let $\varphi$ be a formula with possibly free individual and predicate variables.
Then we write $\mathcal{M}, \theta \models \varphi$, if $\varphi$ holds in $\mathcal{M}$ when interpreting free individual variables $u$ by $\theta(u)$ and free predicate variables $X$ by $\theta(X)$.
We write $\mathcal{M} \models \varphi$ if~\mbox{$\mathcal{M}, \theta \models \varphi$} for all environments $\theta$.
For a term $t$ we write $t^{\mathcal{M}, \theta}$ for the element in $M$ which results from interpreting $t$ in $\mathcal{M}$ in the environment $\theta$.
We extend this notation to tuple of terms $\seq{t} = (t_1, \dots, t_k)$ by $\seq{t}^{\mathcal{M}, \theta} := (t_1^{\mathcal{M}, \theta}, \dots, t_k^{\mathcal{M}, \theta})$.

\section{The SCAN Algorithm}
\label{sec.scan}

The SCAN algorithm\footnote{SCAN stands for ``Synthesizing Correspondence Axioms for Normal logics''. This stems from the fact that it was originally intended for a particular application and its more general applicability has been realized later on \cite{Gabbay92Quantifier}} takes as input a formula $\Exists{\seq{X}} \varphi$, where $\varphi$ is first-order
and then transforms $\varphi$ into its clausal normal form $N$, which might include a Skolemization step.
Next, $N$ is saturated by a \emph{purification} process:
SCAN picks a clause $C$ and a designated $\seq{X}$-literal $L \in C$.
All non-redundant constraint resolution inferences (defined below) between $C$ (on that literal $L$) and the rest of the clauses are performed.
In addition, constraint factors (defined below) of $C$ are added, and constraint elimination and redundancy criteria are used to simplify the resulting clauses.
If all non-redundant consequences have been generated, $C$ is deleted from the clause set and the process is repeated with a new clause and designated $\seq{X}$-literal.
This process might not terminate, but if we reach a point where no clause contains an $\seq{X}$-literal the result is a clause set~$N^\ast$ such that $\Exists{\seq{X}} N \iff N^\ast$.
SCAN attempts to reverse the Skolemization step from the beginning, if performed, and returns a first-order formula logically equivalent to $\Exists{\seq{X}} \varphi$.
As we present our results for clause sets, we do not deal with reverse Skolemization.

The calculus underlying the SCAN algorithm has three parts: inference steps, a notion of redundancy, and deletion steps that remove clauses that are no longer required.
We devote a subsection to each of these parts.
Then we introduce a notion of derivation for the SCAN algorithm and give some properties of the calculus.

Our description of the calculus underlying SCAN is more detailed than other descriptions in the literature, e.g., \cite{Gabbay92Quantifier,Gabbay08Second}.
We will note at the appropriate points how our description differs.
Most of the additional complexity comes from a more detailed analysis of the purification process, in particular the deletion of a clause after all non-redundant inferences have been performed.
This serves three purposes: 
\begin{enumerate}
\item It lays the groundwork for the construction of witnesses that we present in \Cref{sec.resolution-witnesses}.
\item It separates the redundancy criteria used for the removal of redundant clauses that are already part of the clause set (backward redundancy) from the redundancy criteria used to determine whether a clause generated during the purification process should be kept (forward redundancy).
This is useful as our witness construction depends on the forward redundancy notion, but not so much on the backward redundancy notion.
\item It allows us to give a new correctness proof of SCAN for a specific set of forward redundancy criteria.
\end{enumerate}
Furthermore, our description makes the notion of SCAN derivation explicit, which is important for the witness construction in the next section.

\subsection{Inference Steps}

The following inference steps are used to add clauses to the active clause set.
The goal is to generate sufficiently many first-order consequences from the given second-order clause set such that existing second-order clauses can be deleted in a later deletion step.
The inference calculus resembles an ordinary resolution calculus, with the difference that resolution and factoring introduce constraints between the terms of the resolved literals instead of applying unification.
Unification is introduced as a constraint elimination step.
We collect the inference steps in the inference system $\mathcal{I}$.
\begin{definition}
  The inference system $\mathcal{I}$ is given by the following rules.
\begin{flalign*}
                & \textbf{Constraint resolution:}
                &                                  & \AxiomC{$C \lor L(\seq{t})$}
  \AxiomC{$C' \lor L(\seq{t'})^\perp$}
  \RightLabel{$\mathrm{Res}$}
  \BinaryInfC{$\seq{t} \noeq \seq{t'} \lor C \lor C'$}
  \DisplayProof &                                                                                 \\
  \intertext{where $L$ is a literal and w.l.o.g. $C \lor L(\seq{t})$ and $C' \lor L(\seq{t'})^\perp$ are assumed to be variable-disjoint.
    The clause $\seq{t} \noeq \seq{t'} \lor C \lor C'$ is called the \emph{resolvent} of this inference.}
                & \textbf{Constraint factoring:}
                &                                  & \AxiomC{$C \lor L(\seq{t}) \lor L(\seq{t'})$}
  \RightLabel{$\mathrm{Fac}$}
  \UnaryInfC{$\seq{t} \noeq \seq{t'} \lor C \lor L(\seq{t})$}
  \DisplayProof &                                                                                 \\
  \intertext{where $L$ is a literal.
    The clause $\seq{t} \noeq \seq{t'} \lor C \lor L(\seq{t})$ is called the \emph{factor} of this inference.}
                & \textbf{Constraint elimination:}
                &                                  & \AxiomC{$\seq{t} \noeq \seq{t'} \lor C$}
  \RightLabel{$\mathrm{ConstrElim}$}
  \UnaryInfC{$C\sigma$}
  \DisplayProof &
\end{flalign*}
where $\sigma$ is a most general unifier of $\seq{t}$ and $\seq{t'}$.
\end{definition}

It is important to separate unification from the introduction of constraints since otherwise we would not be able to infer certain first-order consequences.
For example, the clause set $\mset{X(a), \neg X(c)}$ has the first-order consequence $a \noeq c$ (in fact, we have~$\Exists{X} (X(a) \land \neg X(c)) \iff a \noeq c$), but the ordinary resolution calculus cannot infer any new clauses from the clause set, as the constants $a$ and $c$ are not unifiable.

Compared to other descriptions of the SCAN algorithm, we add constraint elimination as an inference rule. This allows the calculus to simulate ordinary resolution and factoring by appending a constraint elimination inference after the corresponding constraint resolution or constraint factoring inference.
This also shows refutational completeness of our inference calculus.
\begin{lemma}
  \label{soundness-of-inferences}
  $\mathcal{I}$ is sound and refutationally complete for first-order logic without equality.
\end{lemma}
\begin{proof}
  To show soundness, consider a model $\mathcal{M}$ and let $L$ be a literal, $C(\seq{u}), C'(\seq{v})$ clauses and let $\seq{t}$, $\seq{t'}$ be tuples of terms.
  If $\mathcal{M} \models C \lor L(\seq{t})$, $\mathcal{M} \models C' \lor L(\seq{t'})^\perp$ and $\mathcal{M} \models \seq{t} \oeq \seq{t'}$, then~$\mathcal{M} \models L(\seq{t})^\perp$ and thus $\mathcal{M} \models C \lor C'$ which shows soundness of $\mathrm{Res}$.

  If $\mathcal{M} \models C \lor L(\seq{t}) \lor L(\seq{t'})$ and $\mathcal{M} \models \seq{t} \oeq \seq{t'}$, then also $\mathcal{M} \models C \lor L(\seq{t})$ which shows soundness of $\mathrm{Fac}$.

  If $\mathcal{M} \models \seq{t} \noeq \seq{t'} \lor C$ and $\sigma$ is any substitution, then $\mathcal{M} \models (\seq{t} \noeq \seq{t'} \lor C)\sigma$.
  Let~$\seq{r} := \seq{t}\sigma$ where $\sigma$ is the most general unifier of $\seq{t}$ and $\seq{t'}$.
  Since $\sigma$ is a unifier of $\seq{t}$ and $\seq{t'}$ we get~\mbox{$(\seq{t} \noeq \seq{t'} \lor C)\sigma = \seq{r} \noeq \seq{r} \lor C\sigma$} and thus $\mathcal{M} \models C\sigma$ which shows soundness of $\mathrm{ConstrElim}$.

  We show refutational completeness by showing that the rules of the usual resolution calculus can be simulated by $\mathcal{I}$.
  The usual resolution rule
  \begin{prooftree}
    \AxiomC{$C \lor L(\seq{t})$}
    \AxiomC{$C' \lor L(\seq{s})^\perp$}
    \BinaryInfC{$(C \lor C')\sigma$}
  \end{prooftree}
  where $\sigma$ is a most general unifier of $\seq{t}$ and $\seq{s}$ is simulated by
  \begin{prooftree}
    \AxiomC{$C \lor L(\seq{t})$}
    \AxiomC{$C' \lor L(\seq{s})^\perp$}
    \RightLabel{$\mathrm{Res}$}
    \BinaryInfC{$\seq{t} \noeq \seq{s} \lor C \lor C'$}
    \RightLabel{$\mathrm{ConstrElim}.$}
    \UnaryInfC{$(C \lor C')\sigma$}
  \end{prooftree}
  Similarly, the usual factoring rule
  \begin{prooftree}
    \AxiomC{$C \lor L(\seq{t}) \lor L(\seq{s})$}
    \UnaryInfC{$(C \lor L(\seq{t}))\sigma$}
  \end{prooftree}
  where $\sigma$ is a most general unifier of $\seq{t}$ and $\seq{s}$ is simulated by
  \begin{prooftree}
    \AxiomC{$C \lor L(\seq{t}) \lor L(\seq{s})$}
    \RightLabel{$\mathrm{Fac}$}
    \UnaryInfC{$\seq{t} \noeq \seq{s} \lor C \lor L(\seq{t})$}
    \RightLabel{$\mathrm{ConstrElim}.$}
    \UnaryInfC{$(C \lor L(\seq{t}))\sigma$}
  \end{prooftree}
\end{proof}

In practice constraint resolution and constraint factoring is often restricted to $\seq{X}$-literals, i.e., those literals whose symbols are to be eliminated.
While this reduces the search space for inferences, this strategy is not refutationally complete, e.g., if $A$ is a nullary predicate constant then $\mset{A, \neg A}$ is not refutable.
Particular search strategies to improve performance are discussed in \Cref{sec.implementation}.

\subsection{Redundancy}
\label{subsec.redundancy}
To remove unnecessary clauses from the active clause set we employ a notion of redundancy.
Some of the simplest examples are the removal of tautologies and the removal of clauses that are subsumed by a clause in the active clause set.

\begin{definition}
  \label{def:redundant}
  Let $C, C'$ be clauses.
  We say \emph{$C$ is a tautology} (or \emph{$C$ is tautological}) if there is a literal $L$ with $L \in C$ and $L^\perp \in C$.
  We say \emph{$C$ subsumes $C'$} (in symbols: $C \subsumes C'$) if there is a substitution $\sigma$ with $C\sigma \subseteq C'$.
  For a clause set $N$ we say \emph{$N$ subsumes $C'$} (in symbols:~$N \subsumes C'$) if there is a clause $C \in N$ with $C \subsumes C'$.
  We say \emph{$C$ is redundant in $N$} if $C$ is a tautology or~$N$ subsumes $C$.
\end{definition}

We use this to introduce a redundancy deletion rule in our calculus:
\begin{flalign*}
  & \text{\smallskip\noindent\emph{\textbf{Backward Redundancy deletion:}}} &&\AxiomC{$N \union \mset{C}$}
\RightLabel{$\mathrm{RedDel}$}
\UnaryInfC{$N$}
\DisplayProof & \\
\intertext{if $C$ is redundant in $N$.}
\end{flalign*}

Note that $\mathrm{RedDel}$ preserves logical equivalence.

We can also simplify clauses by eliminating variable constraints inside clauses.
To that end consider the following rule.
\begin{flalign*}
  & \text{\smallskip\noindent\emph{\textbf{Variable elimination:}}} &&
  \AxiomC{$N \union \mset{v \noeq t \lor C}$}
\RightLabel{$\mathrm{VarElim}$}
\UnaryInfC{$N \union \mset{C[v \leftarrow t]}$}
\DisplayProof & \\
\intertext{where $v$ is a variable which is not a proper subterm of $t$.}
\end{flalign*}
This is a special case of constraint elimination, but preserves logical equivalence, hence it is a simplification, not an inference.

One can add further redundancy elimination rules as long as they preserve logical equivalence between the clause sets.
We refer to the rules given here as \emph{backward} redundancy rules as they act on a clause set that already contains a redundant clause and remove it.
In the next subsection, besides discussing other deletion rules, we give a notion of \emph{forward} redundancy which is used during the purification process to determine whether a generated clause should be kept.

\subsection{Deletion Steps}

For a given clause set $N$ the goal of SCAN is to compute a set $N^\ast$ of first-order clauses such that $\Exists{\seq{X}} N \iff N^\ast$.
We discussed the inference rules that add new clauses to the clause set and a notion of redundancy.
Using the existentially quantified nature of the problem, we can also remove second-order clauses from the clause set under certain conditions.

One deletion rule is \emph{extended purity deletion}.
\begin{flalign*}
                & \text{\smallskip\noindent\emph{\textbf{Extended purity deletion:}}}
                &                                                                     &
                \AxiomC{$N$}
                \RightLabel{$\mathrm{ExtPurDel}_X^{p}$}
                \UnaryInfC{$N \setminus \mset{C \in N \suchthat \text{$C$ contains $X$}}$}
                \DisplayProof
                &                                                                                    \\
  \intertext{for $p\in \mset{+, -}$, where $X$ is a predicate variable and every clause in $N$ containing~$X$ contains an occurrence of $X$ with $p$-polarity.}
\end{flalign*}

The other deletion rule, the \emph{purified clause deletion rule}, is applicable once sufficiently many inferences have been performed with a given clause $P$ \emph{on one of its literals} and the rest of the clause set.
To give a name to this concept we introduce the notion of \emph{pointed clause} (in analogy to \emph{pointed sets} in mathematics):
\begin{definition}
  Let $L$ be a literal and $C$ a clause.
  A \emph{pointed clause} $P$ is an expression $\underline{L} \lor C$, i.e., the literal $L$ is underlined in the clause $L \lor C$.
  We call $L$ the \emph{designated literal of $P$}.
  For a given clause $C$ and a literal $L \in C$ we also write $C[\underline{L}]$ for the pointed clause $\underline{L} \lor (C \setminus \mset{L})$.
  Let $C'$ be a clause.
  The set of \emph{$P$-resolvable literals of $C'$} is the $L^\perp$-literals in $C'$.
  A pointed clause \emph{$Q$ is resolvable with $P$} (or \emph{$P$-resolvable}) if the designated literal of $Q$ is $P$-resolvable.
  If $P = \underline{L(\seq{t})} \lor C$ and $Q = \underline{L(\seq{s})^\perp} \lor C'$, then we set $\mathrm{Res}(P, Q) := \seq{t} \noeq \seq{s} \lor C \lor C'$, i.e., $\mathrm{Res}(P, Q)$ denotes the constraint resolvent of $P$ and $Q$.
\end{definition}

By performing sufficiently many inferences we mean that the clause set $N$ resulting from the deletion of the clause $P$ should still logically imply the resolution closure of~$N$ with respect to $P$. 
To that end we define the resolution closure of a clause set with respect to a pointed clause $P$.
\begin{definition}
  Let $N$ be a clause set and let $P$ be a pointed clause.
  Define $\mathrm{Res}_P^{\leq k}(N)$, i.e., the clauses reachable from clauses in $N$ by resolution with $P$ in at most $k$ steps, inductively by
  \begin{align*}
    \mathrm{Res}_P^{\leq 0}(N) &:= N \\
    \mathrm{Res}_P^{\leq k+1}(N) &:= \mathrm{Res}_P^{\leq k}(N) \union \mset{\mathrm{Res}(P, P') \suchthat \text{$P' \in \mathrm{Res}_P^{\leq k}(N)$ is resolvable with $P$}}.
  \end{align*}
  Furthermore define the resolution closure of $N$ with respect to $P$ by
  $$
  \mathrm{Res}_P^{< \omega}(N) := \Union_{k \in \omega} \mathrm{Res}_P^{\leq k}(N).
  $$
  If $C$ is a clause we define $\mathrm{Res}_P^{\leq k}(C) := \mathrm{Res}_P^{\leq k}(\mset{C})$ and $\mathrm{Res}_P^{< \omega}(C) := \mathrm{Res}_P^{<\omega}(\mset{C})$.
\end{definition}

To avoid having to derive the whole (potentially infinite) resolution closure we incorporate a forward notion of redundancy into the purification process.
This leads us to the notion of a pointed clause $P$ being purified in $N$.

\begin{definition}
  Let $C, C'$ be clauses.
  We write $C \velim C'$, if $C = v \noeq t \lor C_0$ for some clause~$C_0$, some variable $v$ and term $t$ that $v$ is not a proper subterm of and $C' = C_0[v \leftarrow t]$.
  For pointed clauses $P = \underline{L(\seq{s})} \lor C$ and $P' = \underline{L(\seq{s'})} \lor C'$ we write $P \velim P'$ if $C = v \noeq t \lor C_0$ for some clause $C_0$, variable $v$ and term $t$ such that $v$ is not a proper subterm of $t$, $C' = C_0[v \leftarrow t]$ and~$\seq{s'}=\seq{s}[v \leftarrow t]$, i.e., the designated literal is preserved across $\velim$ and the constraint literal is distinct from the designated literal.
  By $\velimtrans$ we denote the reflexive, transitive closure of~$\velim$.

  Let $S, C$ be clauses, let $L$ be a literal and let $\sigma$ be a substitution.
  The substitution $\sigma$ is called \emph{$L$-injective in $C$} if for all $L$-literals $L_1 \neq L_2$ in $C$ we have $L_1\sigma \neq L_2\sigma$.
  We write~$S \subsumesL{L} C$ if there is a substitution $\sigma$ which is $L$-injective in $S$ such that $S\sigma \subseteq C$.
  
  We write $S \subsumesLvelim{L} C$ if $S \subsumesL{L} C'$ for some $C'$ with $C \velimtrans C'$.
  Let $N$ be a clause set and~$C$ a clause.
  Then we write $N \subsumesLvelim{L}  C$ if there is a clause $S \in N$ with $S \subsumesLvelim{L}  C$.
  Let $N'$ be another clause set.
  Then we write $N \subsumesLvelim{L}  N'$ if for all $C \in N'$ we have $N \subsumesLvelim{L}  C$.
  We write $\subsumesLvelimtrans{L}$ for the transitive closure of $\subsumesLvelim{L}$.

  Let $P$ be a pointed clause with designated literal $L$ and let $N$ be a clause set.
  We say \emph{$P$ is purified in $N$} if $N \subsumesLvelim{L^\perp} \mathrm{Res}_P^{\leq 1}(N)$.
\end{definition}

We remark that the $\velim$-relation on pointed clauses is stronger than the~$\velim$-relation on the underlying clauses since the designated literals of the pointed clause also need to be preserved.
For example, we have~\mbox{$u \noeq a \lor X(u) \lor X(b) \velim X(a) \lor X(b)$}.
However, we do not want to have~\mbox{$u \noeq a \lor \underline{X(u)} \lor X(b) \velim X(a) \lor \underline{X(b)}$}, but instead only~\mbox{$u \noeq a \lor \underline{X(u)} \lor X(b) \velim \underline{X(a)} \lor X(b)$}.

Also note that the relation $\subsumesLvelim{L}$ is not transitive, 
but we use the transitive closure later on so we define it here. 
To see that $\subsumesLvelim{L}$ is not transitive 
consider~\mbox{$X(f(u)) \subsumesLvelim{X} v \noeq f(c) \lor X(v)$} 
and $v \noeq f(c) \lor X(v) \subsumesLvelim{X} g(c) \noeq f(c) \lor X(g(c))$.
Then $X(f(u)) \not\subsumesLvelim{X} g(c) \noeq f(c) \lor X(g(c))$ since the constraint $g(c) \noeq f(c)$ can't be removed by variable elimination.

Now we can introduce the purified clause deletion rule.
\begin{flalign*}
                & \text{\smallskip\noindent\emph{\textbf{Purified clause deletion:}}}
                &                                                                     &
                \AxiomC{$N \uplus \mset{P}$}
                \RightLabel{$\mathrm{PurDel}_P$}
                \UnaryInfC{$N$}
                \DisplayProof
                &                                                                                    \\
  \intertext{where $P$ is purified in $N$ and $\uplus$ denotes the disjoint union.}
\end{flalign*}

Note that for purified clause deletion to be applicable it is not necessary to add constraint resolvents of $P$ with itself since we are only interested in $P$ being purified in the remaining clause set $N$.

The condition $P$ being purified in $N$ ensures that the whole resolution closure~\mbox{$\mathrm{Res}_P^{< \omega}(N)$} is semantically entailed by $N$ which we show in this section.

\begin{restatable}{theorem}{purifiedImpliesResolutionClosure}
  \label{purified-implies-resolution-closure}
  Let $N$ be a clause set and let $P$ be a pointed clause such that $P$ is purified in~$N$. Then~\mbox{$N \imp \mathrm{Res}_P^{< \omega}(N)$}.
\end{restatable}
This is an important property for our purposes as it facilitates the witness construction in \Cref{sec.resolution-witnesses}.

Often the side condition of purified clause deletion is given as: All non-redundant inferences between $P$ and the rest of the clause set $N$ have been performed \cite{Gabbay08Second}.
Our description of the purified clause deletion rule differs from this one in two aspects.
\begin{enumerate}
  \item We do not require that all constraint factors of $P$ have been generated.
        While in implementations we still want to generate the constraint factors to retain refutational completeness of the calculus the generation of the constraint factors is not necessary for the purified clause deletion rule to be correct, i.e., to preserve the logical equivalence $\Exists{\seq{X}} (N \union \mset{P}) \iff \Exists{\seq{X}} N$.
  \item We use a specific forward redundancy criterion to determine whether a given constraint resolvent $R$ needs to be added to the clause set, namely we check $N \subsumesLvelim{L^\perp} R$ where $L$ is the designated literal of the pointed clause $P$ to be deleted.
    This differs from the usual side condition in three ways:
    \begin{enumerate}
      \item We take variable elimination into account during subsumption. For example, the constraint resolvent between $P = \underline{X(a)}$ and $Q = \underline{\neg X(v)} \lor B(v) $ is~\mbox{$\mathrm{Res}(P, Q) = a \noeq v \lor B(v)$} which after variable elimination becomes $B(a)$.
      Assume $B(a)$ is already in the clause set.
      Then $B(a)$ does not subsume the constraint resolvent $a \noeq v \lor B(v)$ directly, but it does subsume the logically equivalent variable eliminated clause $B(a)$.
      Thus, taking variable elimination into account during subsumption avoids having to add resolvents to $N$ for which a logically equivalent clause already exists in $N$.
      \item This subsumption relation requires that the subsuming substitution $\sigma$ does not identify any $L^\perp$-literals in the subsuming clauses.
        This is important for our proof of \Cref{purified-implies-resolution-closure} above.
      \item We do not treat tautological clauses as forwardly redundant here.
        Note that the redundancy deletion rule can still be applied to remove tautological clauses, but doing it eagerly could remove tautological clauses that are necessary for $P$ being purified in $N$.
    \end{enumerate}
\end{enumerate}

The following example shows that it is necessary to not treat tautological clauses as forwardly redundant in the purified clause deletion rule as it would otherwise violate \Cref{purified-implies-resolution-closure}.
\begin{example}
  \label{ex:necessary-to-treat-tautological-clause-not-as-redundant}
  Let $P = \underline{\neg X(v)} \lor X(f(v)) \lor B(v)$ and let $N$ be the clause set consisting of the clauses
  \begin{align*}
    (1)\ & X(a) \lor \neg X(f(a)) \\
    (2)\ & X(b) \\
    (3)\ & \neg X(c) \\
    (4)\ & B(b)
  \end{align*}
  The only inferences to be performed between $P$ and $N$ are resolution inferences as no factors can be formed.
  We show that all constraint resolvents between $P$ and $N$ are either tautologies or are subsumed by $N$ after removing variable constraints.
  The resolvent between~$P$ and clause~$1$ is~\mbox{$v\noeq a \lor X(f(v)) \lor B(v) \lor \neg X(f(a))$}.
  After variable elimination this becomes~\mbox{$X(f(a)) \lor B(a) \lor \neg X(f(a))$} which is a tautology.
  The resolvent between $P$ and clause~$2$ is $v \noeq b \lor X(f(v)) \lor B(v)$.
  After variable elimination this becomes $X(f(b)) \lor B(b)$ which is subsumed by clause $4$.
  There are no other possible resolvents.
  In this sense all non-redundant inferences between $P$ and $N$ have been performed.
  
  Removing $P$ from $N \union \mset{P}$ would result in the clause set $N$.
  We now show that there is a model $\mathcal{M}$ and a relation $R$ such that $\mathcal{M} \models N[X \leftarrow R]$, but $\mathcal{M} \not\models \mathrm{Res}_P^{< \omega}(N)[X \leftarrow R]$.
  Let the domain of the model have four elements, i.e., $M = \mset{m_1, m_2, m_3, m_4}$ and consider the following interpretations of $a, b, c, B$
  \begin{align*}
    a^{\mathcal{M}} = m_1, \quad
    b^{\mathcal{M}} = m_2, \quad
    c^{\mathcal{M}} = m_3, \quad
    B^\mathcal{M} = \mset{m_2}
  \end{align*}
  and the following interpretation of $f$.
  \begin{align*}
    f^\mathcal{M}(m_1) = m_4, \quad
    f^\mathcal{M}(m_2) = m_2, \quad
    f^\mathcal{M}(m_3) = m_3, \quad
    f^\mathcal{M}(m_4) = m_3.
  \end{align*}
  Setting $R = \mset{m_1, m_2}$ we get
  $\mathcal{M} \models R(a) \lor \neg R(f(a))$, $\mathcal{M} \models R(b)$, $\mathcal{M} \models \neg R(c)$ and $\mathcal{M} \models B(b)$, i.e., $\mathcal{M} \models N[X \leftarrow R]$.

  Now consider
  \begin{align*}
    C_1 &:= \mathrm{Res}(P, \underline{X(a)} \lor \neg X(f(a))) \\
        &= v \noeq a \lor X(f(v)) \lor B(v) \lor \neg X(f(a)) \in \mathrm{Res}_P^{\leq 1}(N) \\
    \intertext{and}
    C_2 &:= \mathrm{Res}(P, C_1[\underline{X(f(v))}]) \\
        &= v' \noeq f(v) \lor X(f(v')) \lor B(v') \lor v \noeq a \lor B(v) \lor \neg X(f(a)) \in \mathrm{Res}_P^{\leq 2}(N)
  \end{align*}
  which after variable elimination is equivalent to $C_2' := X(f(f(a))) \lor B(f(a)) \lor B(a)$.
  We show~$\mathcal{M} \not\models C_2'[X \leftarrow R]$ which implies $\mathcal{M} \not\models \mathrm{Res}_P^{< \omega}(N)[X \leftarrow R]$.
  Since~\mbox{$f(f(a))^\mathcal{M}= m_3 \not\in R$}, $f(a)^\mathcal{M} = m_4 \not\in B^\mathcal{M}$ and $a^\mathcal{M} = m_1 \not\in B^\mathcal{M}$ we get~\mbox{$\mathcal{M} \not\models R(f(f(a))) \lor B(f(a)) \lor B(a)$}, i.e.,~\mbox{$\mathcal{M} \not\models C_2'[X \leftarrow R]$}.
\end{example}

To be able to test whether a given purified clause deletion step is correct we need to show that given a pointed clause $P$ and a finite clause set $N$ we can decide whether $P$ is purified in $N$.
\begin{lemma}
  We have
  \begin{enumerate}
    \item $\subsumesL{L}$ is decidable.
    \item $\subsumesLvelim{L}$ is decidable.
    \item The relation ``$P$ is purified in $N$'' is decidable for finite clause sets $N$.
  \end{enumerate}
\end{lemma}
\begin{proof}
  \begin{enumerate}
    \item Since $\subsumes$ is decidable we can show that $\subsumesL{L}$ is decidable by reducing it to $\subsumes$.
  
    Consider clauses $C$ and $E$ and a literal $L$ which is an $X$-literal.
    We show how to decide~$C \subsumesL{L} E$.
    We have $C = C_0 \lor \Lor_{i=1}^p L(\seq{s_i})$ and $E = E_0 \lor \Lor_{i=1}^q L(\seq{t_i})$ for some clauses~$C_0$ and $E_0$ that do not contain $L$-literals.
    Now note that for any $\sigma$ which is $L$-injective in $C$ we have that $C\sigma$ also has $p$ many $L$-literals.
    Thus if $p > q$, then $C\sigma \not\subseteq E$, therefore $C \not\subsumesL{L} E$.
    
    So now consider the case $p \leq q$.
    Let $X_1, \dots, X_q$ be fresh predicate symbols with the same arity as $X$.
    For $1 \leq i \leq q$ denote by $L_i(\seq{t})$ the $X_i$-literal with terms $\seq{t}$ which has the same polarity as $L$.
    Furthermore for $f: \mset{1, \dots, p} \to \mset{1, \dots, q}$ injective define~
    \mbox{$C_f :=C_0 \lor \Lor_{i=1}^p L_{f(i)}(\seq{s_i})$} and $E' := E_0 \lor \Lor_{i=1}^q L_{i}(\seq{t_i})$.
    Now we show $C \subsumesL{L} E$ if and only if there exists an injective $f: \mset{1, \dots, p} \to \mset{1, \dots, q}$ such that $C_f \subsumes E'$.
    Then checking $C \subsumesL{L} E$ amounts to checking $C_f \subsumes E'$ for all injective $f: \mset{1, \dots, p} \to \mset{1, \dots, q}$ of which there are only finitely many.
    It remains to show the equivalence.

    $\implies$: We have a $\sigma$ which is $L$-injective in $C$ with $C\sigma \subseteq E$.
    Thus we have $C_0\sigma \subseteq E_0$ and $\seq{s_i}\sigma = \seq{t_i}$.
    Now let $f$ be the identity function.
    Then $C_f = C_0 \lor \Lor_{i=1}^p L_i(\seq{s_i})$ and we have $C_f\sigma = C_0\sigma \lor \Lor_{i=1}^p L_i(\seq{t_i}) \subseteq E'$, i.e., $C_f \subsumes E'$.

    $\impliedby$: Assume there is an injective $f: \mset{1, \dots, p} \to \mset{1, \dots, q}$ such that $C_f \subsumes E'$.
    Then there exists a substitution $\sigma$ such that $C_f\sigma \subseteq E'$.
    Since the $X_1, \dots, X_q$ were fresh we get $C_0\sigma \subseteq E_0$ and for all $1 \leq i \leq p$ we have $\seq{s_i}\sigma = \seq{t_{f(i)}}$.
    We now show that $\sigma$ is $L$-injective in $C$.
    Let $L', L''$ be $L$-literals in $C$ with $L'\sigma = L''\sigma$.
    Then $L' = L(\seq{s_i})$, $L'' = L(\seq{s_j})$ for some $i, j \in \mset{1, \dots, p}$ and $\seq{s_i}\sigma = \seq{s_j}\sigma$, i.e., $\seq{t_{f(i)}} = \seq{t_{f(j)}}$ which implies $f(i) = f(j)$.
    Since~$f$ is injective we get $i = j$ and thus $L' = L''$, i.e., $\sigma$ is $L$-injective in $C$.
    Furthermore we have $C\sigma = C_0\sigma \lor \Lor_{i=1}^p L(\seq{s_i}\sigma) \subseteq E_0 \lor \Lor_{i=1}^p L(\seq{t_{f(i)}}) \subseteq E$
    This means $C \subsumesL{L} E$.
  \item Let $C$ and $E$ be clauses and $L$ a literal.
    Since $\subsumes{L}$ is decidable it suffices to show that for any given clause $E$ there are only finitely many $E'$ with $E \velimtrans E'$.
    In every $\velim$-step the number of literals reduces by at least one since the constraint literal is removed and a substitution is applied to the rest of the clause.
    Thus there are only finitely many $E'$ with $E \velimtrans E'$ and exhaustively checking $C \subsumesLvelim{L} E'$ for all of them gives a decision procedure for $\subsumesLvelim{L}$.
  \item For a given finite $N$ and a pointed clause $P$ we have that $\mathrm{Res}_P^{\leq 1}(N)$ is finite.
    Since $\subsumesLvelim{L^\perp}$ is decidable we exhaustively check $N \subsumesLvelim{L^\perp} R$ for all $R \in \mathrm{Res}_P^{\leq 1}(N)$ and get a decision procedure for $P$ being purified in $N$. \qedhere
  \end{enumerate}
\end{proof}

The goal for the rest of this subsection is to prove \Cref{purified-implies-resolution-closure}.
To prove the theorem and to better understand the introduced relations we lay some groundwork.
The following lemma shows that the introduced relations are all sound and in the case of $\velim$ and $\velimtrans$ even equivalence-preserving.
\begin{lemma}
  \label{subsumption-velim-implies-logical-implication}
  Let $C, C'$ be clauses and $L$ a literal.
  Then 
  \begin{enumerate}
    \item \label{subsumption-velim-implies-logical-implication:velim} $C \velim C'$ implies $C \iff C'$,
    \item \label{subsumption-velim-implies-logical-implication:velim-trans} $C \velimtrans C'$ implies $C \iff C'$,
    \item \label{subsumption-velim-implies-logical-implication:subsumes} $C \subsumesL{L} C'$ implies $C \imp C'$,
    \item \label{subsumption-velim-implies-logical-implication:subsumes-velim} $C \subsumesLvelim{L} C'$ implies $C \imp C'$ and
    \item \label{subsumption-velim-implies-logical-implication:subsumes-velim-trans} $C \subsumesLvelimtrans{L} C'$ implies $C \imp C'$.
  \end{enumerate}
\end{lemma}
\begin{proof}
  \begin{enumerate}
    \item For the left-to-right direction let $\mathcal{M}$ be a model of $C$.
    Then $\mathcal{M} \models C[v \leftarrow t]$. 
    Since~$C[v \leftarrow t] = t \noeq t \lor C'$ and $\mathcal{M} \not\models t \noeq t$ we get $\mathcal{M} \models C'$.
    For the right-to-left direction let $\mathcal{M}$ be a model of $C' = C_0[v \leftarrow t]$ and let $\theta$ be an environment.
    If $\mathcal{M}, \theta \models v \noeq t$, then~$\mathcal{M}, \theta \models C$ and we are done.
    Otherwise $\mathcal{M}, \theta \models v \oeq t$ in which case $\mathcal{M}, \theta \models C_0$.
    Again, $\mathcal{M}, \theta \models C$ and we are done.
    \item Follows by induction on the length of a $\velim$-path between $C$ and $C'$ using \ref{subsumption-velim-implies-logical-implication:velim}.
    \item Let $\mathcal{M}$ be a model of $C$ and let $\sigma$ be $L$-injective in $C$ such that $C\sigma \subseteq C'$.
      Then $\mathcal{M} \models C\sigma$ and since $C\sigma \subseteq C'$ we get $\mathcal{M} \models C'$.
    \item We have $C \subsumesL{L} C''$ for some $C''$ with $C' \velimtrans C''$.
    Now let $\mathcal{M}$ be a model of $C$.
    By \ref{subsumption-velim-implies-logical-implication:subsumes} we get $\mathcal{M} \models C''$ which then by \ref{subsumption-velim-implies-logical-implication:velim-trans} implies $\mathcal{M} \models C'$.
    \item Follows by induction on the length of a $\subsumesLvelim{L}$-path between $C$ and $C'$ using \ref{subsumption-velim-implies-logical-implication:subsumes-velim}. \qedhere
  \end{enumerate}
\end{proof}

Also note the following properties and relationships between the introduced subsumption relations.
\begin{lemma}
  \label{relationships-between-subsumption-relations}
  Let $C, E$ be clauses and $L$ a literal.
  Then
  \begin{enumerate}
    \item $C \subseteq E$ implies $C \subsumesL{L} E$.
    \item $C \subsumesL{L} E$ implies $C \subsumes E$.
    \item $C \subsumesL{L} E$ implies $C \subsumesLvelim{L} E$.
    \item $C \subsumesLvelim{L} E$ implies $C \subsumesLvelimtrans{L} E$.
  \end{enumerate}
\end{lemma}
\begin{proof}
  \begin{enumerate}
    \item We have $C\mathrm{id} \subseteq E$. Since the identity substitution $\mathrm{id}$ is injective it is in particular~$L$-injective, so we get $C \subsumesL{L} E$.
    \item Follows since the substitution $\sigma$ satisfies even stronger conditions in addition to $C\sigma \subseteq E$.
    \item Since $\velimtrans$ is reflexive by definition we get $E \velimtrans E$ so we get $C \subsumesLvelim{L} E$.
    \item Follows since $\subsumesLvelimtrans{L}$ is the reflexive, transitive closure of $\subsumesLvelim{L}$.
  \end{enumerate}
\end{proof}

Note the following result about $\subsumes$ and $\subsumesL{L}$.
It shows that $N \subsumes C$ implies $N' \subsumesL{L} C$ for a clause set $N'$ which results from $N$ by adding appropriate factors.
\begin{lemma}
  Let $C$ and $C'$ be clauses and let $L$ be a literal.
  If $C \subsumes C'$, then there exists a clause~$C_L$ which results from $C$ by a finite number of applications of $\mathrm{Fac}$- and $\mathrm{ConstrElim}$-inferences such that $C_L \subsumesL{L} C'$.
\end{lemma}
\begin{proof}
  Since $C \subsumes C'$ there is a substitution $\sigma$ such that $C\sigma \subseteq C'$.
  We say a pair of $L$-literals~$(L', L'')$ is $\sigma$-unified in $C$ if $L', L'' \in C$, $L' \neq L''$ and $L'\sigma = L''\sigma$.
  Now let $k(C, \sigma)$ be the number of $\sigma$-unified pairs in $C$, i.e., 
  $$
  k(C, \sigma):= \abs{\mset{(L', L'') \suchthat L', L'' \in C, L' \neq L'', L'\sigma = L''\sigma}}.
  $$
  We show that for all $l \in \N$, all clauses $C$ and substitutions $\sigma$ with $C\sigma \subseteq C'$ and $k(C, \sigma) \leq l$ there exists a clause $C_L$ which results from $C$ by a finite number of applications of $\mathrm{Fac}$- and $\mathrm{ConstrElim}$-inferences in $\mathcal{I}$ such that $C_L \subsumesL{L} C'$.

  We do this by induction on $l$.
  If $k(C, \sigma) \leq 0$, then no two distinct literals in $C$ are unified by $\sigma$, i.e., $\sigma$ is $L$-injective in $C$ so
  we have $C_L := C \subsumesL{L} C'$.
  
  Now to the induction step.
  Let $k(C, \sigma) \leq l + 1$.
  If $k(C, \sigma) \leq l$, then the statement follows by the induction hypothesis.
  So let $k(C, \sigma) = l+1$.
  Then there exists a $\sigma$-unified pair in $C$.
  Let $(L(\seq{t}), L(\seq{s}))$ be such a pair.
  Then $\seq{t}\sigma = \seq{s}\sigma$ and $C = L(\seq{t}) \lor L(\seq{s}) \lor C_0$.
  Applying constraint factoring to $C$ results in $C_f := \seq{t} \noeq \seq{s} \lor L(\seq{t}) \lor C_0$.
  Since $\sigma$ unifies $\seq{t}$ and $\seq{s}$ there exists a most general unifier $\tau$.
  Then applying constraint elimination to $C_f$ results in $C_{fc} := L(\seq{t}\tau) \lor C_0\tau$.
  Since $\tau$ is a most general unifier of $\seq{t}$ and $\seq{s}$ and $\sigma$ is a unifier of $\seq{t}$ and $\seq{s}$ there exists a substitution $\rho$ with $\sigma = \tau\rho$.
  Then we get $C_{fc}\rho = L(\seq{t}\tau\rho) \lor C_0\tau\rho = L(\seq{t}\sigma) \lor C_0\sigma = C\sigma \subseteq C'$.

  Now let $(L_1, L_2)$ be a $\rho$-unified pair in $C_{fc}$.
  Then $L_1 = L'\tau$ and $L_2 = L''\tau$ for some~$L', L'' \in C_f$.
  Since $L_1 \neq L_2$ we have $L' \neq L''$.
  Furthermore $(L', L'')$ is a $\sigma$-unified pair in $C_f$ since $L'\sigma = L'\tau\rho = L_1\rho = L_2\rho = L''\tau\rho = L''\sigma$.
  Therefore there is a surjective mapping from the $\sigma$-unified pairs in $C_f$ to the $\rho$-unified pairs in $C_{fc}$, thus $k(C_{fc}, \rho) \leq k(C_f, \sigma)$.

  Furthermore $k(C_f, \sigma) < k(C, \sigma)$ since $C_f$ does not contain the $\sigma$-unified pair $(L(\seq{s}), L(\seq{t}))$.
  Therefore $k(C_{fc}, \rho) < k(C, \sigma)$ and by induction hypothesis there exists a $C_L$ which results from $C_{fc}$ by a finite number of applications of $\mathrm{Fac}$- and $\mathrm{ConstrElim}$-inferences from $C_{fc}$ and which satisfies $C_L \subsumesL{L} C'$.
  Since $C_{fc}$ results from $C$ by a finite number of $\mathrm{Fac}$- and $\mathrm{ConstrElim}$-inferences this means $C_L$ results from $C$ by a finite number of $\mathrm{Fac}$- and $\mathrm{ConstrElim}$-inferences which concludes the proof.
\end{proof}

Next we show that $\velim$ and $\velim^\ast$ are compatible with $\mathrm{Res}(P, \cdot)$ and removing designated literals. 
\begin{lemma}
  \label{velim-is-congruent-wrt-res-p}
  Let $P$ be a pointed clause and let $Q, Q'$ be $P$-resolvable pointed clauses.
  Then
  \begin{enumerate}
    \item \label{velim-is-congruent-wrt-res-p:velim} if $Q \velim Q'$ then $\mathrm{Res}(P, Q) \velim \mathrm{Res}(P, Q')$,
    \item if $Q \velimtrans Q'$ then $\mathrm{Res}(P, Q) \velimtrans \mathrm{Res}(P, Q')$,
    \item \label{velim-is-congruent-wrt-res-p:remove-literal-velim} if $L$ and  $L'$ are the designated literals of $Q$ and $Q'$ respectively, then $Q \velim Q'$ implies $Q \setminus \mset{L} \velim Q' \setminus \mset{L'}$ and
    \item \label{velim-is-congruent-wrt-res-p:remove-literal-velim-trans} if $L$ and  $L'$ are the designated literals of $Q$ and $Q'$ respectively, then $Q \velimtrans Q'$ implies $Q \setminus \mset{L} \velimtrans Q' \setminus \mset{L'}$.
  \end{enumerate}
\end{lemma}
\begin{proof}
  \begin{enumerate}
    \item Let $P = \underline{L(\seq{t})} \lor C$ and let $Q \velim Q'$. Then $Q = \underline{L(\seq{s})^\perp} \lor v \noeq r \lor C'$ and~\mbox{$Q' = \underline{L(\seq{s}[v \leftarrow r])^\perp} \lor~C'[v \leftarrow r]$}.
      Without loss of generality, $P$ and $Q$ as well as $P$ and $Q'$ are variable-disjoint.
      Then 
      $$\mathrm{Res}(P, Q) = \seq{t} \noeq \seq{s} \lor v \noeq r \lor C' \velim \seq{t} \noeq \seq{s}[v \leftarrow r] \lor C \lor C'[v \leftarrow r] = \mathrm{Res}(P, Q').$$
    \item Follows by induction on the length of a $\velim$-path from $Q$ to $Q'$ using \ref{velim-is-congruent-wrt-res-p:velim}.
    \item Since $Q \velim Q'$ we have $Q = \underline{L(\seq{s})} \lor v \noeq t \lor C$ for some term $t$ and a variable $v$ which is not a proper subterm of $t$. Furthermore we have $Q' = \underline{L(\seq{s}[v \leftarrow t])} \lor C[v \leftarrow t]$.
      Now we have $Q \setminus \mset{L(\seq{s})} = v \noeq t \lor C \velim C[v \leftarrow t] = Q' \setminus \mset{L(\seq{s}[v \leftarrow t])}$.
    \item Follows by induction on the length of a $\velim$-path from $Q$ to $Q'$ using \ref{velim-is-congruent-wrt-res-p:remove-literal-velim}. \qedhere
  \end{enumerate}
\end{proof}

The next lemma shows how $\subsumesLvelim{L^\perp}$ is compatible with $\mathrm{Res}(P, \cdot)$.
\begin{lemma}
  \label{res-subsumption-velim-commutation}
  Let $P$ be a pointed clause with designated literal $L$ and let $Q$ be a $P$-resolvable pointed clause with designated $L^\perp$-literal $L'$.
  Furthermore let $S$ be a clause with~$S \subsumesLvelim{L^\perp}~Q$.
  Then either $S \subsumesLvelim{L^\perp} Q \setminus \mset{L'}$ or there exists an $L^\perp$-literal $L'' \in S$ such that $\mathrm{Res}(P, S[\underline{L'}]) \subsumesLvelim{L^\perp} \mathrm{Res}(P, Q)$.
\end{lemma}
\begin{proof}
  Let $P = \underline{L(\seq{t})} \lor C$ and $Q = \underline{L(\seq{s})^\perp} \lor E$ be variable-disjoint.
  With $R := \mathrm{Res}(P, Q)$ we have $R' = \seq{t} \noeq \seq{s} \lor C \lor E$.
  Furthermore there exists a pointed clause $Q' = \underline{L(\seq{s'})^\perp} \lor E'$ with~$Q \velimtrans Q'$ and $S \subsumesL{L^\perp} Q'$.
  We distinguish two cases.

  Case 1: $S \subsumesL{L^\perp} E'$. 
  Note that $E \velimtrans E'$ by \Cref{velim-is-congruent-wrt-res-p}.
  Therefore $S \subsumesLvelim{L^\perp} E = Q \setminus \mset{L'}$.
  
  Case 2: There exists a substitution $\sigma$ which is $L^\perp$-injective in $S$ with $S\sigma \subseteq Q'$ and~$L(\seq{s'})^\perp \in S\sigma$.
  Since $\sigma$ is $L^\perp$-injective we get $S = L(\seq{r})^\perp \lor S_0$ for some terms $\seq{r}$ with~$\seq{r}\sigma = \seq{s'}$ and $S_0\sigma \subseteq E'$.
  Without loss of generality $S$ and $P$ are variable-disjoint.
  Now consider the clause $R_S := \mathrm{Res}(P, S[\underline{L(\seq{r})^\perp}]) = \seq{t} \noeq \seq{r} \lor C \lor S_0$.
  Then $R_S\sigma \subseteq \seq{t} \noeq \seq{s'} \lor C \lor E' := R'$.
  By \Cref{velim-is-congruent-wrt-res-p} we have $R \velimtrans R'$.
  Thus we get $R_S \subsumesLvelim{L^\perp} \mathrm{Res}(P, Q)$ which finishes the proof.
\end{proof}

We now show the same property for $\subsumesLvelimtrans{L}$.
\begin{lemma}
  \label{res-subsumption-velim-trans-commutation}
  Let $P$ be a pointed clause with designated literal $L$ and let $Q$ be a $P$-resolvable pointed clause with designated $L^\perp$-literal $L'$.
  Furthermore let $S$ be a clause with $S \subsumesLvelimtrans{L^\perp} Q$.
  Then either $S \subsumesLvelimtrans{L^\perp} Q \setminus \mset{L'}$ or there exists a literal $L'' \in S$ such that $S[\underline{L''}]$ is $P$-resolvable and $\mathrm{Res}(P, S[\underline{L''}]) \subsumesLvelimtrans{L^\perp} \mathrm{Res}(P, Q)$.
\end{lemma}
\begin{proof}
  Since $S \subsumesLvelimtrans{L^\perp} Q$ there is a $k \in N$ and clauses $S_0, \dots, S_k$ with $S_0 = S$, $S_k = Q$ such that for all $0 \leq i < k$ we have $S_i \subsumesLvelim{L^\perp} S_{i+1}$. 
  We now prove the statement by induction on~$k$.
  
  If $k = 0$, then $S = S_0 = Q$, so the statement follows by picking $L'' \in S$ to be the designated literal of $Q$ and $\mathrm{Res}(P, S[\underline{L''}]) \subsumesLvelimtrans{L^\perp} \mathrm{Res}(P, Q)$ follows by reflexivity of $\subsumesLvelimtrans{L^\perp}$.

  Now to the induction step.
  Let $S_0, \dots, S_k, S_{k+1}$ be clauses with $S_0 = S$, $S_{k+1} = Q$ and for all $0 \leq i < k$ we have $S_i \subsumesLvelim{L^\perp} S_{k+1}$.
  So we have $S \subsumesLvelimtrans{L^\perp} S_k$ and $S_k \subsumesLvelim{L^\perp} Q$.
  By \Cref{res-subsumption-velim-commutation} we either have $S_k \subsumesLvelim{L^\perp} Q \setminus \mset{L'}$ or there exists an $L^\perp$-literal $L'' \in S_k$ such that $\mathrm{Res}(P, S_k[\underline{L''}]) \subsumesLvelim{L^\perp} \mathrm{Res}(P, Q)$.
  
  If $S_k \subsumesLvelim{L^\perp} Q \setminus \mset{L'}$ then we get $S \subsumesLvelimtrans{L^\perp} Q \setminus \mset{L'}$ by transitivity of $\subsumesLvelimtrans{L^\perp}$.
  
  Otherwise there is an $L^\perp$-literal $L'' \in S_k$ such that $\mathrm{Res}(P, S_k[\underline{L''}]) \subsumesLvelim{L^\perp} \mathrm{Res}(P, Q)$.
  By induction hypothesis applied to $S_0, \dots, S_k$ we either have $S \subsumesLvelimtrans{L^\perp} S_k \setminus \mset{L''}$ or there exists an $L^\perp$-literal $L''' \in S$ such that $\mathrm{Res}(P, S[\underline{L'''}]) \subsumesLvelimtrans{L^\perp} \mathrm{Res}(P, S_k[\underline{L''}])$.
  
  If we have $S \subsumesLvelimtrans{L^\perp} S_k \setminus \mset{L''}$, then note that $S_k \setminus \mset{L''} \subseteq \mathrm{Res}(P, S_k[\underline{L''}])$ and therefore we have~\mbox{$S_k \setminus \mset{L''}\subsumesLvelim{L^\perp} \mathrm{Res}(P, S_k[\underline{L''}])$}.
  Thus by transitivity of $\subsumesLvelimtrans{L^\perp}$ we get~\mbox{$S\subsumesLvelimtrans{L^\perp}\mathrm{Res}(P, Q)$}.
  
  Otherwise we have $\mathrm{Res}(P, S[\underline{L'''}]) \subsumesLvelimtrans{L^\perp} \mathrm{Res}(P, S_k[\underline{L''}])$ for some $L^\perp$-literal $L''' \in S$ and we get $\mathrm{Res}(P, S[\underline{L'''}]) \subsumesLvelimtrans{L^\perp} \mathrm{Res}(P, Q)$ by transitivity of $\subsumesLvelimtrans{L^\perp}$.
\end{proof}

The following lemma shows that if $P$ is purified in $N$ then the set of clauses $C$ with $N \subsumesLvelimtrans{L^\perp} C$ is closed under constraint resolution with $P$.
\begin{lemma}
  \label{purified-implies-trans-subsume-velim-is-closed-under-resolution}
  Let $P$ be a pointed clause with designated literal $L$, let $Q$ be a $P$-resolvable pointed clause with designated $L^\perp$-literal $L'$ and let $N$ be a clause set.
  If $P$ is purified in $N$ and \mbox{$N \subsumesLvelimtrans{L^\perp} Q$}, then \mbox{$N \subsumesLvelimtrans{L^\perp} \mathrm{Res}(P, Q)$}.
\end{lemma}
\begin{proof}
  By assumption there is an $S \in N$ such that $S \subsumesLvelimtrans{L^\perp} Q$.
  Then by \Cref{res-subsumption-velim-trans-commutation} we either get $S \subsumesLvelimtrans{L^\perp} Q \setminus \mset{L'} \subsumesLvelim{L^\perp} \mathrm{Res}(P, Q)$ and we are done.
  Otherwise there is an~$L^\perp$-literal~$L'' \in S$ such that $\mathrm{Res}(P, S[\underline{L''}]) \subsumesLvelimtrans{L^\perp} \mathrm{Res}(P, Q)$.
  Since $P$ is purified in $N$ we have~$N \subsumesLvelim{L^\perp} \mathrm{Res}_P^{\leq 1}(N)$.
  Since $\mathrm{Res}(P, S[\underline{L''}]) \in \mathrm{Res}_P^{\leq 1}(N)$ we then get $N \subsumesLvelimtrans{L^\perp} \mathrm{Res}(P, Q)$ by transitivity of $\subsumesLvelimtrans{L^\perp}$.
\end{proof}

Finally, we show that $N$ implies the resolution closure with respect to $P$ if $P$ is purified in $N$.
\begin{proof}[Proof of \Cref{purified-implies-resolution-closure}]
  We first show $N \subsumesLvelimtrans{L^\perp} \mathrm{Res}_P^{\leq k}(N)$ by induction on $k$.
  Then the statement follows by \Cref{subsumption-velim-implies-logical-implication}.
  If $k=0$, then $\mathrm{Res}_P^{\leq k}(N) = N$ so the statement follows by reflexivity of $\subsumesLvelimtrans{L^\perp}$.

  Now to the induction step.
  Let $R \in \mathrm{Res}_P^{\leq k+1}(N)$.
  If $R \in \mathrm{Res}_P^{\leq k}(N)$, then the statement follows by induction hypothesis.
  Otherwise $R = \mathrm{Res}(P, Q)$ for some $Q \in \mathrm{Res}_P^{\leq k}(N)$.
  Then by induction hypothesis we have $N \subsumesLvelimtrans{L^\perp} Q$.
  By \Cref{purified-implies-trans-subsume-velim-is-closed-under-resolution} and since $P$ is purified in $N$ we then get $N\subsumesLvelimtrans{L^\perp}\mathrm{Res}(P, Q)=R$ which concludes the proof.
\end{proof}

\subsection{Derivations}

Let us now define the calculus underlying SCAN and analyze derivations in the calculus.

  \begin{definition}
  We denote by $\mathcal{C}$ the calculus operating on clause sets with the following derivation steps:
  \begin{flalign*}
                  & \text{\smallskip\noindent\emph{\textbf{Inference:}}}
                  &                                                                     & \AxiomC{$N$}
    \RightLabel{$I$}
    \UnaryInfC{$N \union \mset{C}$}
    \DisplayProof &                                                                                    \\
    \intertext{for $I \in \mset{\mathrm{Res},\mathrm{Fac},\mathrm{ConstrElim}}$ if \ $\AxiomC{$\ C_1\cdots\ C_n$}\RightLabel{$I$}\UnaryInfC{$C$}\DisplayProof$
      is an $\mathcal{I}$-inference and $C_1,\dots,C_n\in~N$.}
                  & \text{\medskip\noindent\emph{\textbf{Backward Redundancy deletion:}}}
                  &                                                                     &
    \AxiomC{$N \union \mset{C}$}
    \RightLabel{$\mathrm{RedDel}$}
    \UnaryInfC{$N$}
    \DisplayProof &                                                                                    \\
    \intertext{if $C$ is redundant in $N$.}
    & \text{\medskip\noindent\emph{\textbf{Variable elimination:}}}
                  &                                                                     &
    \AxiomC{$N \union \mset{v \noeq t \lor C}$}
    \RightLabel{$\mathrm{VarElim}$}
    \UnaryInfC{$N \union \mset{C[v \leftarrow t]}$}
    \DisplayProof &                                                                                    \\
    \intertext{where $v$ is not a proper subterm of $t$.}
                  & \text{\smallskip\noindent\emph{\textbf{Extended purity deletion:}}}
                  &                                                                     &
                  \AxiomC{$N$}
                  \RightLabel{$\mathrm{ExtPurDel}_X^{p}$}
                  \UnaryInfC{$N \setminus \mset{C \in N \suchthat \text{$C$ contains $X$}}$}
                  \DisplayProof
                  &                                                                                    \\
    \intertext{for $p\in \mset{+, -}$, where $X$ is a predicate variable and every clause in $N$ containing~$X$ contains~$X$ with polarity $p$.}
                  & \text{\medskip\noindent\emph{\textbf{Purified clause deletion:}}}
                  &                                                                     &
    \AxiomC{$N \uplus \mset{P}$}
    \RightLabel{$\mathrm{PurDel}_P$}
    \UnaryInfC{$N$}
    \DisplayProof
    ,             &
    \intertext{if $P$ is purified in $N$.}
  \end{flalign*}
\end{definition}
All derivation steps are sound.
\begin{lemma}
  \label{derivation-steps-are-sound}
  If $N/N'$ is a $\mathcal{C}$-derivation step, then $N \imp N'$.
\end{lemma}
\begin{proof}
  For the inference steps $\mathrm{Res}, \mathrm{Fac}$ and $\mathrm{ConstrElim}$, this follows from the soundness of $\mathcal{I}$ (\Cref{soundness-of-inferences}).
  For redundancy deletion, extended purity deletion and purified clause deletion it follows from the fact that the conclusion of the derivation step is a subset of its premise.
  For variable elimination this follows from \Cref{subsumption-velim-implies-logical-implication}.
\end{proof}

Furthermore, one can show that $\mathcal{C}$ preserves equivalence with respect to existential quantification which shows that SCAN is correct. In \Cref{sec.resolution-witnesses} we give a proof of this result.
\begin{restatable}{theorem}{calculusIsSoundAndExistentialEquivalencePreserving}
  \label{correctness-of-scan}
  If $N/N'$ is a $\mathcal{C}$-derivation step, then $\Exists{\seq{X}} N \iff \Exists{\seq{X}} N'$.
\end{restatable}

We record the derivation steps performed during the saturation process using the following definition:
\begin{definition}
  Let $D := (S_i)_{1 \leq i \leq m}$ be a finite sequence of $\mathcal{C}$-derivation steps and let~$N,N'$ be clause sets.
  $D$ is a \emph{$\mathcal{C}$-derivation from $N$ with conclusion $N'$} if $N'$ results from $N$ by successively applying the derivation steps $S_1, \dots S_m$.
  We say \emph{$D$ is a $\mathcal{C}$-derivation from $N$}, if~$D$ is a $\mathcal{C}$-derivation from $N$ with conclusion $N'$ for some clause set $N'$.
  We assume the $S_i$ contain enough information so that $N'$ is uniquely determined by $D$ and $N$ so we call $N'$ \emph{the conclusion of $D$ from $N$}.
  For a $\mathcal{C}$-derivation $D$ from $N$ we define the sequence~\mbox{$N(D) :=(N_i(D))_{0 \leq i \leq m}$} of intermediate clause sets in which
  $N_i(D)$ is defined to be the conclusion of the $\mathcal{C}$-derivation $(S_j)_{1 \leq j \leq i}$ of length $i$ from $N$.

  For a tuple of predicate variables $\seq{X}$ we say \emph{$D$ is $\seq{X}$-eliminating from $N$} if the conclusion of $D$ from $N$ contains no predicate variables from $\seq{X}$.
\end{definition}

\begin{example}
  \label{ex.main-example.derivation}
  Consider the clause set $N$ with the clauses
  \begin{align*}
    (1)\ B(a,v) \qquad (2)\ X(a) \qquad (3)\ B(u,v) \lor \neg X(u) \lor X(v) \qquad  (4)\ \neg X(c)
  \end{align*}
  and suppose we want to eliminate $X$.
  We denote clauses by their numbers, e.g.,~$1$ refers to clause $B(a,v)$.
  Pointed clauses are referred to by dot notation, e.g., $3.2$ refers to the pointed clause $B(u,v) \lor \underline{\neg X(u)} \lor X(v)$.
  Applying SCAN, some resolvents from $N$ are
  \begin{align*}
    (5)\ a \noeq u \lor B(u,v) \lor X(v) \quad \text{($2$ with $3$)} \quad \text{and} \quad
    (6)\  & a \noeq c & \text{($2$ with $4$)}
  \end{align*}
  Note that clause~$5$ is redundant because of clause~$1$,  as $1$ subsumes $5$ after variable elimination.

  Let $D_1$ be this derivation:
  $$
    \Axiom$\mset{1,2,3,4}\fCenter$
    \RightLabel{$\mathrm{Res}_{2.1,4.1}$}
    \UnaryInf$\mset{1,2,3,4,6}\fCenter$
    \RightLabel{$\mathrm{PurDel}_{2.1}$}
    \UnaryInf$\mset{1,3,4,6}\fCenter$
    \RightLabel{$\mathrm{ExtPurDel}_X^{-}$}
    \UnaryInf$\mset{1,6}\fCenter$
    \DisplayProof
  $$
  The subscripts in $\mathrm{Res}$ indicate the corresponding resolution premises.
  Note that the pointed clause~$2.1$ is purified in $\mset{1,3,4,6}$ since $5$, the resolvent of $2$ with $4$, is redundant in $\mset{1,3,4,6}$.
  Thus the $\mathrm{PurDel}_{2.1}$ step is applicable to $\mset{1,2,3,4,6}$.
  As $\mset{1,6}$ does not contain $X$, we have that $D$ is an $X$-eliminating derivation.

  There is another $X$-eliminating derivation from $N$, given by
  $$
    D_2 = 
    \Axiom$\mset{1,2,3,4}\fCenter$
    \RightLabel{$\mathrm{PurDel}_{3.2}$}
    \UnaryInf$\mset{1, 2, 4}\fCenter$
    \RightLabel{$\mathrm{Res}_{2.1,4.1}$}
    \UnaryInf$\mset{1,2,4,6}\fCenter$
    \RightLabel{$\mathrm{PurDel}_{2.1}$}
    \UnaryInf$\mset{1,4,6}\fCenter$
    \RightLabel{$\mathrm{ExtPurDel}_X^{-}$}
    \UnaryInf$\mset{1,6}\fCenter$
    \DisplayProof
  $$
  Note that $3.2$ is purified in $\mset{1, 2, 4}$ since the only resolvent between $2$ and $3$ is subsumed by~$1$ as already pointed out, so clause $3$ can be deleted.
  We then derive $a \noeq c$ which also ensures that $2.1$ is purified in $\mset{1, 4, 6}$.
  Finally, extended purity deletion deletes the remaining clauses containing $X$.
\end{example}

\section{Resolution Witnesses}
\label{sec.resolution-witnesses}

In this section we introduce an algorithm that constructs witnesses for a given $\mathcal{C}$-derivation.
We first make our notion of witness precise:

\begin{definition}
  Let $\varphi$ be a first-order formula.
  A predicate substitution $\sigma$ is a \emph{witness of $\Exists{\seq{X}} \varphi$} if $\dom(\sigma) \subseteq \seq{X}$ and $\Exists{\seq{X}} \varphi \iff \varphi\sigma$.
\end{definition}

For all $\sigma$ we have $\varphi\sigma \imp \Exists{\seq{X}} \varphi$, thus it suffices to prove $\Exists{\seq{X}} \varphi \imp \varphi\sigma$ to show that~$\sigma$ is a witness for $\Exists{\seq{X}} \varphi$.

Given a $\mathcal{C}$-derivation, the idea of our algorithm is to construct a witness back-to-front: Let $D=(S_1, \dots, S_m)$ be an $\seq{X}$-eliminating $\mathcal{C}$-derivation from $N$.
Then the final clause set $N_m(D)$ does not contain any variables from $\seq{X}$ and thus any substitution~$\sigma_m$ with $\dom(\sigma_m) \subseteq \seq{X}$ is a witness for $\Exists{\seq{X}} N_m(D)$.

For all $1 < i \leq m$ and a witness $\sigma_i$ for $\Exists{\seq{X}} N_i(D)$ we now construct a witness~$\sigma_{i-1}$ for $\Exists{\seq{X}} N_{i-1}(D)$ by prepending a substitution $\tau_{S_i}$ that depends on the derivation step~$S_i$, i.e., $\sigma_{i-1} = \tau_{S_i} \sigma_i$.
Thus $\sigma_0 = \tau_{S_1} \dots \tau_{S_m}\sigma_m$ is a witness for $\Exists{\seq{X}} N$.
The algorithm is illustrated in the following figure.
$$
  \begin{array}{rcccccccl}
  N = N(D)_0
               & \xrightarrow{S_1}         & N(D)_1
               & \xrightarrow{S_2}         & \quad \dots \quad         &
               N(D)_{m-1} & \xrightarrow{S_m}         & \ N(D)_m                                                                       \\[0.3em]
               \sigma_{0}
               & \xleftarrow{\tau_{S_1}}      & \sigma_{1}
               & \xleftarrow{\tau_{S_2}}      & \quad \dots \quad         & \sigma_{m-1} & \xleftarrow{\tau_{S_m}}
               & \ \sigma_m
               \end{array}
$$

For this to work the substitutions $\tau_{S_i}$ need to have the following property.
\begin{definition}
  Let $S$ be a $\mathcal{C}$-derivation step from $N$ to $N'$.
  A predicate substitution $\tau$ is called \emph{witness-preserving across $S$} if for all witnesses $\sigma$ of $N'$ the substitution $\tau\sigma$ is a witness of~$\Exists{\seq{X}} N$.
  If $N' \imp N\tau$, then we call $\tau$ \emph{witness-transforming} across $S$.
\end{definition}

The substitutions involved in our witness construction are all witness-transforming.
We show that this also makes them witness-preserving.
\begin{lemma}
  \label{witness-preservation-sufficient-condition}
  Let $S$ be a $\mathcal{C}$-derivation step from $N$ to $N'$.
  If $\tau$ is witness-transforming across~$S$, then $\tau$ is witness-preserving across $S$.
\end{lemma}
\begin{proof}
  Let $\sigma$ be a witness of $\Exists{\seq{X}} N'$.
  We need to show that $\tau\sigma$ is a witness of $\Exists{\seq{X}} N$, i.e., that~\mbox{$\Exists{\seq{X}} N \imp N\tau\sigma$}.
  Since all $\mathcal{C}$-derivation steps are sound by \Cref{derivation-steps-are-sound} we get (i)~\mbox{$\Exists{\seq{X}} N \imp \Exists{\seq{X}} N'$}.
  Since $\sigma$ is a witness of $\Exists{\seq{X}} N'$ we get (ii) $\Exists{\seq{X}} N' \imp N'\sigma$.
  By assumption~$\tau$ satisfies $N' \imp N\tau$ so we get (iii) $N'\sigma \imp N\tau\sigma$.
  Taking together (i), (ii) and (iii) we get~$\Exists{\seq{X}} N \imp N\tau\sigma$, i.e., $\tau\sigma$ is a witness of $\Exists{\seq{X}} N$.
\end{proof}

For all $\mathcal{C}$-derivation steps except $\mathrm{PurDel}_P$ we can already define a witness-transforming substitution.
\begin{definition}
  We define  
  \begin{align*}
    \tau_S &:= \mathrm{id} \quad \text{for $S \in \mset{\mathrm{Res}, \mathrm{ConstrElim}, \mathrm{Fac}, \mathrm{RedDel}, \mathrm{VarElim}}$} \\
    \tau_{\mathrm{ExtPurDel}_X^{-}} &:= [X \leftarrow \lambda \seq{u}. \bot] \\
    \tau_{\mathrm{ExtPurDel}_X^{+}} &:= [X \leftarrow \lambda \seq{u}. \top]
  \end{align*} 
\end{definition}

\begin{lemma}
  \label{witness-preservation-for-steps-without-purified-clause-deletion}
  For all $\mathcal{C}$-derivation steps $S$ other than $\mathrm{PurDel}_P$ we have that $\tau_S$ is witness-transforming across $S$.
\end{lemma}
\begin{proof}
  Let $S$ be a derivation step from $N$ to $N'$.

  For $S \in \mset{\mathrm{Res}, \mathrm{Fac}, \mathrm{ConstrElim}}$ this follows from $N \subseteq N'$. For $S = \mathrm{RedDel}$ we have that~\mbox{$N = N' \union \mset{C}$} where $C$ is redundant in $N'$ and thus $N' \imp C$.
  Therefore we get $N' \imp N$.
  For $S = \mathrm{VarElim}$ we have $N \iff N'$ by \Cref{subsumption-velim-implies-logical-implication} and thus $N' \imp N$.

  For $S = \mathrm{ExtPurDel}_X^{-}$ we need to show $N' \imp N[X \leftarrow \lambda \seq{u}. \bot]$.
  Since $N'$ consists of all the clauses in $N$ which do not contain $X$ it suffices to show $N' \imp C[X \leftarrow \lambda \seq{u}. \bot]$ for all clauses~$C \in N$ that do contain $X$.
  By the side condition of $\mathrm{ExtPurDel}_X^{-}$ we also know that all clauses in $N$ containing $X$ contain $X$ negatively.
  Thus $C$ contains $X$ negatively and therefore~$C[X \leftarrow \lambda \seq{u}. \bot] \iff \top$ which shows $N' \imp C[X \leftarrow \lambda \seq{u}. \bot]$.
  For $S = \mathrm{ExtPurDel}_X^{+}$ the argument is analogous to the negative case.
\end{proof}

The hard part is finding a witness-transforming substitution across purified clause deletion.
Recall that the purified clause deletion rule is $N \uplus \mset{P} / N$ under the condition that $P$ is purified in $N$ which is a condition that ensures $N \imp \mathrm{Res}_P^{< \omega}(N)$.
We make use of this side condition by building a witness candidate based on a resolution closure.
\begin{definition}
  Let $P$ be a pointed clause with designated $X$-literal $L(\seq{t})$ where $\seq{t}$ is a $k$-tuple of terms and let $\seq{c}$ be a $k$-tuple of fresh constants.
  We define the \emph{local resolution closure of $P$} by the predicate expression 
  $$\mathrm{\ell Res}_P := \lambda \seq{c}. \Land_{R(\seq{c}, \seq{v}) \in \mathrm{Res}_P^{< \omega}(L(\seq{c})^\perp)} \Forall{\seq{v}} R(\seq{c}, \seq{v}).$$
\end{definition}

Note that the predicate expression $\mathrm{\ell Res}_P$ can be infinite, if the underlying clause set~\mbox{$\mathrm{Res}_P^{< \omega}(L(\seq{u})^\perp)$} is infinite.

\begin{example}
  \label{lres-p-example}
  For $P = \underline{X(a)}$ we have $\mathrm{Res}_P^{< \omega}(\neg X(c)) = \mset{\neg X(c), a \noeq c}$ and therefore we get~\mbox{$\mathrm{\ell Res}_P = \lambda c. \neg X(c) \land a \noeq c$}.

  For $P = \underline{\neg X(u)} \lor B(u,v) \lor X(v)$ we get 
  \begin{align*}
    \mathrm{Res}_P^{< \omega}(X(c)) = \{&X(c), \\
        &u \noeq c \lor B(u,v) \lor X(v), \\
        &u \noeq c \lor B(u,v) \lor v \noeq u' \lor B(u', v') \lor X(v'), \\
        &u \noeq c \lor B(u, v) \lor v \noeq u' \lor B(u', v') \lor v' \noeq u'' \lor B(u'', v'') \lor X(v''), \\
        &\dots \},
  \end{align*}
  i.e., $\mathrm{Res}_P^{< \omega}(X(c))$ is infinite and thus $\mathrm{\ell Res}_P$ is the infinite predicate expression
  \begin{align*}
    \lambda c. &X(c) \\
      \land &\Forall{v}\Forall{v}(u \noeq c \lor B(u,v) \lor X(v)) \\
      \land &\Forall{u}\Forall{v}\Forall{u'}\Forall{v'}(u \noeq c \lor B(u,v) \lor v \noeq u' \lor B(u', v') \lor X(v')) \\
      \land &\Forall{u}\Forall{v}\Forall{u'}\Forall{v'}\Forall{u''}\Forall{v''}(u \noeq c \lor B(u, v) \lor v \noeq u' \lor B(u', v') \lor v' \noeq u'' \lor B(u'', v'') \lor X(v'')) \\
      \land &\dots
  \end{align*}
\end{example}

Based on the sign of the designated literal of $P$ we define a witness-transforming substitution.
\begin{definition}
  \label{def:witness-transforming-purified-clause-deletion}
  Let $P$ be a pointed clause whose designated literal $L$ is an $X$-literal.
  Set $$
  \tau_{\mathrm{PurDel}_P} := \begin{cases}
    [X \leftarrow \mathrm{\ell Res}_P] & \text{if $L$ is negative}, \\
    [X \leftarrow \neg \mathrm{\ell Res}_P] & \text{if $L$ is positive}.
  \end{cases}
  $$
\end{definition}

We now define the witness corresponding to a $\mathcal{C}$-derivation.
\begin{definition}
  For a $\mathcal{C}$-derivation $D = (S_1, \dots, S_m)$ we set $\sigma(D) := \tau_{S_1} \cdots \tau_{S_m}$.
\end{definition}

Before we show that $\tau_{\mathrm{PurDel}_P}$ is witness-preserving across $\mathrm{PurDel}_P$ we demonstrate our method on an example.
\begin{example}
  \label{ex.main-example.first-witness}
  Recall the clause set $N$ from \Cref{ex.main-example.derivation}
  \begin{align*}
    (1)\ B(a,v) \qquad (2)\ X(a) \qquad (3)\ B(u,v) \lor \neg X(u) \lor X(v) \qquad  (4)\ \neg X(c)
  \end{align*}
  and one of its resolvents
  $$
  (6)\ a \noeq c
  $$
  In \Cref{ex.main-example.derivation} we showed that the following derivation is $X$-eliminating from $N:= \mset{1, 2,3, 4}$.
  \begin{align*}
    D &=
    \Axiom$\mset{1,2,3,4}\fCenter$
    \RightLabel{$\mathrm{Res}_{2.1,4.1}$}
    \UnaryInf$\mset{1,2,3,4,6}\fCenter$
    \RightLabel{$\mathrm{PurDel}_{2.1}$}
    \UnaryInf$\mset{1,3,4,6}\fCenter$
    \RightLabel{$\mathrm{ExtPurDel}_X^{-}$}
    \UnaryInf$\mset{1,6}\fCenter$
    \DisplayProof
  \end{align*}
  We now compute the corresponding witness given by $\sigma(D) = \tau_{\mathrm{Res}}\tau_{\mathrm{PurDel}_{2.1}}\tau_{\mathrm{ExtPurDel}_X^{-}}$, where~$\tau_{\mathrm{Res}} = \mathrm{id}$, $\tau_{\mathrm{PurDel}_{2.1}} = [X \leftarrow \neg \mathrm{\ell Res}_{2.1}]$ and $\tau_{\mathrm{ExtPurDel}_X^{-}} = [X \leftarrow \lambda u. \bot]$.
  Note that~\mbox{$\mathrm{\ell Res}_{2.1} \iff \lambda u. \neg X(u) \land u \noeq a$} since~$\mathrm{Res}(\underline{X(a)}, \neg X(u)) = u \noeq a$.
  Therefore 
  \begin{align*}
    \sigma(D) &\iff [X \leftarrow \neg \mathrm{\ell Res}_{2.1}[X \leftarrow \lambda u. \bot]] \\
    &\iff [X \leftarrow \lambda u. \neg (u \noeq a)] \\ 
    &\iff [X \leftarrow \lambda u. u \oeq a]
  \end{align*}
  To certify that $\sigma(D)$ is a witness it suffices to show $\Exists{X} N \imp N\sigma(D)$.
  We show later in \Cref{correctness-of-scan} that all $\mathcal{C}$-derivation steps are equivalence-preserving with respect to existential quantification.
  Assuming this for now we get $\Exists{X} N \iff \mset{1,6}$.
  Thus it suffices to show~\mbox{$\mset{1,6} \imp N\sigma(D)$}.
  Note that $N\sigma(D)$ is equivalent to the formula
  \begin{align*}
    \Forall{v} B(a,v) 
    \quad \land \quad a \oeq a 
    \quad \land \quad \Forall{u}\Forall{v} (B(u,v) \lor u \noeq a \lor v \oeq a) 
    \quad \land \quad c \noeq a.
  \end{align*}
  The first conjunct is implied by $1$.
  The second conjunct is an instance of reflexivity.
  The third conjunct is equivalent to $\Forall{v}(B(a,v) \lor v \oeq a)$ and is thus implied by $1$.
  The third conjunct is implied by $6$.
  In total we get $\mset{1,6} \imp N\sigma(D)$, i.e., $\sigma(D)$ is a witness for $\Exists{X} N$.  
\end{example}

We now turn to showing that $\tau_{\mathrm{PurDel}_P}$ is witness-transforming across $\mathrm{PurDel}_P$, i.e., we need to show that for all clause sets $N$ and pointed clauses $P$ purified in~$N$ we have~\mbox{$N \imp (N \union \mset{P})\tau_{\mathrm{PurDel}_P}$}.
We break this problem into two parts: $N \imp N\tau_{\mathrm{PurDel}_P}$ and $\models P\tau_{\mathrm{PurDel}_P}$.
We now give a series of lemmas necessary to prove these two statements.
The first step is showing how the finite resolution closure iterations $\mathrm{Res}_P^{\leq k}(C)$ of a clause $C$ relate to the finite resolution closure iterations $\mathrm{Res}_P^{\leq k}(L(\seq{c})^\perp)$ of a single~$P$-resolvable literal.
\begin{lemma}
  \label{finite-local-resolution-closure-describes-clause-resolution-closure}
  Let $P$ be a pointed clause, let $L(\seq{t}(\seq{v}))$ be the designated literal of $P$
  and let~\mbox{$C = C_0(\seq{w}) \lor \Lor_{i=1}^q L(\seq{s_i}(\seq{w}))^\perp$} be a clause 
  where $C_0(\seq{w})$ contains no $L^\perp$-literals.
  Furthermore let $\seq{c}$ be a tuple of fresh constants and let $k \in \N$.
  Then $R \in \mathrm{Res}_P^{\leq k}(C)$ if and only if there exist $k_1, \dots, k_q \in \N$ with $\sum_{i=1}^q k_i \leq k$ and $R_i(\seq{c}, \seq{z_i}) \in \mathrm{Res}_P^{\leq k_i}(L(\seq{c})^\perp)$ such that~\mbox{$R = C_0(\seq{w}) \lor \Lor_{i=1}^q R_i(\seq{s_i}(\seq{w}), \seq{z_i})$}.
\end{lemma}
\begin{proof}
  We prove the equivalence by using proof by induction on $k$.
  For $k = 0$ we have~\mbox{$\mathrm{Res}_P^{\leq k}(C) = \mset{C}$}.
  With~\mbox{$k_1 = \dots = k_q = 0$} and $R_i(\seq{c}) = L(\seq{c})^\perp \in \mathrm{Res}_P^{\leq 0}(L(\seq{c})^\perp)$ we get~\mbox{$C = C_0(\seq{w}) \lor \Lor_{i=1}^q R_i(\seq{s_i}(\seq{w}))$}.
  Now we have $R \in \mathrm{Res}_P^{\leq k}(C)$ if and only if $R = C$ which shows the base case of the induction statement.

  For the induction step we show both implications separately.
  For the left-to-right direction let $R \in \mathrm{Res}_P^{\leq k+1}(C)$.
  If $R \in \mathrm{Res}_P^{\leq k}(C)$ the statement follows from the induction hypothesis.
  Otherwise $R = \mathrm{Res}(P, P')$ for a pointed clause $P' \in \mathrm{Res}_P^{\leq k}(C)$.
  Then by induction hypothesis, there are $k_1, \dots k_q \in \N$ with $\sum_{i=1}^q k_i \leq k$ and for $1 \leq i \leq q$ there is an~\mbox{$R_i(\seq{c}, \seq{z_i}) \in \mathrm{Res}_P^{\leq k_i}(L(\seq{c})^\perp)$} such that $P' = C_0(\seq{w}) \lor \Lor_{i=1}^q R_i(\seq{s_i}(\seq{w}), \seq{z_i})$.
  Since $C_0$ does not contain $L^\perp$-literals the designated literal of $P'$, call it $L'$, must be in one of the $R_i(\seq{s_i}(\seq{w}), \seq{z_i})$, say in the $j$-th one.
  Now 
  $$R = \mathrm{Res}(P, P') = C_0(\seq{w}) \lor \Lor_{\substack{i=1 \\ i\neq j}}^q R_i(\seq{s_i}(\seq{w}), \seq{z_i}) \lor \mathrm{Res}(P, R_i(\seq{s_i}(\seq{w}), \seq{z_i})[\underline{L'}])$$
  and note that we have $\mathrm{Res}(P, R_i(\seq{s_i}(\seq{w}), \seq{z_i})[\underline{L'}]) \in \mathrm{Res}_P^{\leq k_j+1}(L(\seq{c})^\perp)$.
  Set $k_i':= k_i$ if~\mbox{$i \neq j$} and~\mbox{$k_j':= k_j+1$}.
  Then we have $\sum_{i=1}^q k_i' = (\sum_{i=1}^q k_i) + 1 \leq k+1$.
  Furthermore set~\mbox{$R_i'(\seq{c}, \seq{z_i}) := R_i(\seq{c}, \seq{z_i})$} if $i \neq j$ and $R_j'(\seq{c}, \seq{z_j}) := \mathrm{Res}(P, R_i(\seq{c}, \seq{z_i})[\underline{L'}]) \in \mathrm{Res}_P^{\leq k_j'}(L(\seq{c})^\perp)$.
  Then for all $1 \leq i \leq q$ we have $R_i' \in \mathrm{Res}_P^{\leq k_i'}(L(\seq{c})^\perp)$ and
  $$
    R = C_0(\seq{w}) \lor \Lor_{\substack{i=1}}^q R_i'(\seq{s_i}(\seq{w}), \seq{z_i}).
  $$

  Now to the right-to-left direction.
  Let $R$ be such that there exist $k_1, \dots, k_q \in \N$ with~\mbox{$\sum_{i=1}^q k_i \leq k+1$} and for $1 \leq i \leq q$ there is an $R_i(\seq{c}, \seq{z_i}) \in \mathrm{Res}_P^{\leq k_i}(L(\seq{c})^\perp)$ with~\mbox{$R = C_0(\seq{w}) \lor \Lor_{i=1}^q R_i(\seq{s_i}(\seq{w}))$}.
  If for all $1 \leq i \leq q$ we have $R_i(\seq{c}, \seq{z_i}) \in \mathrm{Res}_P^{\leq 0}(L(\seq{c})^\perp)$, then $R = C \in \mathrm{Res}_P^{\leq 0}(C)$.
  Otherwise there exists a $1 \leq j \leq q$ such that $k_j \geq 1$ and~\mbox{$R_j(\seq{c}, \seq{z_j}) = \mathrm{Res}(P, P')$} for some pointed clause $P'$ with $P' \in \mathrm{Res}_P^{\leq k_j - 1}(L(\seq{c})^\perp)$.
  By induction hypothesis we then have $$R':=C_0(\seq{w}) \lor \Lor_{\substack{i=1 \\ i\neq j}}^q R_i(\seq{s_i}(\seq{w}), \seq{z_j}) \lor P'(\seq{c}, \seq{z_j}) \in \mathrm{Res}_P^{\leq k}(C).$$
  But then we have $R = \mathrm{Res}(P, R') \in \mathrm{Res}_P^{\leq k+1}(C)$ which concludes the proof.
\end{proof}

Using this we can show a key property of $\mathrm{\ell Res}_P$:
The $P$-resolution closure of a clause $C$ can be constructed by substituting $\mathrm{\ell Res}_P$ for the $P$-resolvable literals in $C$.
\begin{lemma}
  \label{local-resolution-closure-describes-clause-resolution-closure}
  Let $P$ be a pointed clause, let $L(\seq{t}(\seq{v}))$ be the designated literal of $P$ and let~\mbox{$C = C_0(\seq{w}) \lor \Lor_{i=1}^q L(\seq{s_i}(\seq{w}))^\perp$} be a clause where $C_0(\seq{w})$ contains no $L^\perp$-literals.
  Furthermore let $\seq{c}$ be a tuple of fresh constants.
  Then 
  \begin{align*}
  \mathrm{Res}_P^{< \omega}(C) 
  &= \mset{C_0(\seq{w}) \lor \Lor_{i=1}^q R_i(\seq{s_i}(\seq{w}), \seq{z_i}) \suchthat R_1(\seq{c}, \seq{z_1}), \dots, R_q(\seq{c}, \seq{z_q}) \in \mathrm{Res}_P^{< \omega}(L(\seq{c})^\perp)} \\
  &\iff \Forall{w}(C_0(\seq{w}) \lor \Lor_{i=1}^q \mathrm{\ell Res}_P(\seq{s_i}(\seq{w}))).
  \end{align*}
\end{lemma}
\begin{proof}
  Let 
  $$R_P(C) := \mset{C_0(\seq{w}) \lor \Lor_{i=1}^q R_i(\seq{s_i}(\seq{w}), \seq{z_i}) \suchthat R_1(\seq{c}, \seq{z_1}), \dots, R_q(\seq{c}, \seq{z_q}) \in \mathrm{Res}_P^{< \omega}(L(\seq{c})^\perp)}.$$

  We first show $\mathrm{Res}_P^{< \omega}(C) = R_P(C)$.
  For the first inclusion $\mathrm{Res}_P^{< \omega}(C) \subseteq R_P(C)$ consider $R \in \mathrm{Res}_P^{< \omega}(C)$. Then we have $R \in \mathrm{Res}_P^{\leq k}(C)$ for some $k \in \N$.
  Then by \Cref{finite-local-resolution-closure-describes-clause-resolution-closure} there are $k_1, \dots, k_q \in \N$ and~\mbox{$R_i(\seq{c}, \seq{z_i}) \in \mathrm{Res}_P^{\leq k_i}(L(\seq{c})^\perp) \subseteq \mathrm{Res}_P^{< \omega}(L(\seq{c})^\perp)$} such that~\mbox{$R = C_0(\seq{w}) \lor \Lor_{i=1}^q R_i(\seq{s_i}(\seq{w}), \seq{z_i}) \in R_P(C)$}.

  To show the other inclusion inclusion $R_P(C) \subseteq \mathrm{Res}_P^{< \omega}(C)$ let $R \in R_P(C)$.
  Then we have~\mbox{$R = C_0(\seq{w}) \lor \Lor_{i=1}^q R_i(\seq{s_i}(\seq{w}), \seq{z_i})$} for some $R_1(\seq{c}, \seq{z_1}), \dots, R_q(\seq{c}, \seq{z_q}) \in \mathrm{Res}_P^{< \omega}(L(\seq{c})^\perp)$.
  By definition of $\mathrm{\ell Res}_P$ this means for~\mbox{$1 \leq i \leq q$} we have $R_i(\seq{c}, \seq{z_q}) \in \mathrm{Res}_P^{< \omega}(L(\seq{c})^\perp)$, i.e., there are~\mbox{$k_1, \dots, k_q \in \N$} with~\mbox{$R_i(\seq{c}, \seq{z_i}) \in \mathrm{Res}_P^{\leq k_i}(L(\seq{c})^\perp)$}.
  Then by \Cref{finite-local-resolution-closure-describes-clause-resolution-closure} for $k := \sum_{i=1}^q k_i$ we have~\mbox{$R = C_0(\seq{w}) \lor \Lor_{i=1}^q R_i(\seq{s_i}(\seq{w}), \seq{z_i}) \in \mathrm{Res}_P^{k}(C) \subseteq \mathrm{Res}_P^{< \omega}(C)$} which concludes the proof of the first equality.

  It remains to show $R_P(C) \iff C_0(\seq{w}) \lor \Lor_{i=1}^q \mathrm{\ell Res}_P(\seq{s_i}(\seq{w}))$.
  Note that
  \begin{align*}
    R_P(C) &\iff \Land_{R_1(\seq{c}, \seq{z_1}), \dots, R_q(\seq{c}, \seq{z_q}) \in \mathrm{Res}_P^{< \omega}(L(\seq{c})^\perp)}\Forall{\seq{w}}\Forall{\seq{z_1}}\dots\Forall{\seq{z_q}}(C_0(\seq{w}) \lor \Lor_{i=1}^q R_i(\seq{s_i}(\seq{w}), \seq{z_i})). \\
    \intertext{Moving the $\Forall{\seq{z_i}}$ quantifiers inside as far as possible gives}
    &\iff \Land_{R_1(\seq{c}, \seq{z_1}), \dots, R_q(\seq{c}, \seq{z_q}) \in \mathrm{Res}_P^{< \omega}(L(\seq{c})^\perp)}\Forall{\seq{w}}(C_0(\seq{w}) \lor \Lor_{i=1}^q \Forall{\seq{z_i}} R_i(\seq{s_i}(\seq{w}), \seq{z_i})). \\
    \intertext{Moving the $\Forall{\seq{w}}$ quantifiers outside results in}
    &\iff \Forall{\seq{w}}\Land_{R_1(\seq{c}, \seq{z_1}), \dots, R_q(\seq{c}, \seq{z_q}) \in \mathrm{Res}_P^{< \omega}(L(\seq{c})^\perp)}(C_0(\seq{w}) \lor \Lor_{i=1}^q \Forall{\seq{z_i}} R_i(\seq{s_i}(\seq{w}), \seq{z_i})). \\
    \intertext{Splitting the conjunction into $q$ separate conjunctions over each of the $R_i$ and using distributivity of $\land$ and $\lor$ we can move the conjunctions inside to get}
    &\iff \Forall{\seq{w}}\left(C_0(\seq{w}) \lor \Lor_{i=1}^q \Land_{R_i(\seq{c}, \seq{z_i}) \in \mathrm{Res}_P^{< \omega}(L(\seq{c})^\perp)} \Forall{\seq{z_i}} R_i(\seq{s_i}(\seq{w}), \seq{z_i})\right) \\
    \intertext{which means}
    &\iff \Forall{\seq{w}}(C_0(\seq{w}) \lor \Lor_{i=1}^q \mathrm{\ell Res}_P(\seq{s_i}(\seq{w}))).
  \end{align*}
\end{proof}

The following is a monotonicity result for applying predicate substitutions to clauses.
\begin{lemma}
  \label{monotonicity-applying-substitutions}
  Let $L(\seq{u})$ be an $X$-literal, let $\alpha$ be a predicate expression of the same arity as $X$ such that $\alpha \imp \lambda \seq{u}. L(\seq{u})^\perp$ and set
  $$
  \tau := \begin{cases}
    [X \leftarrow \alpha] & \text{if $L$ is negative} \\
    [X \leftarrow \neg \alpha] & \text{if $L$ is positive.}
  \end{cases}
  $$
  Then for all clauses $C$ that do not contain any $L^\perp$-literals we have $C \imp C\tau$.
\end{lemma}
\begin{proof}
  We have $C = C_0 \lor \Lor_{i=1}^q L(\seq{s_i})$ for some $C_0$ which does not contain any $X$-literals.
  Then~\mbox{$C\tau = C_0 \lor \Lor_{i=1}^q L(\seq{s_i})\tau$} and we have $L(\seq{s_i})\tau \iff \neg \alpha(\seq{s_i})$ by definition of $\tau$.
  By contraposition of the assumption $\models \Forall{\seq{u}}(\alpha(\seq{u}) \limp L(\seq{u})^\perp)$ we get $L(\seq{s_i}) \imp \neg \alpha(\seq{s_i})$
  which concludes the proof.
\end{proof}

The next result shows how $\tau_{\mathrm{PurDel}_P}$ relates to the resolution closure of a single clause.
\begin{lemma}
  \label{local-resolution-closure-application-is-implied-by-clause-resolution-clsoure}
  Let $P$ be a pointed clause and $C$ a clause.
  Then $\mathrm{Res}_P^{<\omega}(C) \imp C\tau_{\mathrm{PurDel}_P}$. 
\end{lemma}
\begin{proof}
  Let $L(\seq{t}(\seq{v}))$ be the designated literal of $P$.
  Then $C = C_0(\seq{w}) \lor \Lor_{i=1}^q L(\seq{s_i}(\seq{w}))^\perp$ for a clause $C_0(\seq{w})$ which does not contain $L^\perp$-literals.
  By \Cref{local-resolution-closure-describes-clause-resolution-closure} we get 
  $$\mathrm{Res}_P^{< \omega}(C) \iff \Forall{\seq{w}}(C_0(\seq{w}) \lor \Lor_{i=1}^q \mathrm{\ell Res}_P(\seq{s_i}(\seq{w}))).$$
  Also note that 
  $$C\tau_{\mathrm{PurDel}_P} \iff C_0(\seq{w})\tau_{\mathrm{PurDel}_P} \lor \Lor_{i=1}^q \mathrm{\ell Res}_P(\seq{s_i}(\seq{w}))$$
  by the definition of $\tau_{\mathrm{PurDel}_P}$ after the potential elimination of double negation (in case $L$ is a positive literal).
  Since $C_0(\seq{w})$ contains no $L^\perp$-literals and by definition of~\mbox{$\mathrm{\ell Res}_P$} we have~\mbox{$\mathrm{\ell Res}_P(\seq{u}) \imp \lambda \seq{u}. L(\seq{u})^\perp$} we can apply \Cref{monotonicity-applying-substitutions} for $\alpha = \mathrm{\ell Res}_P$ and get~\mbox{$C_0(\seq{w}) \imp C_0(\seq{w})\tau_{\mathrm{PurDel}_P}$} which concludes the proof.
\end{proof}

This allows us to show $\models P\tau_{\mathrm{PurDel}_P}$.
\begin{lemma}
  \label{lres-satisfies-pointed-clause}
  Let $P$ be a pointed clause.
  Then $\models P\tau_{\mathrm{PurDel}_P}$.
\end{lemma}
\begin{proof}
  Let $P = \underline{L(\seq{t}(\seq{v}))} \lor C_0(\seq{v})$.
  Then we have 
  $$P\tau_{\mathrm{PurDel}_P} \iff \neg \mathrm{\ell Res}_P(\seq{t}(\seq{v})) \lor C_0(\seq{v})\tau_{\mathrm{PurDel}_P}.$$
  Therefore it suffices to show the implication~\mbox{$\mathrm{\ell Res}_P(\seq{t}(\seq{v})) \imp C_0(\seq{v})\tau_{\mathrm{PurDel}_P}$}.
  Remember that~\mbox{$\mathrm{\ell Res}_P(\seq{t}(\seq{v})) \iff \mathrm{Res}_P^{< \omega}(L(\seq{t}(\seq{v}))^\perp)$}.
  Thus
  $$\mathrm{\ell Res}_P(\seq{t}(\seq{v})) \iff \mset{L(\seq{t}(\seq{v}))^\perp} \union \mathrm{Res}_P^{< \omega}(\mathrm{Res}(P, L(\seq{t}(\seq{v}))^\perp)).$$
  Before we apply $\mathrm{Res}$ we have to rename variables so that $P$ and $L(\seq{t}(\seq{v}))^\perp$ are variable-disjoint.
  So let $\seq{v'}$ be some fresh variables.
  Then 
  $\mathrm{Res}(P, L(\seq{t}(\seq{v'}))^\perp) = \seq{t}(\seq{v}) \noeq \seq{t}(\seq{v'}) \lor C_0(\seq{v}).$
  Thus we also get 
  $$\mathrm{Res}_P^{< \omega}(\mathrm{Res}(P, L(\seq{t}(\seq{v'}))^\perp)) = \mset{\seq{t}(\seq{v}) \noeq \seq{t}(\seq{v'}) \lor R \suchthat R \in \mathrm{Res}_P^{< \omega}(C_0(\seq{v}))}$$
  which implies $\mathrm{Res}_P^{< \omega}(\mathrm{Res}(P, L(\seq{t}(\seq{v'}))^\perp)) \iff \mathrm{Res}_P^{< \omega}(C_0(\seq{v}))$ after elimination of~\mbox{$\seq{t}(\seq{v}) \noeq \seq{t}(\seq{v'})$} with the unifier $[\seq{v'} \leftarrow \seq{v}]$ and using the fact that $\seq{v'}$ are fresh variables.
  This gives us~\mbox{$\mathrm{\ell Res}_P(\seq{t}(\seq{v})) \imp \mathrm{Res}_P^{< \omega}(C_0(\seq{v}))$}
  and by \Cref{local-resolution-closure-application-is-implied-by-clause-resolution-clsoure} we get~\mbox{$\mathrm{Res}_P^{< \omega}(C_0(\seq{v})) \imp C_0(\seq{v})\tau_{\mathrm{PurDel}_P}$} which concludes the proof.
\end{proof}

Furthermore we can show $N \imp N\tau_{\mathrm{PurDel}_P}$ exploiting the property of $P$ being purified in $N$.
\begin{lemma}
  \label{N-implies-N-with-lres}
  Let $N$ be a clause set and $P$ a pointed clause such that $P$ is purified in $N$.
  Then~\mbox{$N \imp N\tau_{\mathrm{PurDel}_P}$}.
\end{lemma}
\begin{proof}
  Using \Cref{purified-implies-resolution-closure} we get $N \imp \mathrm{Res}_P^{< \omega}(N)$ so it suffices to show $\mathrm{Res}_P^{< \omega}(N) \imp C\tau_{\mathrm{PurDel}_P}$ for all $C \in N$ which follows from \Cref{local-resolution-closure-application-is-implied-by-clause-resolution-clsoure} and the fact that $\mathrm{Res}_P^{< \omega}(C) \subseteq \mathrm{Res}_P^{< \omega}(N)$.
\end{proof}

This shows that $\tau_{\mathrm{PurDel}_P}$ is witness-transforming.
\begin{lemma}
  \label{local-resolution-closure-is-witness-transforming-across-purified-clause-deletion}
  Let $P$ be a pointed clause.
  Then the substitution $\tau_{\mathrm{PurDel}_P}$ is witness-transforming across~\mbox{$\mathrm{PurDel}_P$}.
\end{lemma}
\begin{proof}
  Let the purified clause deletion step be $N \uplus \mset{P} / N$.
  Then we get $N \imp N\tau_{\mathrm{PurDel}_P}$ by \Cref{N-implies-N-with-lres} and $\models P\tau_{\mathrm{PurDel}_P}$ by \Cref{lres-satisfies-pointed-clause}.
  Thus $N \imp (N \union \mset{P})\tau_{\mathrm{PurDel}_P}$, i.e., $\tau_{\mathrm{PurDel}_P}$ is witness-transforming.
\end{proof}

Finally, this concludes the proof that the constructed substitution $\sigma(D)$ is actually a witness.
\begin{theorem}
  \label{eliminating-derivations-admit-infinite-witness}
  If $D$ is an $\seq{X}$-eliminating derivation from $N$, then $\sigma(D)$ is a witness for $\Exists{\seq{X}} N$.
\end{theorem}
\begin{proof}
  Let $D = (S_1, \dots, S_m)$.
  We proceed by induction on $m$.
  If $m = 0$, we have $\sigma(D) = \mathrm{id}$ and $N$ does not contain any variables from $\seq{X}$.
  Thus $\Exists{\seq{X}} N \imp N \imp N\sigma(D)$, i.e., $\sigma(D)$ is a witness for $\Exists{\seq{X}} N$.

  Now to the induction step.
  Let $D = (S_1, \dots, S_{m+1})$ be an $\seq{X}$-eliminating derivation from~$N$.
  Denote by $N_1$ the conclusion of $S_1$ with premise $N$.
  Then $D_1 = (S_2, \dots, S_{m+1})$ is an~$\seq{X}$-eliminating derivation from $N_1$ and by induction hypothesis $\sigma(D_1)$ is a witness for~$\Exists{\seq{X}} N_1$.
  Now $\sigma(D) = \tau_{S_1}\sigma(D_1)$ and since $\tau_{S_1}$ is witness-transforming by \Cref{witness-preservation-for-steps-without-purified-clause-deletion} and \Cref{local-resolution-closure-is-witness-transforming-across-purified-clause-deletion} it follows from \Cref{witness-preservation-sufficient-condition} it follows that $\sigma(D)$ is a witness for $\Exists{\seq{X}} N$.
\end{proof}

The results from this section allow us to show the correctness of SCAN by using the witness-transforming substitutions $\tau_S$.
\calculusIsSoundAndExistentialEquivalencePreserving
\begin{proof}
  Let $S$ be the $\mathcal{C}$-derivation step from $N$ to $N'$.
  Then $\Exists{\seq{X}} N \imp \Exists{\seq{X}} N'$ follows from soundness of $\mathcal{C}$ (\Cref{derivation-steps-are-sound}).
  For $\Exists{\seq{X}} N' \imp \Exists{\seq{X}} N$ it suffices to show $N' \imp N\tau$ for some substitution $\tau$.
  This follows from \Cref{witness-preservation-for-steps-without-purified-clause-deletion} and \Cref{local-resolution-closure-is-witness-transforming-across-purified-clause-deletion} by picking $\tau = \tau_S$.
\end{proof}

We end this section by giving an example of a derivation where the current method does not produce a finite witness.
\begin{example}
  \label{ex.main-example.second-witness}
  Recall the clause set $N$ from \Cref{ex.main-example.derivation}
  \begin{align*}
    (1)\ B(a,v) \qquad (2)\ X(a) \qquad (3)\ B(u,v) \lor \neg X(u) \lor X(v) \qquad  (4)\ \neg X(c)
  \end{align*}
  and one of its resolvents
  $$
  (6)\ a \noeq c
  $$
  In \Cref{ex.main-example.derivation} we showed that the following derivation is $X$-eliminating from $N:= \mset{1, 2,3, 4}$.
  \begin{align*}
    D_2 = 
    \Axiom$\mset{1,2,3,4}\fCenter$
    \RightLabel{$\mathrm{PurDel}_{3.2}$}
    \UnaryInf$\mset{1, 2, 4}\fCenter$
    \RightLabel{$\mathrm{Res}_{2.1,4.1}$}
    \UnaryInf$\mset{1,2,4,6}\fCenter$
    \RightLabel{$\mathrm{PurDel}_{2.1}$}
    \UnaryInf$\mset{1,4,6}\fCenter$
    \RightLabel{$\mathrm{ExtPurDel}_X^{-}$}
    \UnaryInf$\mset{1,6}\fCenter$
    \DisplayProof
  \end{align*}
  We now compute the corresponding witness as 
  \begin{align*}
    \sigma(D_2) &= \tau_{\mathrm{PurDel}_{3.2}}\tau_{\mathrm{Res}}\tau_{\mathrm{PurDel}_{2.1}}\tau_{\mathrm{ExtPurDel}_X^{-}}.
  \end{align*}
  In \Cref{ex.main-example.first-witness} we already computed $$\sigma(D_1) = \tau_{\mathrm{Res}}\tau_{\mathrm{PurDel}_{2.1}}\tau_{\mathrm{ExtPurDel}_X^{-}} \iff [X \leftarrow \lambda u. u \oeq a].$$
  Thus $\sigma(D_2) \iff [X \leftarrow \mathrm{\ell Res}_{3.2}[X \leftarrow \lambda u. u \oeq a]]$.
  We now compute $\mathrm{\ell Res}_{3.2}(u)$.
  We have~\mbox{$R_0 := X(u) \in \mathrm{\ell Res}_{3.2}(u)$}, then $R_1 := \mathrm{Res}(3.2, R_0[\underline{X(u)}]) \in \mathrm{\ell Res}_{3.2}(u)$ with 
  $$R_1 \iff B(u,v) \lor X(v),$$
  then $R_2 := \mathrm{Res}(3.2, R_1[\underline{X(v)}]) \in \mathrm{\ell Res}_{3.2}(u)$ with 
  $$R_2 \iff B(u, v) \lor B(v, v') \lor X(v').$$
  More generally, with $R_{i+1} := \mathrm{Res}(P, R_i)$ we get
  $$R_n \iff \Lor_{i=1}^n B(v_i, v_{i+1}) \lor X(v_{n+1})$$
  where $v_1 = u$.
  Now we have that $\mathrm{\ell Res}_{3.2}(u)$ is infinite and logically equivalent to the clause set
  $\mset{R_n \suchthat n \in \N}$.
  Furthermore
  $$
  R_n[X \leftarrow \lambda u. u \oeq a] \iff \Lor_{i=1}^n B(v_i, v_{i+1}) \lor v_{n+1} \oeq a.
  $$
  Then $\mathrm{\ell Res}_{3.2}[X \leftarrow \lambda u. u \oeq a]$ is equivalent to the infinite predicate expression
  $$
  \lambda v_1. \Land_{n \in \N} \Forall{v_2}\dots \Forall{v_n} (\Lor_{i=1}^n B(v_i, v_{i+1}) \lor v_{n+1} \oeq a)
  $$
  which shows that the produced witness $\sigma(D_2)$ is infinite.
\end{example}

In the next section (\Cref{sec.fixpoint-witnesses}) we show how to describe $\mathrm{\ell Res}_P$ via a fixpoint and in the section after that (\Cref{sec.first-order-witnesses}) we go further and introduce a first-order witness construction which can associate a first-order witness to a class of $\mathcal{C}$-derivations including $D_2$.

\section{Fixpoint Witnesses}
\label{sec.fixpoint-witnesses}

In this section we show that the witnesses defined in the previous section can be expressed in fixpoint logic.
We only focus on greatest fixpoints as our results are naturally expressed in terms of these.
\begin{definition}
  We define the set of \emph{fixpoint formulas} as the smallest sets such that
  \begin{enumerate}
    \item every first-order formula is a fixpoint formula,
    \item the set of fixpoint formulas is closed under Boolean connectives and first-order quantification and
    \item if $\varphi(\seq{u})$ is a fixpoint formula where $\seq{u}$ is a tuple of $k$ variables, $Y$ is a predicate variable of arity $k$, $Y$ is positive in $\varphi$ and $\seq{t}$ is a tuple of $k$ terms then $(\mathrm{gfp}_{Y,\seq{u}} \varphi(\seq{u}))(\seq{t})$ is a fixpoint formula. The predicate variable $Y$ is considered bound inside $\mathrm{gfp}_Y \alpha$.
  \end{enumerate}
  If $\varphi(\seq{u})$ is a fixpoint formula then we call $\lambda \seq{u}. \varphi(\seq{u})$ a \emph{fixpoint predicate}.
  If $\alpha$ is a fixpoint predicate of arity $k$ and $Y$ is a predicate variable of arity $k$ such that $Y$ is positive in $\alpha$ then we define $\mathrm{gfp}_Y \alpha := \lambda \seq{u}. (\mathrm{gfp}_{Y, \seq{v}} \alpha(\seq{v}))(\seq{u})$.
  A witness $\sigma$ is called a \emph{fixpoint witness} if $\sigma(X)$ is a fixpoint predicate for all $X \in \mathrm{dom}(\sigma)$.
\end{definition}

A well-known result due to Knaster and Tarski is the existence of least and greatest fixpoints for monotone functions.
\begin{theorem}[Knaster-Tarski theorem]
  \label{knaster-tarski-theorem}
  Let $M$ be a set, denote by $\mathcal{P}(M)$ the powerset of $M$ and let $f: \mathcal{P}(M) \to \mathcal{P}(M)$ be monotone, i.e., $f(M_1) \subseteq f(M_2)$ for all $M_1, M_2 \subseteq M$.
  Then there exists a unique set $\mathrm{GFP}(f) \subseteq M$ such that $f(\mathrm{GFP}(f)) = \mathrm{GFP}(f)$ and for all $M' \subseteq M$ with $f(M') = M'$ we have $M' \subseteq \mathrm{GFP}(f)$.
\end{theorem}
\begin{proof}
  See, e.g., \cite{Fritz2002}.
\end{proof}

This allows us to define the semantics of greatest fixpoint formulas.
\begin{definition}
  Let $\mathcal{M}$ be a structure with domain $M$, let $\theta$ be an environment and let~\mbox{$\varphi = (\mathrm{gfp}_{Y, \seq{u}} \psi(\seq{u}))(\seq{t})$} where $\psi(\seq{u})$ is a fixpoint formula, $Y$ is a predicate variable of arity $k$ which is positive in $\psi$.
  We extend the semantics for first-order formulas so we can assume that $\mathcal{M}, \theta \models \psi$ is already defined.
  Let $f_{\varphi}: \mathcal{P}(M^k) \to \mathcal{P}(M^k)$ be the function defined by
  $$
  f_{\varphi}(M') := \mset{\seq{m} \in M^k \suchthat \mathcal{M}, \theta\{Y \leftarrow M'\} \models \psi(\seq{m})}
  $$
  Note that $f_{\varphi}$ is monotone since $Y$ is positive in $\psi$, so by \Cref{knaster-tarski-theorem} the set $\mathrm{GFP}(f)$ exists.
  We write $\mathcal{M}, \theta \models  (\mathrm{gfp}_{Y, \seq{u}} \psi(\seq{u}))(\seq{t})$ if and only if $\seq{t}^{\mathcal{M}, \theta} \in \mathrm{GFP}(f_{\varphi})$.
  We also naturally extend the notations $\varphi_1 \imp \varphi_2$, $\varphi_1 \iff \varphi_2$ to fixpoint formulas and fixpoint predicates.
\end{definition}

We collect a few properties about fixpoint predicates that follow from their semantics.
\begin{proposition}
  \label{fixpoint-properties}
  Let $\alpha$ and $\beta$ be fixpoint predicates and let $Y$ be a predicate variable that is positive in $\alpha$.
  Then
  \begin{enumerate}
    \item $\mathrm{gfp}_Y \alpha \iff \alpha[Y \leftarrow \mathrm{gfp}_Y \alpha]$.
    \item If $\beta \iff \alpha[Y \leftarrow \beta]$, then $\beta \imp \mathrm{gfp}_Y \alpha$.
  \end{enumerate}
\end{proposition}

We now show that $\mathrm{\ell Res}_P$ can be expressed as a greatest fixpoint predicate.
\begin{definition}
  Let $P = \underline{L(\seq{t}(\seq{v}))} \lor C(\seq{v}) \lor \Lor_{i=1}^p L(\seq{t_i}(\seq{v}))^\perp$ be a pointed clause where $L$ is an $X$-literal and $C$ does not contain $L^\perp$-literals.
  Let $Y$ be a fresh predicate variable with the same arity as $X$.
  We set 
  \begin{align*}
    \alpha_{P, Y} &:= \lambda \seq{u}. L(\seq{u})^\perp \land \Forall{\seq{v}}(\seq{u} \noeq \seq{t}(\seq{v}) \lor C(\seq{v}) \lor \Lor_{i=1}^p Y(\seq{t_i}(\seq{v}))).\\
  \end{align*}
\end{definition}

\begin{lemma}
  \label{local-resolution-closure-iff-fixpoint-resolution-closure}
  Let $P$ be a pointed clause where the designated literal is an $X$-literal and let $Y$ be a fresh predicate variable with the same arity as $X$.
  Then $\mathrm{\ell Res}_P \iff \mathrm{gfp}_Y \alpha_{P, Y}$.
\end{lemma}
\begin{proof}
  Let $P = \underline{L(\seq{t}(\seq{v}))} \lor C(\seq{v}) \lor \Lor_{i=1}^p L(\seq{t_i}(\seq{v}))^\perp$ be a pointed clause where $L$ is an $X$-literal and $C$ does not contain $L^\perp$-literals.
  Furthermore let $\seq{c}$ be a tuple of fresh constants. 
  We show both directions of the equivalence.

  $\rimp$: Let $\beta$ be a fixpoint of $\alpha_{P,Y}$, i.e., $\beta \iff \alpha_{P,Y}[Y \leftarrow \beta]$.
  We show $\beta \imp \lambda \seq{c}. \mathrm{Res}_P^{\leq k}(L(\seq{c})^\perp)$ by induction on $k$.
  We have
  \begin{align*}
    \beta &\iff \alpha_{P, Y}[Y \leftarrow \mathrm{gfp_Y}\beta] \\
    &= \lambda \seq{u}. L(\seq{u})^\perp \land \Forall{\seq{v}}(\seq{u} \noeq \seq{t}(\seq{v}) \lor C(\seq{v}) \lor \Lor_{i=1}^p \beta(\seq{t_i}(\seq{v}))).
  \end{align*}
  Thus $\beta \imp \lambda \seq{c}. L(\seq{c})^\perp$ which proves the statement for $k=0$.
  For the induction step we show~\mbox{$\beta \imp \lambda \seq{c}. \Forall{\seq{w}} R(\seq{c}, \seq{w})$} for all $R(\seq{c}, \seq{w}) \in \mathrm{Res}_P^{\leq k+1}(L(\seq{c})^\perp)$.
  If $R(\seq{c}, \seq{w}) \in \mathrm{Res}_P^{\leq k}(L(\seq{c})^\perp)$ the statement follows by the induction hypothesis.
  Otherwise $R(\seq{c}, \seq{w}) \in \mathrm{Res}_P^{\leq k}(\mathrm{Res}(P, L(\seq{c})^\perp))$.
  Note that 
  $$\mathrm{Res}(P, L(\seq{c})^\perp) = \seq{c} \noeq \seq{t}(\seq{v}) \lor C(\seq{v}) \lor \Lor_{i=1}^p L(\seq{t_i}(\seq{v}))^\perp.$$
  Thus by \Cref{finite-local-resolution-closure-describes-clause-resolution-closure} there are $k_1, \dots, k_p \in \N$ with $\sum_{i=1}^p k_i \leq k$ and for all $1 \leq i \leq p$ there is an $R_i(\seq{c}, \seq{z_i}) \in \mathrm{Res}_P^{\leq k_i}(L(\seq{c})^\perp)$ such that $R(\seq{c}, \seq{w}) = \seq{c} \noeq \seq{t}(\seq{v}) \lor C(\seq{v}) \lor \Lor_{i=1}^p R_i(\seq{t_i}(\seq{v}), \seq{z_i})$.
  Since $\sum_{i=1}^p k_i \leq k$ we get $k_i \leq k$ for all $1 \leq i \leq p$ and thus by induction hypothesis we have~\mbox{$\beta \imp \lambda \seq{c}. \Forall{\seq{z_i}} R_i(\seq{c}, \seq{z_i})$} and in particular $\beta(\seq{t_i}(\seq{v})) \imp R_i(\seq{t_i}(\seq{v}), \seq{z_i})$.
  Therefore
  \begin{align*}
    \beta &\imp \lambda \seq{c}. \Forall{\seq{v}}(\seq{c} \noeq \seq{t}(\seq{v}) \lor C(\seq{v}) \lor \Lor_{i=1}^p \beta(\seq{t_i}(\seq{v}))) \\
    &\imp \lambda \seq{c}. \Forall{\seq{v}}(\seq{c} \noeq \seq{t}(\seq{v}) \lor C(\seq{v}) \lor \Lor_{i=1}^p R_i(\seq{t_i}(\seq{v}), \seq{z_i})) \\
    &\iff \lambda \seq{c}. \Forall{\seq{w}} R(\seq{c}, \seq{w}).
  \end{align*}
  This means we have 
  $$\beta \imp \lambda \seq{c}. \mathrm{Res}_P^{< \omega}(L(\seq{c})^\perp) \iff \mathrm{\ell Res}_P.$$
  By \Cref{fixpoint-properties} we get~\mbox{$\mathrm{gfp}_Y \alpha_{P,Y} \iff \alpha_{P,Y}[Y \leftarrow \mathrm{gfp}_Y \alpha_{P,Y}]$} so the statement follows by setting $\beta = \mathrm{gfp}_Y \alpha_{P,Y}$.

  $\imp$: 
  We show that $\mathrm{\ell Res}_P$ is a fixpoint of $\alpha_{P, Y}$ with respect to~$Y$, i.e., we show that ~\mbox{$\mathrm{\ell Res}_P \iff \alpha_{P, Y}[Y \leftarrow \mathrm{\ell Res}_P]$}.
  Applying \Cref{fixpoint-properties} we then get~\mbox{$\mathrm{\ell Res}_P \imp \mathrm{gfp}_Y \alpha_{P, Y}$}.
  Note that
  \begin{align*}
    \mathrm{\ell Res}_P(\seq{c}) &= \mathrm{Res}_P^{< \omega}(L(\seq{c})^\perp) = \mset{L(\seq{c})^\perp} \union \mathrm{Res}_P^{< \omega}(\mathrm{Res}(P, L(\seq{c})^\perp)).\\
    \intertext{So we get}
    \mathrm{\ell Res}_P(\seq{c}) &\iff L(\seq{c})^\perp \land \mathrm{Res}_P^{< \omega}(\seq{c} \noeq \seq{t}(\seq{v}) \lor C(\seq{v}) \lor \Lor_{i=1}^{p} L(\seq{t_i}(\seq{v}))^\perp). \\
    \intertext{Then by \Cref{local-resolution-closure-describes-clause-resolution-closure} applied to $\seq{c} \noeq \seq{t}(\seq{v}) \lor C(\seq{v}) \lor \Lor_{i=1}^{p} L(\seq{t_i}(\seq{v}))^\perp$ we get}
    \mathrm{\ell Res}_P(\seq{c}) &\iff L(\seq{c})^\perp \land \Forall{\seq{v}}(\seq{c} \noeq \seq{t}(\seq{v}) \lor C(\seq{v}) \lor \Lor_{i=1}^{p} \mathrm{\ell Res}_P(\seq{t_i}(\seq{v}))) = \alpha_{P, Y}[Y \leftarrow \mathrm{\ell Res}_P],
  \end{align*}
  i.e., $\mathrm{\ell Res}_P$ is a fixpoint of $\alpha_{P, Y}$.
\end{proof}

We use this to define an alternative witness by replacing all instances of $\mathrm{\ell Res}_P$ by~$\mathrm{gfp}_Y \alpha_{P, Y}$.
\begin{definition}
  Let $P$ be a pointed clause with designated literal $L$ whose predicate symbol is~$X$.
  Also let $Y$ be a fresh predicate variable with the same arity as $X$.
  Set $$
  \tau_{\mathrm{PurDel}_P}^\ast := \begin{cases}
    [X \leftarrow \mathrm{gfp}_Y \alpha_{P, Y}] & \text{if $L$ is negative}, \\
    [X \leftarrow \neg \mathrm{gfp}_Y \alpha_{P, Y}] & \text{if $L$ is positive}.
  \end{cases}
  $$
  Furthermore, for  $\mathcal{C}$-derivation steps $S$ other than $\mathrm{PurDel}_P$ set $\tau_S^\ast := \tau_S$.
  For a $\mathcal{C}$-derivation~\mbox{$D = (S_1, \dots, S_m)$}, we define $\sigma_{\mathrm{fp}}(D) := \tau_{S_1}^\ast \dots\tau_{S_m}^\ast$.
\end{definition}
This shows that we can express the infinite witnesses from the previous section via fixpoints.
\begin{theorem}
  \label{eliminating-derivations-admit-fixpoint-witness}
  If $D$ is an $\seq{X}$-eliminating $\mathcal{C}$-derivation from $N$, then $\sigma_{\mathrm{fp}}(D)$ is a fixpoint witness for $\Exists{\seq{X}} N$.
\end{theorem}
\begin{proof}
  The proof is analogous to the proof of \Cref{eliminating-derivations-admit-infinite-witness}, but every instance of $\tau_{\mathrm{PurDel}_P}$ is replaced by $\tau_{\mathrm{PurDel}_P}^\ast$. \Cref{local-resolution-closure-iff-fixpoint-resolution-closure} ensures $\tau_{\mathrm{PurDel}_P} \iff \tau_{\mathrm{PurDel}_P}^\ast$ and thus $\sigma(D) \iff \sigma_{\mathrm{fp}}(D)$.
\end{proof}

We show how this construction produces a fixpoint witness for derivation $D_2$ from \Cref{ex.main-example.second-witness} where the previous construction only found an infinite witness.
\begin{example}
  \label{ex.main-example.fixpoint-witness}
  Recall the clause set $N$ from \Cref{ex.main-example.second-witness}
  \begin{align*}
    (1)\ B(a,v) \qquad (2)\ X(a) \qquad (3)\ B(u,v) \lor \neg X(u) \lor X(v) \qquad  (4)\ \neg X(c)
  \end{align*}
  and one of its resolvents
  $$
  (6)\ a \noeq c
  $$
  and the $X$-eliminating derivation $D_2$ from $N:= \mset{1, 2,3, 4}$ with.
  \begin{align*}
    D_2 = 
    \Axiom$\mset{1,2,3,4}\fCenter$
    \RightLabel{$\mathrm{PurDel}_{3.2}$}
    \UnaryInf$\mset{1, 2, 4}\fCenter$
    \RightLabel{$\mathrm{Res}_{2.1,4.1}$}
    \UnaryInf$\mset{1,2,4,6}\fCenter$
    \RightLabel{$\mathrm{PurDel}_{2.1}$}
    \UnaryInf$\mset{1,4,6}\fCenter$
    \RightLabel{$\mathrm{ExtPurDel}_X^{-}$}
    \UnaryInf$\mset{1,6}\fCenter$
    \DisplayProof
  \end{align*}
  We now compute the corresponding fixpoint witness as 
  \begin{align*}
    \sigma_{\mathrm{fp}}(D_2) &= \tau_{\mathrm{PurDel}_{3.2}}^\ast\tau_{\mathrm{Res}}^\ast\tau_{\mathrm{PurDel}_{2.1}}^\ast\tau_{\mathrm{ExtPurDel}_X^{-}}^\ast \\
    &= \tau_{\mathrm{PurDel}_{3.2}}^\ast\tau_{\mathrm{PurDel}_{2.1}}^\ast[X \leftarrow \lambda u. \bot] \\
    &= [X \leftarrow \mathrm{gfp}_Y \alpha_{3.2, Y}][X \leftarrow \neg \mathrm{gfp}_Y \alpha_{2.1, Y}][X \leftarrow \lambda u. \bot].
  \end{align*}
  Note that clause $2$ only has a positive $X$-literal so we get $\alpha_{2.1, Y} = \lambda u. \neg X(u) \land u \noeq a$ which does not contain $Y$ at all.
  Therefore by \Cref{fixpoint-properties} we get 
  $$\mathrm{gfp}_Y \alpha_{2.1, Y} \iff \alpha_{2.1, Y} = \lambda u. \neg X(u) \land u \noeq a$$
  and $\neg \mathrm{gfp}_Y \alpha_{2.1, Y} \iff \lambda u. X(u) \lor u \oeq a$.
  Furthermore we have 
  \begin{align*}
    \alpha_{3.2, Y} 
  &= \lambda u. X(u) \land \Forall{u'}\Forall{v'}(u \noeq u' \lor B(u,v) \lor Y(v'))  \\
  &\iff \lambda u. X(u) \land \Forall{v}(B(u,v) \lor Y(v)).
  \end{align*}
  In total we get 
  $$\sigma_{\mathrm{fp}}(D_2) = \mathrm{gfp}_Y(\lambda u. u \oeq a \land \Forall{v}(B(u, v) \lor Y(v))).$$
  
  Let $\alpha := \lambda u. u \oeq a \land \Forall{v}(B(u, v) \lor Y(v))$.
  Let us take a look at the greatest fixpoint iterations $\mathrm{gfp}_Y^n\alpha$ defined inductively by
  $\mathrm{gfp}_Y^0\alpha := \lambda u. \top$ and $\mathrm{gfp}_Y^{n+1}\alpha := \alpha[Y \leftarrow \mathrm{gfp}_Y^n(\alpha)].$
  Computing these we get
  \begin{align*}
    \mathrm{gfp}_Y^0\alpha &= \lambda u. \top \\
    \mathrm{gfp}_Y^1 \alpha &\iff \lambda u. u \oeq a \\
    \mathrm{gfp}_Y^2 \alpha &\iff \lambda u. u \oeq a \land \Forall{v}(B(u,v) \lor v \oeq a) \\
    \mathrm{gfp}_Y^3 \alpha &\iff \lambda u. u \oeq a \land \Forall{v_2}(B(u,v) \lor v_2 \oeq a) \land \Forall{v_2}\Forall{v_3}(B(u, v_2) \lor B(v_2, v_3) \lor v_3 \oeq a) \\
    &\vdots \\
    \mathrm{gfp}_Y^n \alpha &\iff \lambda v_1. \Land_{i < n} \Forall{v_2}\dots\Forall{v_i}(\Lor_{j< i} B(v_j, v_{j+1}) \lor v_{i+1} \oeq a)
  \end{align*}
  which are approximations of the witness $\sigma(D_2)$ computed in \Cref{ex.main-example.second-witness}.
\end{example}

\section{First-Order Witnesses}
\label{sec.first-order-witnesses}

In this section we introduce a witness construction method that can produce first-order witnesses.
\begin{definition}
  A witness $\sigma$ is called \emph{first-order}, if its range consists only of first-order predicates.
\end{definition}
For all $\mathcal{C}$-derivation steps except $\mathrm{PurDel}_P$, $\tau_S$ is first-order. Thus, it remains to consider purified clause deletion: we have a clause set $N$ and a pointed clause $P$ which is purified in $N$, and we need to find a predicate substitution $\tau$ such that~\mbox{$N \imp (N \union \mset{P})\tau$} and the range of $\tau$ only consists of first-order predicates.
There is a class of pointed clauses $P$, introduced in \cite{AchammerHetzlSchmidt25}, for which the method from \Cref{sec.resolution-witnesses} already provides first-order witnesses.
\begin{definition}
  Let $P$ be a pointed clause with designated literal $L$ whose predicate symbol is~$X$.
  Then $P$ is called \emph{one-sided} if all $X$-literals in $P$ have the same polarity as $L$.
\end{definition}

\begin{lemma}
  Let $P$ be a one-sided pointed clause.
  Then the range of $\tau_{\mathrm{PurDel}_P}$ only contains first-order predicates.
\end{lemma}
\begin{proof}
  We have $P = \underline{L(\seq{t})} \lor C$ for some literal $L(\seq{t})$ and clause $C$.
  Since $P$ is one-sided, $C$ contains no $L^\perp$-literals so we get~\mbox{$\mathrm{\ell Res}_P(\seq{c}) = \mset{L(\seq{c})^\perp, \seq{c} \noeq \seq{t} \lor C}$}, i.e., $\mathrm{\ell Res}_P(\seq{c})$ is finite which means that the range of~$\tau_{\mathrm{PurDel}_P}$ only consists of first-order predicates.
\end{proof}

In this section we define a new class for which we can produce first-order witnesses that encompasses one-sided pointed clauses. 
It is defined based on a graph that is induced by $N$ and $P$.
To define this graph we lay some groundwork.
\begin{definition}
  Let $C$ be a clause and let $L$ be a literal.
  Then denote by $C_L$ the subclause of~$C$ that consists exactly of the $L$-literals of $C$, i.e., $C_L := \mset{L' \in C \suchthat \text{$L'$ is an $L$-literal}}$.
\end{definition}

The following definition introduces a concept that is necessary for the definition of the graph.
\begin{definition}
  Let $N$ be a clause set, let $P$ be a pointed clause and let $L$ be the designated literal of $P$.
  Denote by $R_P(N)$ the set of $P$-resolvable pointed clauses in $N$, i.e.,~\mbox{$R_P(N) := \mset{C[\underline{L'}] \suchthat \text{$C \in N, L' \in C_{L^\perp}$}}$}.
  Let $s: R_P(N) \to N$ be a function.
  We call $s$ a \emph{$P$-purification subsumption for $N$} if for all $P' \in R_P(N)$ we have $s(P') \subsumesLvelim{L^\perp} \mathrm{Res}(P, P')$.
\end{definition}

Purification subsumptions need not be unique for given $N$ and $P$ as the following example shows.
\begin{example}
  \label{ex.purification-subsumption}
  Let $P = \underline{\neg X(v_1, v_2)} \lor X(v_2, v_1) \lor A(v_1, v_2)$ and 
  $$N = \mset{X(a,b), \quad X(b,a), \quad A(b,a)}.$$
  Then we have 
  $$R_P(N) = \mset{\underline{X(a,b)}, \quad \underline{X(b,a)}}.$$
  Note that
  \begin{align*}
    \mathrm{Res}(P, \underline{X(a,b)}) &= v_1 \noeq a \lor v_2 \noeq b \lor X(v_2, v_1) \lor A(v_1, v_2) \velimtrans X(b, a) \lor A(a,b) \text{ and} \\
    \mathrm{Res}(P, \underline{X(b,a)}) &= v_1 \noeq b \lor v_2 \noeq a \lor X(v_2, v_1) \lor A(v_1, v_2) \velimtrans X(a,b) \lor A(b,a)
  \end{align*}
  Then there are two $P$-purification subsumptions for $N$, namely, $s_1$ given by
  \begin{align*}
    s_1(\underline{X(a,b)}) &= X(b,a) \subsumesLvelim{X} \mathrm{Res}(P, \underline{X(a,b)}) \\
    s_1(\underline{X(b,a)}) &= A(b,a) \subsumesLvelim{X} \mathrm{Res}(P, \underline{X(b,a)}) \\
  \end{align*}
  and $s_2$ given by
  \begin{align*}
    s_2(\underline{X(a,b)}) &= X(b,a) \subsumesLvelim{X} \mathrm{Res}(P, \underline{X(a,b)})\\
    s_2(\underline{X(b,a)}) &= X(a,b) \subsumesLvelim{X} \mathrm{Res}(P, \underline{X(b,a)}). \\
  \end{align*}
\end{example}
The existence of $P$-purification subsumptions for $N$ is equivalent to $P$ being purified in $N$.
\begin{lemma}
  \label{purified-iff-exists-purification-subsumption}
  Let $N$ be a clause set and $P$ a pointed clause.
  Then $P$ is purified in $N$ if and only if there exists a $P$-purification subsumption for $N$.
\end{lemma}
\begin{proof}
  Let $P$ have designated literal $L$.
  
  $\implies$:
  Since $P$ is purified in $N$ we have~\mbox{$N \subsumesLvelim{L^\perp} \mathrm{Res}_P^{\leq 1}(N)$}.
  Thus for every clause~\mbox{$R \in \mathrm{Res}_P^{\leq 1}(N)$} there exists an $S \in N$ such that $S \subsumesLvelim{L^\perp} R$.
  Now for all~\mbox{$P' \in R_P(N)$} pick an $S_{P'} \in N$ such that~\mbox{$S_{P'} \subsumesLvelim{L^\perp}\mathrm{Res}(P, P')$}.
  Then setting $s(P') := S_{P'}$ for all~\mbox{$P' \in R_P(N)$} results in a $P$-purification subsumption for $N$.

  $\impliedby$: 
  Let $s: R_P(N) \to N$ be a $P$-purification subsumption for $N$ and let $R \in \mathrm{Res}_P^{\leq 1}(N)$.
  We need to show that there exists an $S \in N$ with $S \subsumesLvelim{L^\perp} R$.
  If $R \in N$, then with $S = R$ we get $S \subsumesLvelim{L^\perp} R$ by reflexivity of $\subsumesLvelim{L^\perp}$.
  Otherwise $R = \mathrm{Res}(P, P')$ for some $P$-resolvable pointed clause $P' \in N$.
  This means $P' \in R_P(N)$.
  Since $s$ is a $P$-purification subsumption we get $s(P') \in N$ and $s(P') \subsumesLvelim{L^\perp} \mathrm{Res}(P, P') = R$ which concludes the proof.
\end{proof}

We now introduce the notion of \emph{purification subsumption graphs}.
\begin{definition}
  Let $N$ be a clause set, $P$ be a pointed clause and $s$ be a $P$-purification subsumption for $N$.
  Define the \emph{purification subsumption graph $G(N,P,s) = (V,E)$} as the edge-labelled graph with $V = N$ and 
  $$
  E = \mset{P' \overset{L'}{\to} s(P') \suchthat \text{$P' \in R_P(N)$ and $L'$ is the designated literal of $P'$}}.
  $$
\end{definition}

Note that different choices of purification subsumptions can lead to different graphs~$G(N,P,s)$ and can also affect properties such as being cyclic or acyclic as the following example shows.
\begin{example}
  Let $N, P, s_1$ and $s_2$ as in \Cref{ex.purification-subsumption}.
Then $G(N,P,s_1)$ is given by
\begin{center}
  \begin{tikzpicture}[>=stealth, node distance=2.5cm,
  every node/.style={draw=none},
  every edge/.style={->, thick},
  every edge quotes/.style={draw=none, fill=none, font=\footnotesize, inner sep=0pt}]

    \node (1) {$X(a,b)$};
    \node (2) [right of=1] {$X(b,a)$};
    \node (3) [right of=2] {$A(b,a)$};

    \draw[->] (1) -- node[above, draw=none, inner sep=0pt,font=\footnotesize] {$X(a,b)$} (2);
    \draw[->] (2) -- node[above, draw=none, inner sep=0pt,font=\footnotesize] {$X(b,a)$} (3);
  \end{tikzpicture}
\end{center}
which is acyclic whereas $G(N,P,s_2)$ is given by
\begin{center}
  \begin{tikzpicture}[>=stealth, node distance=2.5cm,
  every node/.style={draw=none},
  every edge/.style={->, thick},
  every edge quotes/.style={draw=none, fill=none, font=\footnotesize, inner sep=0pt}]

    \node (1) {$X(a,b)$};
    \node (2) [right of=1] {$X(b,a)$};
    \node (3) [right of=2] {$A(b,a)$};

    \draw[->] (1) to[bend left=15] node[above=0.1cm, draw=none, inner sep=0pt,font=\footnotesize] {$X(a,b)$} (2);
    \draw[->] (2) to[bend left=15] node[below=0.1cm, draw=none, inner sep=0pt,font=\footnotesize] {$X(b,a)$} (1);
  \end{tikzpicture}
\end{center}
and contains a cycle.
\end{example}

We now introduce the condition on purified clause deletions under which our witness construction will work.
\begin{definition}
  Let $N$ be a clause set, $P$ a pointed clause and $k \in \N$.
  We say $P$ is \emph{$k$-acyclically purified in $N$} if there exists a $P$-purification subsumption $s$ for $N$ such that $G(N,P,s)$ is acyclic and has longest path length $k$.
\end{definition}

Analogous to $\mathrm{\ell Res}_P$ in \Cref{sec.resolution-witnesses} we now introduce a building block for creating witness-transforming substitutions across purified clause deletion.
\begin{definition}
  Let $P = \underline{L(\seq{t}(\seq{v}))} \lor C(\seq{v}) \lor \Lor_{i=1}^p L(\seq{t_i}(\seq{v}))^\perp$ be a pointed clause where $L$ is an~$X$-literal and $C$ does not contain $L^\perp$-literals.
  Let $X$ have arity $k$ and let $\seq{c}$ be a $k$-tuple of fresh constants.
  We inductively define the predicate expression $B_P^k$ by
  \begin{align*}
    B_P^0 &:= \lambda \seq{c}. \mset{\bot} \\
    B_P^{k+1} &:= \lambda \seq{c}. \mset{L(\seq{c})^\perp} \union \mset{\seq{c} \noeq \seq{t}(\seq{v}) \lor C(\seq{v}) \lor \Lor_{i=1}^p R_i(\seq{t_i}(\seq{v}), \seq{z_i}) \suchthat R_i(\seq{c}, \seq{z_1}) \in B_P^k(\seq{c})}
  \end{align*}
  Furthermore set
  $$
  \tau_{\mathrm{PurDel}_P}^k := \begin{cases}
    [X \leftarrow B_P^k] & \text{if the designated literal of $P$ is negative,} \\
    [X \leftarrow \neg B_P^k] & \text{if the designated literal of $P$ is positive.}
  \end{cases}
  $$
\end{definition}

Note that the definition of $B_P^k$ is basically the $k$-fold iteration of $\alpha_{P,Y}$ in a least fixpoint sequence, but it is convenient for this section to have $B_P^k(\seq{c})$ be a clause set, so we give this definition separately.
The following lemma makes the connection to~$\alpha_{P,Y}$ precise.
\begin{lemma}
  Let $P = \underline{L(\seq{t}(\seq{v}))} \lor C(\seq{v}) \lor \Lor_{i=1}^p L(\seq{t_i}(\seq{v}))^\perp$ be a pointed clause where $L$ is an $X$-literal and $C$ does not contain $L^\perp$-literals.
  Furthermore let $Y$ be a fresh predicate variable with the same arity as $X$.
  Then for all $k \in \N$: 
  $$B_P^{k+1} \iff \alpha_{P, Y}[Y \leftarrow B_P^k].$$
\end{lemma}
\begin{proof}
  Let $\seq{c}$ be a tuple of fresh constants. 
  Then we have 
  \begin{align*}
    B_P^{k+1}
    &\iff \lambda \seq{c}. L(\seq{c})^\perp \land \Land_{R_1, \dots, R_p \in B_P^k(\seq{c})} \Forall{\seq{v}}(\seq{c} \noeq \seq{t}(\seq{v}) \lor C(\seq{v}) \lor \Lor_{i=1}^p R_i(\seq{t_i}(\seq{v}), \seq{z_i})) \\
    &\iff \lambda \seq{c}. L(\seq{c})^\perp \land \Forall{\seq{v}}(\seq{c} \noeq \seq{t}(\seq{v}) \lor C(\seq{v}) \lor \Lor_{i=1}^p B_P^k(\seq{t_i}(\seq{v}))) \\
    &\iff \alpha_{P, Y}[Y \leftarrow B_P^k].
\end{align*}
\end{proof}

The following example computes $B_P^k$ for a few iterations for a pointed clause $P$ from \Cref{ex.main-example.fixpoint-witness}.
\begin{example}
  Let $P = B(v_1,v_2) \lor \underline{\neg X(v_1)} \lor X(v_2)$ and let $c$ be a fresh constant.
  Then we have
  \begin{align*}
    B_P^0(c) &= \mset{\bot} \iff \bot \\
    B_P^1(c) &\iff \mset{X(c), c \noeq v_1 \lor B(v_1,v_2)} \iff X(c) \land \Forall{v_1} B(c,v_1) \\
    B_P^2(c) &\iff \mset{X(c), B(c, v_1) \lor X(v_1), B(c, v_1) \lor B(v_1, v_2)} \\
    &\iff X(c) \land \Forall{v_1} (B(c,v_1) \lor X(v_1)) \land \Forall{v_1}\Forall{v_2}(B(c,v_1) \lor B(v_1, v_2)). \\
    \intertext{For arbitrary $k$ we get}
    B_P^k(c) &\iff \biggl(\Land_{0 \leq i < k}\Forall{v_1}\dots\Forall{v_{i}}\left(\Lor_{0 \leq j < i}B(v_j, v_{j+1}) \lor X(v_{i})\right) \\
    &\qquad\quad\  \land \Forall{v_1}\dots\Forall{v_{k}}\Lor_{0 \leq j < k}B(v_j, v_{j+1})\biggr)[v_0 \leftarrow c].
  \end{align*}
\end{example}

Since the initial value for the fixpoint iteration is $\bot$, the iterations $B_P^k(\seq{u})$ correspond to a \emph{least} fixpoint sequence of $\alpha_{P, Y}$.
This also means that the $B_P^k$ are monotone in $k$.
\begin{lemma}
  \label{b-resolvents-monotonicity}
  For all $k \in \N$ we have $B_P^k \imp B_P^{k+1}$.
\end{lemma}

The goal is now to prove that $\tau_{\mathrm{PurDel}_P}^k$ is witness-transforming across purified clause deletion steps $N \uplus \mset{P} / N$ given that $P$ is $k$-acyclically purified in $N$.
We follow a similar strategy as in the case for resolution witnesses in \Cref{sec.resolution-witnesses}.
We show two results analogous to \Cref{lres-satisfies-pointed-clause} and \Cref{N-implies-N-with-lres}, i.e., we show $\models P\tau_{\mathrm{PurDel}_P}^k$ and~\mbox{$N \imp N\tau_{\mathrm{PurDel}_P}^k$} if $P$ is $k$-acyclically purified in $N$.
\begin{lemma}
  \label{finite-witness-satisfies-pointed-clause}
  Let $P$ be a pointed clause and $k \in \N$.
  Then $\models P\tau_{\mathrm{PurDel}_P}^k$.
\end{lemma}
\begin{proof}
  We have $P = \underline{L(\seq{t}(\seq{v}))} \lor C(\seq{v}) \lor \Lor_{i=1}^p L(\seq{t_i}(\seq{v}))$ where $C$ does not contain $L^\perp$-literals and~$L(\seq{t}(\seq{v}))$ is the designated literal of $P$.
  Also let $\seq{c}$ be a tuple of fresh constants.
  If $k=0$ we have $B_P^k \iff \lambda \seq{c}. \bot$, thus~\mbox{$L(\seq{t}(\seq{v}))\tau_{\mathrm{PurDel}_P}^k \iff \neg B_P^k(\seq{u}) \iff \top$}, so we get $\models P\tau_{\mathrm{PurDel}_P}^k$.
  For $k > 0$ we have
  \begin{equation}
    \label{b-definition}
    B_P^{k} \iff \lambda \seq{c}. L(\seq{c})^\perp \land \Forall{\seq{v}}(\seq{c} \noeq \seq{t}(\seq{v}) \lor C(\seq{v}) \lor \Lor_{i=1}^p B_P^{k-1}(\seq{t_i}(\seq{v}))).
  \end{equation}
  Then we get
  $$P\tau_{\mathrm{PurDel}_P}^k \iff \neg B_P^k(\seq{t}(\seq{v})) \lor C(\seq{v})\tau_{\mathrm{PurDel}_P}^k \lor \Lor_{i=1}^p B_P^{k}(\seq{t_i}(\seq{v})).$$
  By \Cref{b-resolvents-monotonicity} we have $B_P^{k-1}(\seq{t_i}(\seq{v})) \imp B_P^{k}(\seq{t_i}(\seq{v}))$ so it suffices to show
  $$\models \neg B_P^k(\seq{t}(\seq{v})) \lor C(\seq{v})\tau_{\mathrm{PurDel}_P}^k \lor \Lor_{i=1}^p B_P^{k-1}(\seq{t_i}(\seq{v})).$$
  Since $C$ does not contain~$L^\perp$-literals and $B_P^k \imp \lambda \seq{c}. L(\seq{c})^\perp$ we get $C \imp C\tau_{\mathrm{PurDel}_P}^k$ by \Cref{monotonicity-applying-substitutions}.
  Thus it suffices to show $\models \neg B_P^k(\seq{t}(\seq{v})) \lor C(\seq{v}) \lor \Lor_{i=1}^p B_P^{k-1}(\seq{t_i}(\seq{v}))$.
  From \eqref{b-definition} we get~\mbox{$B_P^k(\seq{t}(\seq{v})) \imp C(\seq{v}) \lor \Lor_{i=1}^p B_P^{k-1}(\seq{t_i}(\seq{v}))$} which finishes the proof.
\end{proof}

Before we turn to proving $N \imp N\tau_{\mathrm{PurDel}_P}^k$ we show a property that $\subsumesLvelimtrans{L}$ inherits from ordinary subsumption, namely if $S \subsumes C$ and $S$ does not contain any $L$-literals, then $S$ subsumes $C \setminus C_L$.
\begin{lemma}
  \label{subsumesLvelimtrans-if-not-containing-L-literals}
  Let $S$ and $C$ be clauses and let $L$ be a literal.
  If $S \subsumesLvelimtrans{L} C$ and $S$ contains no $L$-literals, then $S \subsumesLvelimtrans{L} C \setminus C_L$.
\end{lemma}
\begin{proof}
    We show this by induction on the size of $C_L$.
    If this set is empty, then $C \setminus C_L = C$ and the statement follows directly from the assumption.
    Now let $C_L$ have size $n+1$.
    Thus there is at least one $L$-literal in $C_L$, call it $L'$.
    Since $S$ does not contain any $L$-literals we get that the first case of \Cref{res-subsumption-velim-trans-commutation} must hold, i.e., $S \subsumesLvelimtrans{L} C \setminus \mset{L'}$.
    Now the size of~\mbox{$(C \setminus \mset{L'})_L$} is $n$.
    So by induction hypothesis we get $S \subsumesLvelimtrans{L} (C \setminus \mset{L'}) \setminus (C \setminus \mset{L'})_L$.
    Since $(C \setminus \mset{L'}) \setminus (C \setminus \mset{L'})_L = C \setminus C_L$ the statement follows.
\end{proof}

We now turn to proving $N \imp N\tau_{\mathrm{PurDel}_P}^k$.
To do that we show the following lemma which closely corresponds to \Cref{purified-implies-trans-subsume-velim-is-closed-under-resolution}.
Whereas in \Cref{purified-implies-trans-subsume-velim-is-closed-under-resolution} we showed how $P$ being purified in $N$ implies $N~\subsumesLvelimtrans{L^\perp}~\mathrm{Res}_P^{<\omega}(N)$ we now show that $P$ being $k$-acyclically purified in $N$ implies $N \subsumesLvelimtrans{L^\perp} R$ for all clauses $R$ that result from $C \in N$ by replacing $L^\perp$-literals by clauses from $B_P^{k}$.
We first give a shorthand notation for this replacement operation.
\begin{definition}
  Let $C$ be a clause, $L$ a literal with predicate symbol $X$ of arity $k$, $A \subseteq C_L$, $\seq{c}$ a $k$-tuple of fresh constants, $B(\seq{c})$ a clause set and $E: A \to B(\seq{c}): L' \mapsto E_{L'}(\seq{c}, \seq{z_{L'}})$.
  Then set
  $$
  C^E := (C \setminus A) \lor \Lor_{L(\seq{t}) \in A} E_{L(\seq{t})}(\seq{t}, \seq{z_{L(\seq{t})}})
  $$
\end{definition}
The subset $A$ defines which literals are to be replaced and $E$ defines which clauses from $B(\seq{c})$ they are replaced with.
The fresh constants $\seq{c}$ serve as placeholders for substituting the terms of the literals $L' \in A$ into the clause $E_{L'}(\seq{c}, \seq{z_{L'}})$.

\begin{lemma}
  \label{acyclically-purified-implies-subsumesLvelimtrans-B-closure}
  Let $N$ be a clause set, $P$ a pointed clause with designated literal $L$ such that $P$ is $k$-acyclically purified in $N$.
  Also let $m \geq k$, let $C$ be a clause, $A \subseteq C_{L^\perp}$ and $E: A \to B_P^m(\seq{c})$ where $\seq{c}$ is a tuple of fresh constants.
  Then $N \subsumesLvelimtrans{L^\perp} C$ implies $N \subsumesLvelimtrans{L^\perp} C^E$.
\end{lemma}
\begin{proof}
  Since $P$ is $k$-acyclically purified in $N$ there exists a $P$-purification subsumption $s$ for~$N$ such that $G(N,P,s)$ is acyclic and has maximum path length $k$.
  For $S \in N$ denote by $k(S)$ the maximum length of a path in $G(N,P,s)$ starting from $S$.
  For $l \in \N$, denote by $N_l$ the set of clauses $S \in N$ with $k(S) \leq l$.
  Then note that $N_l = N$ for all $l \geq k$.

  We now show a slightly stronger statement than the statement of the lemma which we denote by $Q(l, n)$ for $l, n \in \N$:
  For all $m \geq l$, clauses $C$, $A \subseteq C_{L^\perp}$ with $\abs{A} \leq n$ and~\mbox{$E: A \to B_P^m(\seq{c})$} we have that $N_l \subsumesLvelimtrans{L^\perp} C$ implies $N_l \subsumesLvelimtrans{L^\perp} C^E$.
  We then get the statement of the lemma from $Q(k, \abs{A})$.

  We now show $Q(l,n)$ for all $l, n \in \N$ by lexicographic induction on $l$ and $n$, i.e., we show the statement for given $l$ and $n$ by using the statement for all $l', n' \in \N$ with either $l' < l$ or~$l' = l$ and $n' < n$.
  We first show two simpler cases, namely $n=0$ and $l=0$.

  First consider $n=0$.
  By assumption we have $N_l \subsumesLvelimtrans{L^\perp} C$.
  Since $n=0$ we also have~\mbox{$\abs{A} = 0$}, i.e., $A$ is empty and therefore $C^E = C$.
  This means $N_l \subsumesLvelimtrans{L^\perp} C^E$.

  Now consider $l=0$. 
  Since $N \subsumesLvelimtrans{L^\perp} C$ there is an $S \in N_l$ with $S \subsumesLvelimtrans{L^\perp} C$.
  We need to show there is an $S' \in N_l$ with $S' \subsumesLvelimtrans{L^\perp} C^E$.
  Since $k(S) \leq l = 0$ we have that $S$ has no outgoing edges in $G(N,P,s)$.
  Therefore $S$ contains no $L^\perp$-literals and so $S \subsumesLvelimtrans{L^\perp} C \setminus C_{L^\perp}$ by \Cref{subsumesLvelimtrans-if-not-containing-L-literals}.
  Since $A \subseteq C_{L^\perp}$ we get $C \setminus C_{L^\perp} \subseteq C \setminus A$ and thus by \Cref{relationships-between-subsumption-relations} and transitivity of $\subsumesLvelimtrans{L^\perp}$ we get
  $$S \subsumesLvelimtrans{L^\perp} (C \setminus A) \lor \Lor_{L(\seq{t})^\perp \in A} E_{L(\seq{t})^\perp}(\seq{t}, \seq{z_{L(\seq{t})^\perp}}) = C^E.$$
  Thus picking $S' = S$ finishes this case.

  Now to the induction step.
  Let $m \geq l$, let $C$ be a clause, $A \subseteq C_{L^\perp}$ with $\abs{A} \leq n$ and let $E: A \to B_P^m(\seq{c}): L' \mapsto E_{L'}(\seq{c}, \seq{z_{L'}})$.
  Furthermore let $N_l \subsumesLvelimtrans{L^\perp} C$.
  We need to show~\mbox{$N_l \subsumesLvelimtrans{L^\perp} C^E$}, i.e., there exists an $S' \in N_l$ with $S \subsumesLvelimtrans{L^\perp} C^E$.
  By induction hypothesis we can use $Q(l', n')$ for $l', n' \in \N$ with either $l' < l$ or $l' = l$ and $n' < n$.
  Since we already showed the cases $l=0$ and $n=0$ we can further assume $l>0$ and $n>0$.

  Since $N_l \subsumesLvelimtrans{L^\perp} C$ there is an $S \in N_l$ such that $S \subsumesLvelimtrans{L^\perp} C$.
  Say $P$ has the form~\mbox{$P = \underline{L(\seq{t}(\seq{v}))} \lor C_P(\seq{v}) \lor \Lor_{i=1}^p L(\seq{t_i}(\seq{v}))^\perp$} where $C_P$ does not contain any $L^\perp$-literals.
  Since $m \geq l > 0$ and using the definition of $B_P^m(\seq{c})$ we have for all $L' \in A$ that either~\mbox{$E(L') = L(\seq{c})^\perp$} or $E(L') = \seq{c} \noeq \seq{t}(\seq{v}) \lor C_P(\seq{v}) \lor \Lor_{i=1}^p R_i(\seq{t_i}(\seq{v}), \seq{z_i})$ for some~\mbox{$R_i(\seq{c}, \seq{z_i}) \in B_P^{m-1}(\seq{c})$}.
  If for all $L' \in A$ we have $E(L') = L(\seq{c})^\perp$, then $C^E = C$ and we are done by picking $S' = S$.
  
  Otherwise there is a $1 \leq j \leq p$ with $L^\ast := L(\seq{s_j}(\seq{w}))^\perp \in A$ such that 
  $$E(L^\ast) = \seq{c} \noeq \seq{t}(\seq{v}) \lor C_P(\seq{v}) \lor \Lor_{i=1}^p R_i(\seq{t_i}(\seq{v}), \seq{z_i})$$ for some $R_1(\seq{c}, \seq{z_1'}), \dots, R_p(\seq{c}, \seq{z_p'}) \in B_P^{m-1}(\seq{c})$.
  Now consider $A^\ast := A \setminus \mset{L^\ast}$ and the function~\mbox{$E^\ast : A^\ast \to B_P^{m-1}(\seq{c}): L' \mapsto E(L')$}, i.e., the restriction of $E$ to $A^\ast$ and set
  $$
  P' := C^{E^\ast}[\underline{L^\ast}] = \underline{L^\ast} \lor (C \setminus A) \lor \Lor_{L(\seq{s})^\perp \in A^\ast} E_{L(\seq{s})^\perp}(\seq{s}, \seq{z_{L(\seq{s})}}).
  $$ 
  Then $\abs{A^\ast} = \abs{A} - 1 < n$.
  Now we use $Q(l, \abs{A^\ast})$ to obtain an $S_1 \in N_l$ with $S_1 \subsumesLvelimtrans{L^\perp} P'$.
  By \Cref{res-subsumption-velim-trans-commutation} there are two cases.

  Case 1: $S_1 \subsumesLvelimtrans{L^\perp} P' \setminus \mset{L^\ast}$.
  Then we have 
  $$
  S_1 \subsumesLvelimtrans{L^\perp} (C \setminus A) \lor \Lor_{L(\seq{s})^\perp \in A^\ast} E_{L(\seq{s})^\perp}(\seq{s}, \seq{z_{L(\seq{s})}}) \subseteq C^E.
  $$
  Thus, picking $S' = S_1$ finishes the proof.

  Case 2: There exists an $L^\perp$-literal $L' \in S_1$ such that~\mbox{$\mathrm{Res}(P, S_1[\underline{L'}]) \subsumesLvelimtrans{L^\perp} \mathrm{Res}(P, P')$}.
  Then $G(N,P,s)$ has an edge $S_1 \overset{L'}{\longrightarrow} s(S_1)$ and we have~\mbox{$S_2 := s(S_1) \subsumesLvelim{L^\perp} \mathrm{Res}(P, S_1[\underline{L'}])$}.
  Since we have $k(S_1) \leq l$ we get $k(S_2) < l$ and
  since $\subsumesLvelimtrans{L^\perp}$ is transitive we further get~\mbox{$S_2 \subsumesLvelimtrans{L^\perp} \mathrm{Res}(P, P') := C_2$}.
  Note that
  $$
  C_2 = (C \setminus A) \lor \Lor_{L(\seq{s})^\perp \in A^\ast} E_{L(\seq{s})^\perp}(\seq{s}, \seq{z_{L(\seq{s})}}) \lor \seq{s_j}(\seq{w}) \noeq \seq{t}(\seq{v}) \lor C_P(\seq{v}) \lor \Lor_{i=1}^p L(\seq{t_i}(\seq{v}))^\perp.
  $$

  Now consider the set of $L^\perp$-literals $A_2 := \mset{L(\seq{t_i}(\seq{v}))^\perp \suchthat 1 \leq i \leq p} \subseteq (C_2)_{L^\perp}$ and the function~\mbox{$E_2 : A_2 \to B_P^{m-1}(\seq{c}): L(\seq{t_i}(\seq{v}))^\perp \mapsto R_i(\seq{t_i}(\seq{v}, \seq{z_i}))$}.
  Since $R_i(\seq{c}, \seq{z_i}) \in B_P^{m-1}(\seq{c})$ and~\mbox{$m-1 \geq l - 1$} we can use $Q(l-1, \abs{A_2})$ to get an $S_2' \in N_{l-1}$ such that $S_2' \subsumesLvelimtrans{L^\perp} C_2^{E_2}$.
  Note that
  \begin{align*}
    C_2^{E_2}
    &= (C \setminus A) \lor \Lor_{L(\seq{s})^\perp \in A^\ast} E_{L(\seq{s})^\perp}(\seq{s}, \seq{z_{L(\seq{s})}}) \lor \seq{s_j}(\seq{w}) \noeq \seq{t}(\seq{v}) \lor C_P(\seq{v}) \lor \Lor_{i=1}^p R_i(\seq{t_i}(\seq{v}), \seq{z_i}) \\
    &= (C \setminus A) \lor \Lor_{L(\seq{s})^\perp \in A^\ast} E_{L(\seq{s})^\perp}(\seq{s}, \seq{z_{L(\seq{s})}}) \lor E_{L^\ast}(\seq{s_j}(\seq{w})) \\
    &= (C \setminus A) \lor \Lor_{L(\seq{s})^\perp \in A} E_{L(\seq{s})^\perp}(\seq{s}, \seq{z_{L(\seq{s})}})\\
    &= C^E.
  \end{align*}
  Also note that $N_{l-1} \subseteq N_l$, so we have $S_2' \in N_l$.
  Therefore picking $S' = S_2'$ concludes the proof.
\end{proof}

The following lemma is an analog to \Cref{purified-implies-resolution-closure}, but instead of the resolution closure we consider the set of clauses resulting from clauses in $N$ by replacing $L^\perp$-literals by $B_P^m$.
\begin{lemma}
  \label{acyclically-purified-implies-B-closure}
  Let $N$ be a clause set and let $P$ be a pointed clause with designated literal $L$ such that $P$ is $k$-acyclically purified in $N$.
  Furthermore let $m \geq k$ and $C \in N$.
  Then~\mbox{$C = C_0(\seq{w}) \lor \Lor_{i=1}^q L(\seq{s_i}(\seq{w}))^\perp$}
  where $C_0$ contains no $L^\perp$-literals and
  we have
  $$N \imp \Forall{\seq{w}}(C_0(\seq{w}) \lor \Lor_{i=1}^q B_P^m(\seq{s_i}(\seq{w})))$$.
\end{lemma}
\begin{proof}
  Note that $C = C_0(\seq{w}) \lor \Lor_{i=1}^q L(\seq{s_i}(\seq{w}))^\perp$ for a clause $C_0$ which does not contain $L^\perp$-literals and let $\seq{c}$ be a tuple of fresh constants.
  By \Cref{acyclically-purified-implies-subsumesLvelimtrans-B-closure} we have for all $m \geq k$ and~\mbox{$R_1(\seq{c}, \seq{z_1}), \dots, R_q(\seq{c}, \seq{z_q}) \in B_P^m(\seq{c})$} that
  \begin{equation*}
  N \subsumesLvelimtrans{L^\perp} C_0(\seq{w}) \lor \Lor_{i=1}^q R_i(\seq{s_i}(\seq{w}), \seq{z_i}).
  \end{equation*}
  By \Cref{subsumption-velim-implies-logical-implication} we get $N \imp C_0(\seq{w}) \lor \Lor_{i=1}^q R_i(\seq{s_i}(\seq{w}), \seq{z_i})$.
  Taking the conjunction over all~\mbox{$R_1(\seq{c}, \seq{z_1}), \dots, R_q(\seq{c}, \seq{z_q}) \in B_P^m(\seq{c})$} we get
  $$N \imp \Land_{R_1(\seq{c}, \seq{z_1}) \in B_P^m(\seq{c})} \dots \Land_{R_q(\seq{c}, \seq{z_q}) \in B_P^m(\seq{c})} (C_0(\seq{w}) \lor \Lor_{i=1}^q R_i(\seq{s_i}(\seq{w}), \seq{z_i})).$$
  Using the distributivity of $\land$ and $\lor$ we move the conjunctions into the disjunction to get
  $$N \imp C_0(\seq{w}) \lor \Lor_{i=1}^q \Land_{R_i(\seq{c}, \seq{z_1}) \in B_P^m(\seq{c})} R_i(\seq{s_i}(\seq{w}), \seq{z_i})$$
  which is equivalent to $N \imp C_0(\seq{w}) \lor \Lor_{i=1}^q B_P^m(\seq{s_i}(\seq{w}))$.
  This implies 
  $$N \imp \Forall{\seq{w}}( C_0(\seq{w}) \lor \Lor_{i=1}^q B_P^m(\seq{s_i}(\seq{w})) ).$$
\end{proof}

Using this lemma we can now show the implication $N \imp N\tau_{\mathrm{PurDel}_P}^{k}$.
\begin{lemma}
  \label{N-implies-N-with-finite-witness}
  Let $N$ be a clause set, $P$ a pointed clause such that $P$ is $k$-acyclically purified in $N$.
  Then for all $m \geq k$ we have $N \imp N\tau_{\mathrm{PurDel}_P}^m$.
\end{lemma}
\begin{proof}
  Let $P$ have designated literal $L$ and let $C \in N$.
  Then $C = C_0(\seq{w}) \lor \Lor_{i=1}^q L(\seq{s_i}(\seq{w}))^\perp \in N$ where $C_0$ does not contain $L^\perp$-literals.
  By \Cref{acyclically-purified-implies-B-closure} we get
  $$N \imp \Forall{\seq{w}} (C_0(\seq{w}) \lor \Lor_{i=1}^q B_P^m(\seq{s_i}(\seq{w}))).$$
  From the definition of $B_P^m$ we get $B_P^m \imp \lambda \seq{u}. L(\seq{u})^\perp$.
  Thus we can apply \Cref{monotonicity-applying-substitutions} to get~\mbox{$C_0(\seq{w}) \imp C_0(\seq{w})\tau_{\mathrm{PurDel}_P}^m$}.
  Therefore 
  $$N \imp C_0(\seq{w})\seq\tau_{\mathrm{PurDel}_P}^m \lor \Lor_{i=1}^q B_P^m(\seq{s_i}(\seq{w})).$$
  Note that $B_P^m(\seq{s_i}(\seq{w})) \iff L(\seq{s_i}(\seq{w}))^\perp\tau_{\mathrm{PurDel}_P}^m$.
  Thus we get 
  $$N \imp C_0(\seq{w})\seq\tau_{\mathrm{PurDel}_P}^m \lor \Lor_{i=1}^q L(\seq{s_i}(\seq{w})^\perp)\tau_{\mathrm{PurDel}_P}^m \iff C\tau_{\mathrm{PurDel}_P}^m.$$
  Since $C$ was arbitrary we get $N \imp N\tau_{\mathrm{PurDel}_P}^m$.
\end{proof}

Finally, this shows that $\tau_{\mathrm{PurDel}_P}^k$ is witness-transforming.
\begin{lemma}
  \label{first-order-witness-transforming-across-purified-clause-deletion}
  Let $N$ be a clause set, $P$ a pointed clause such that $P$ is $k$-acyclically purified in $N$.
  Then for all $m \geq k$ the predicate substitution $\tau_{\mathrm{PurDel}_P}^m$ is witness-transforming across the purified clause deletion step $N \uplus \mset{P} / N$.
\end{lemma}
\begin{proof}
  By \Cref{N-implies-N-with-finite-witness} we get $N \imp N\tau_{\mathrm{PurDel}_P}^m$ and by \Cref{finite-witness-satisfies-pointed-clause} we get $\models P\tau_{\mathrm{PurDel}_P}^m$ which means $N \imp (N \union \mset{P})\tau_{\mathrm{PurDel}_P}^m$, i.e., $\tau_{\mathrm{PurDel}_P}^m$ is witness-transforming across $N \uplus \mset{P} / N$.
\end{proof}

We update the witness constructed from a given $\mathcal{C}$-derivation where all pointed clauses $P$ in purified clause deletion steps are acyclically purified in the respective clause set.
\begin{definition}
  \label{def:first-order-witness}
  Let $D = (S_1, \dots, S_m)$ be a $\mathcal{C}$-derivation from $N$ and let $N = N_0, N_1, \dots, N_m$ be the intermediate clause sets of the derivation.
  Let $1 \leq i_1 < \dots < i_l \leq m$ such that~\mbox{$S_{i_1}, \dots, S_{i_l}$} are exactly the purified clause deletion steps in $D$.
  Denote by $I(D)$ the set of purified clause deletion indices in $D$, i.e., $I(D) := \mset{i_1, \dots, i_l}$.
  For $i \in I(D)$ let $P_i$ be the pointed clause such that $S_{i} = \mathrm{PurDel}_{P_i}$.
  Furthermore let $f: I(D) \to \N$.
  We call $f$ an \emph{acyclic purification annotation for $D$} if for all $i \in D$ there is a $k \leq f(i)$ such that $P_i$ is $k$-acyclically purified in $N_i$.
  For $1 \leq i \leq m$ define
  $$
  \tau_i:= \begin{cases}
    \tau_{S_i} & \text{if $i \not\in I(D)$,} \\
    \tau_{\mathrm{PurDel}_{P_i}}^{f(i)} & \text{if $i \in I(D)$.}
  \end{cases}
  $$
  and set $\sigma_{\mathrm{fo}}(D, f) := \tau_1 \cdots \tau_m$.
\end{definition}

From this we can derive our main result of the paper.
\begin{theorem}
  Let $D$ be an $\seq{X}$-eliminating $\mathcal{C}$-derivation from $N$ and let $f: I(D) \to \N$ be an acyclic purification annotation for $D$.
  Then $\sigma_{\mathrm{fo}}(D, f)$ is a first-order witness for $\Exists{\seq{X}} N$.
\end{theorem}
\begin{proof}
  The proof is similar to the proof of \Cref{eliminating-derivations-admit-infinite-witness}, but uses some different lemmas.

  Let $D = (S_1, \dots, S_m)$.
  We proceed by induction on $m$.
  If $m = 0$, we have $\sigma_{\mathrm{fo}}(D, f) = \mathrm{id}$ and $N$ does not contain any variables from $\seq{X}$.
  Thus $\Exists{\seq{X}} N \imp N \imp N\sigma_{\mathrm{fo}}(D, f)$, i.e., $\sigma_{\mathrm{fo}}(D, f)$ is a witness for $\Exists{\seq{X}} N$.

  Now to the induction step.
  Let $D = (S_1, \dots, S_{m+1})$ be $\seq{X}$-eliminating from $N$ and let~\mbox{$f: I(D) \to \N$} be an acyclic purification annotation for $D$.
  Denote by $N_1$ the conclusion of $S_1$ with premise $N$.
  Then $D_1 = (S_2, \dots, S_{m+1})$ is an $\seq{X}$-eliminating derivation from $N_1$.
  Then consider $f_1: I(D_1) \to \N: i \mapsto f(i - 1)$, i.e., the indices get shifted by $1$.
  Then $f_1$ is an acyclic purification annotation for $D_1$.
  Thus by induction hypothesis we have that $\sigma_{\mathrm{fo}}(D_1, f_1)$ is a witness for $\Exists{\seq{X}} N_1$.
  We have $\sigma_{\mathrm{fo}}(D, f) = \tau_1\sigma_{\mathrm{fo}}(D_1, f_1)$.
  Now there are two cases.
  
  Case 1: $S_1$ is not a purified clause deletion step.
  Then $\tau_1$ is witness-preserving by \Cref{witness-preservation-for-steps-without-purified-clause-deletion} and \Cref{witness-preservation-sufficient-condition} we get that $\sigma_{\mathrm{fo}}(D, f)$ is a witness for $\Exists{\seq{X}} N$.

  Case 2: $S_1$ is a purified clause deletion step.
  We have $S_1 = \mathrm{PurDel}_P$ for some pointed clause $P$.
  Since $f$ is an acyclic purification annotation for $D$ we have that $P$ is $k$-acyclically purified in $N_1$ for some $k \leq f(1)$.
  By \Cref{first-order-witness-transforming-across-purified-clause-deletion} we have that $\tau_1 = \tau_{\mathrm{PurDel}_P}^{f(1)}$ is witness-transforming which by \Cref{witness-preservation-sufficient-condition} means that it is witness-preserving.
  Thus $\sigma_{\mathrm{fo}}(D, f)$ is a witness for $\Exists{\seq{X}} N$.
  
  Since for all $k \in \N$ the predicate expression $B_P^k$ is first-order it follows that $\tau_{\mathrm{PurDel}_P}^k$ is first-order.
  For $\mathcal{C}$-derivation steps $S$ other than purified clause deletion we already know that~\mbox{$\tau_S$} is first-order by construction.
\end{proof}

This theorem gives some freedom in choosing the produced witnesses by varying~$f$.
However, most of the time we are interested in small witnesses so we will choose $f(i)$ to be the smallest $k$ such that $P_i$ is $k$-acyclically purified in $N_i$.
We now show how this construction allows us to find a first-order witness for derivation $D_2$ from \Cref{ex.main-example.second-witness} where our previous methods could not.
\begin{example}
  \label{main-example}
  Recall the clause set $N$ from \Cref{ex.main-example.second-witness}
  \begin{align*}
    (1)\ B(a,v) \qquad (2)\ X(a) \qquad (3)\ B(u,v) \lor \neg X(u) \lor X(v) \qquad  (4)\ \neg X(c)
  \end{align*}
  and one of its resolvents
  $$
  (6)\ a \noeq c
  $$
  and the $X$-eliminating derivation $D_2$ from $N:= \mset{1, 2,3, 4}$.
  \begin{align*}
    D_2 = 
    \Axiom$\mset{1,2,3,4}\fCenter$
    \RightLabel{$\mathrm{PurDel}_{3.2}$}
    \UnaryInf$\mset{1, 2, 4}\fCenter$
    \RightLabel{$\mathrm{Res}_{2.1,4.1}$}
    \UnaryInf$\mset{1,2,4,6}\fCenter$
    \RightLabel{$\mathrm{PurDel}_{2.1}$}
    \UnaryInf$\mset{1,4,6}\fCenter$
    \RightLabel{$\mathrm{ExtPurDel}_X^{-}$}
    \UnaryInf$\mset{1,6}\fCenter$
    \DisplayProof
  \end{align*}
  We check that there is an acyclic purification annotation $f$ for $D_2$.
  Note that \mbox{$I(D_2) = \mset{1, 3}$}.
  Now we need to find $P$-purification subsumptions for all purified clause deletion steps in $D_2$ such that the resulting purification subsumption graph is acyclic and determine its maximal length.

  We first consider $\mathrm{PurDel}_{3.2}$ with conclusion~\mbox{$N_1 := \mset{1,2,4}$}.
  Then the only $3.2$-resolvable pointed clause in $N_1$ is $\underline{X(a)}$.
  Its resolvent with $3.2$ is~\mbox{$a \noeq u \lor B(u,v) \lor X(v)$} which under variable elimination reduces to $B(a, v) \lor X(v)$ which is subsumed by $1$.
  Thus a $3.2$-purification subsumption for $N_1$ is given by $s_1(\underline{X(a)}) = B(a,v)$.
  The corresponding purification subsumption graph $G(N_1, 3.2, s_1)$ is
  \begin{center}
    \begin{tikzpicture}[>=stealth, node distance=2.25cm,
    every node/.style={draw=none},
    every edge/.style={->, thick},
    every edge quotes/.style={draw=none, fill=none, font=\footnotesize, inner sep=0pt}]

      \node (1) {$B(a,v)$};
      \node (2) [right of=1] {$X(a)$};
      \node (3) [right of=2] {$\neg X(c)$};

      \draw[->] (2) -- node[above, draw=none, inner sep=0pt] {$X(a)$} (1);
    \end{tikzpicture}
  \end{center}
  which is acyclic and has longest path length $1$.
  Thus we set $f(1) = 1$.

  Now consider $\mathrm{PurDel}_{2.1}$ with conclusion~\mbox{$N_2 := \mset{1,4,6}$}.
  Then the only $2.1$-resolvable pointed clause in $N_2$ is $\underline{\neg X(c)}$.
  Its resolvent with $2.1$ is $a \noeq c$ which is subsumed by $6$.
  Thus a $2.1$-purification subsumption for $N_2$ is given by $s_2(\underline{\neg X(c)}) = a \noeq c$.
  The corresponding purification subsumption graph $G(N_2, 2.1, s_2)$ is
  \begin{center}
    \begin{tikzpicture}[>=stealth, node distance=2.25cm,
    every node/.style={draw=none},
    every edge/.style={->, thick},
    every edge quotes/.style={draw=none, fill=none, font=\footnotesize, inner sep=0pt}]

      \node (1) {$B(a,v)$};
      \node (2) [right of=1] {$\neg X(c)$};
      \node (3) [right of=2] {$a \noeq c$};

      \draw[->] (2) -- node[above, draw=none, inner sep=0pt] {$\neg X(c)$} (3);
    \end{tikzpicture}
  \end{center}
  which is acyclic and has longest path length $1$.
  Thus we set $f(3) = 1$.

  This shows that $f$ is an acyclic purification annotation for $D_2$.
  Thus we can compute the corresponding witness by
  $\sigma_{\mathrm{fo}}(D_2, f) = \tau_{\mathrm{PurDel}_{3.2}}^1\mathrm{id}\tau_{\mathrm{PurDel}_{2.1}}^1[X \leftarrow \lambda u. \bot]$.
  We have~\mbox{$B_{3.1}^1 \iff \lambda d. \mset{X(d), B(d, v)}$} and thus $\tau_{3.2}^1 \iff [X \leftarrow \lambda u. X(u) \land \Forall{v} B(u,v)]$.
  Furthermore~\mbox{$B_{2.1}^1 \iff \lambda d. \mset{\neg X(d), d \noeq c}$} and thus $\tau_{2.1}^1 \iff [X \leftarrow \lambda u. X(u) \lor u \oeq a]$.
  In total we get 
  $$
  \sigma_{\mathrm{fo}}(D_2, f) \iff \tau_{3.2}^1[X \leftarrow \lambda u. u \oeq c] \iff [X \leftarrow \lambda u. u \oeq a \land \Forall{v} B(u,v)].
  $$
  We verify that $\sigma_{\mathrm{fo}}(D_2, f)$ is a witness for $\Exists{X} N$.
  It suffices to show $\Exists{X} N \imp N\sigma_{\mathrm{fo}}(D_2, f)$.
  Since SCAN is correct by \Cref{correctness-of-scan} we have $\Exists{X} N \iff \mset{1,6}$ thus it suffices to show~\mbox{$\mset{B(a,v), a \noeq c} \imp C\sigma_{\mathrm{fo}}(D_2, f)$} for all $C \in N$.
  For $C = B(a,v)$ this follows immediately.
  For $C = X(a)$ we have $C\sigma_{\mathrm{fo}}(D_2, f) \iff a \oeq a \land \Forall{v} B(a, v)$ which is implied by $\mset{B(a,v)}$.
  For $C = \neg X(c)$ we have $C\sigma_{\mathrm{fo}}(D_2, f) \iff c \noeq a \lor \neg\Forall{v} B(c,v)$ which is implied by $a \noeq c$.
  Thus~\mbox{$\sigma_{\mathrm{fo}}(D_2, f)$} is a first-order witness for $\Exists{X} N$.
\end{example}

We now show that the class of $k$-acyclically purified pointed clauses $P$ in $N$ includes all one-sided pointed clauses and thus extends the witness construction method from \cite{AchammerHetzlSchmidt25}.
\begin{lemma}
  \label{one-sided-implies-acyclically-purified}
  Let $N$ be a clause set and let $P$ be a one-sided pointed clause which is purified in $N$.
  Then $P$ is $k$-acyclically purified in $N$ for some $k \in \N$.
\end{lemma}
\begin{proof}
  Let $P = \underline{L({\seq{t}})} \lor C$.
  Since $P$ is purified in $N$ we get by \Cref{purified-iff-exists-purification-subsumption} that there exists a~$P$-purification subsumption $s$ for $N$.
  Now it suffices to show that $G(N, P, s)$ is acyclic.
  Let 
  $$C_0 \overset{L_0}{\longrightarrow} \dots \overset{L_{n-1}}{\longrightarrow} C_n$$ 
  be a path in $G(N,P,s)$.
  For a clause $C'$ denote by $\#_{L^\perp}(C')$ the number of $L^\perp$-literals in $C'$.
  We need to show that $i \neq j$ implies $C_i \neq C_j$.
  We do this by showing $\#_{L^\perp}(C_{k+1}) < \#_{L^\perp}(C_{k})$ for all $0 \leq k < n$.
  This means all $C_i$ along the path must be pairwise distinct since they have a distinct number of $L^\perp$-literals.

  For $0 \leq k < n$ we have $C_{k+1} = s(C_{k}[\underline{L_{k}}])$.
  Since $s$ is a $P$-purification subsumption for $N$ we have $C_{k+1} \subsumesLvelim{L^\perp} \mathrm{Res}(P, C_{k}[\underline{L_{k}}])$.
  Thus it suffices to show the following two statements:
  \begin{enumerate}
    \item For all clauses $C'$ and $L^\perp$-literals $L' \in C'$ we have $\#_{L^\perp}(\mathrm{Res}(P, C'[\underline{L'}])) < \#_{L^\perp}(C')$.
    \item For all clauses $C', C''$ with $C' \subsumesLvelim{L^\perp} C''$ we have $\#_{L^\perp}(C') \leq \#_{L^\perp}(C'')$.
  \end{enumerate}
  We now prove the two statements to finish the proof.
  \begin{enumerate}
    \item Let $L' = L(\seq{s})^\perp$, then 
    $$\mathrm{Res}(P, C'[\underline{L'}]) = \seq{t} \noeq \seq{s} \lor C \lor (C' \setminus \mset{L'}).$$
    Since $C$ contains no $L^\perp$-literals and $C' \setminus \mset{L'}$ contains one less $L^\perp$-literal than $C'$ we get~\mbox{$\#_{L^\perp}(\mathrm{Res}(P, C')) = \#_{L^\perp}(C') - 1 < \#_{L^\perp}(C')$}.
    \item 
    We first show the property for $\subsumesL{L^\perp}$, $\velim$ and $\velimtrans$ which can then be used to show it for~$\subsumesLvelim{L^\perp}$.

    If $C' \subsumesL{L^\perp} C''$, then $C'\sigma \subseteq C''$ for some substitution $\sigma$ which is $L^\perp$-injective in~$C'$.
    Since $\sigma$ is $L^\perp$-injective we have
    $\#_{L^\perp}(C') = \#_{L^\perp}(C'\sigma)$ and since $C'\sigma \subseteq C''$ we get~\mbox{$\#_{L^\perp}(C'\sigma) \leq \#_{L^\perp}(C'')$}.
    This means $\#_{L^\perp}(C') \leq \#_{L^\perp}(C'')$.
    
    If $C' \velim C''$, then $C' = v \noeq t \lor C_0'$ and $C'' = C_0'[v \leftarrow t]$. Thus we get~\mbox{$\#_{L^\perp}(C'') \leq \#_{L^\perp}(C')$}.
    This also shows if $C' \velimtrans C''$, then $\#_{L^\perp}(C'') \leq \#_{L^\perp}(C')$ by induction on the length of a $\velim$-path between $C'$ and $C''$.

    If $C' \subsumesLvelim{L^\perp} C''$, then there is a clause $C'''$ with $C'' \velimtrans C'''$ and $C' \subsumesL{L^\perp} C'''$.
    Then by the previous statements we get $\#_{L^\perp}(C') \leq \#_{L^\perp}(C''') \leq \#_{L^\perp}(C'')$.
  \end{enumerate}
\end{proof}

The following lemma shows that the witness produced by this new method is the same as the one produced using $\mathrm{\ell Res}_P$, if $P$ is one-sided.
\begin{lemma}
  \label{one-sided-B}
  If $P$ is a one-sided pointed clause, then $B_P^k \iff \mathrm{\ell Res}_P$ for all $k \geq 1$.
\end{lemma}
\begin{proof}
  Let $P = \underline{L(\seq{t})} \lor C$.
  Since $P$ is one-sided we have that $C$ does not contain any $L^\perp$-literals.
  Thus $B_P^k =\lambda \seq{c}. \mset{L(\seq{c})^\perp, \seq{c} \noeq \seq{t} \lor C} \iff \mathrm{\ell Res}_P$.
\end{proof}

The following example shows that the existence of an acyclic purification annotation is necessary for the constructed predicate substitution to be a witness.

\begin{example}
    Consider the $\mathcal{C}$-derivation
    $$
    D = \AxiomC{$\mset{\neg X(v) \lor X(f(v)), X(f(f(v)))}$}
      \RightLabel{$\mathrm{PurDel}_{\underline{\neg X(v)} \lor X(f(v))}$}
      \UnaryInfC{$\mset{X(f(f(v)))}$}
      \RightLabel{$\mathrm{PurDel}_{\underline{X(f(f(v)))}}$}
      \UnaryInfC{$\emptyset$}
      \DisplayProof
    $$
    Note that $P_1 := \underline{\neg X(v)} \lor X(f(v))$ is purified in $N_1 := \mset{X(f(f(v)))}$ since
    the only resolvent between $P_1$ and $X(f(f(v)))$ is $v \noeq f(f(v')) \lor X(f(v))$
    which under variable elimination reduces to $X(f(f(f(v'))))$ and we have $X(f(f(v))) \subsumesL{X(v)} X(f(f(f(v'))))$.
    Also~\mbox{$P_2 := \underline{X(f(f(v)))}$} is purified in $N_2 := \emptyset$ since there are no $P_2$-resolvable clauses in $N_2$.
    Thus $D$ is a $\mathcal{C}$-derivation and it is $X$-eliminating from $N := \mset{\neg X(v) \lor X(f(v)), X(f(f(v)))}$.
    By the correctness of SCAN we also get $\Exists{X} N \iff \top$.

    We now prove that there exists no acyclic purification annotation for $D$
    by showing for all~\mbox{$k \in \N$} that $P_1$ is not $k$-acyclically purified in $N_1$.
    Note that $s_1$ defined by~\mbox{$s_1(\underline{X(f(f(v)))}) = X(f(f(v)))$} is the only $P_1$-purification subsumption for $N_1$.
    The corresponding purification subsumption graph $G(N_1, P_1, s_1)$ is
    \begin{center}
    \begin{tikzpicture}[>=stealth, node distance=3cm,
    every node/.style={draw=none},
    ]
      \node (1) {$X(f(f(v))),$};

      \draw (1) edge[loop above] node[font=\footnotesize] {$X(f(f(v)))$} (1);
    \end{tikzpicture}
  \end{center}
  which contains a cycle.
  Thus $P_1$ is not $k$-acyclically purified in $N_1$ for any $k \in \N$.
  
  Still, one might suspect that $\tau_{P_1}^{k}$ is witness-preserving across $\mathrm{PurDel}_{P_1}$ for some~\mbox{$k \in \N$}.
  This is not the case though as we show now.
  We first show $B_{P_1}^k \iff \lambda c. \bot$ for all $k \in \N$ by induction.
  For $k=0$ this follows by definition of $B_{P_1}^0$.
  For the induction step note by definition we get $B_{P_1}^{k+1} \iff \lambda c. X(c) \land \Forall{v}(c \noeq v \lor B_{P_1}^k(v))$.
  By induction hypothesis we get~\mbox{$B_{P_1}^{k+1} \iff \lambda c. X(u) \land \Forall{v}(c \noeq v) \iff \bot$}.
  Now we have $\tau_{P_1}^k \iff [X \leftarrow \lambda u. \bot]$.
  Let $\sigma$ be any witness for $\Exists{X} N_1$.
  Then $\tau_{P_1}^k\sigma \iff [X \leftarrow \lambda u. \bot]$ and $N[X \leftarrow \lambda u. \bot] \iff \bot$.
  However, since $\Exists{X} N \iff \top$ this shows that $[X \leftarrow \lambda u. \bot]$ is not a witness for $\Exists{X} N$.

  Note that there are other $X$-eliminating $\mathcal{C}$-derivations from $N$ for which our method can produce a witness.
  For example,
  $$
    D_1 = \AxiomC{$\mset{\neg X(v) \lor X(f(v)), X(f(f(v)))}$}
      \RightLabel{$\mathrm{ExtPurDel}_{X}^+$}
      \UnaryInfC{$\emptyset$}
      \DisplayProof
    $$
  for which the empty function $\emptyset$ is an acyclic purification annotation.
  The corresponding witness is $\sigma_{\mathrm{fo}}(D, \emptyset) = [X \leftarrow \lambda u. \top]$.
\end{example}

\section{Equality and First-Order Background Theories}
\label{sec.first-order-background-theories}

In this section we show how to use equality reasoning and treat first-order background theories with respect to WSCAN.

\subsection{Equality Reasoning}
So far the inference calculus $\mathcal{I}$ uses equality only for introducing constraints and eliminating them via constraint elimination.
However, our method can also be augmented to construct witnesses in the presence of full equality by adding a paramodulation rule.
$$
  \AxiomC{$C \lor s \oeq t$}
  \AxiomC{$C'(r)$}
  \RightLabel{$\mathrm{ParMod}$}
  \BinaryInfC{$(C \lor C'(t))\sigma$}
  \DisplayProof
$$
where $s, t, r$ are terms and $\sigma$ is a most general unifier of $s$ and $r$.
This makes the inference system refutationally complete for first-order logic with equality.
\begin{proposition}
  $\mathcal{I} + \mathrm{ParMod}$ is sound and refutationally complete for first-order logic with equality.
\end{proposition}
\begin{proof}
  Resolution with factoring together with paramodulation and reflexivity resolution is refutationally complete, see, e.g., \cite{Nieuwenhuis2001}. 
  So it suffices to show that we can simulate the reflexivity resolution rule
  $$
  \AxiomC{$t \noeq t \lor C$}
  \UnaryInfC{$C$}
  \DisplayProof
  $$
  Note that this is an instance of the constraint elimination rule since the identity is a most general unifier of $t$ and $t$.
\end{proof}

For any sound inference rule $S$ added to $\mathcal{I}$ we can show that the identity substitution is witness-transforming across $S$ like in the case of constraint resolution, constraint factoring and constraint elimination.
Thus, by setting $\tau_{\mathrm{ParMod}} := \mathrm{id}$ the witness construction proceeds like before.
However, adding a new inference rule does increase the search space for finding $\mathcal{C}$-derivations so this affects performance.

We illustrate this on an example.
\begin{example}
  Let $N$ be the clause set consisting of the clauses
  \begin{align*}
  (1)&\ a \oeq b, \\
  (2)&\ \neg E(a,v), \\
  (3)&\ \neg X(u) \lor \neg E(u,v) \lor X(v), \\
  (4)&\ X(b).
  \end{align*}
  Note that $\underline{X(b)}$ is not purified in $N \setminus \mset{X(b)}$ since the resolvent of $\underline{X(b)}$ with pointed clause~$3.1$ is $\neg E(b, v) \lor X(v)$ which does not satisfy $N \setminus \mset{X(b)} \subsumesLvelim{\neg X(b)} \neg E(b,v) \lor X(v)$.
  However, since $a\oeq b, \neg E(a,v) \in N$ we can derive clause $(5)$ $\neg E(b,v)$ by paramodulation:
  $$
    \AxiomC{$a = b$}
    \AxiomC{$\neg E(a,v)$}
    \RightLabel{$\mathrm{ParMod}.$}
    \BinaryInfC{$\neg E(b,v)$}
    \DisplayProof
  $$
  Then $\underline{X(b)}$ is purified in $\mset{1, 2, 3, 5}$.
  Now consider the following $\mathcal{C}$-derivation (also allowing $\mathrm{ParMod}$-inferences)
  $$
    D=
    \AxiomC{$\mset{1, 2, 3, 4}$}
    \RightLabel{$\mathrm{ParMod}$}
    \UnaryInfC{$\mset{1, 2, 3, 4, 5}$}
    \RightLabel{$\mathrm{PurDel}_{\underline{X(b)}}$}
    \UnaryInfC{$\mset{1, 2, 4, 5}$}
    \RightLabel{$\mathrm{ExtPurDel}_X^{-}.$}
    \UnaryInfC{$\mset{1, 2, 5}$}
    \DisplayProof
  $$
  The only $\underline{X(b)}$-purification subsumption for $\mset{1, 2, 3, 5}$ is given by 
  $$s(\underline{\neg X(u)} \lor \neg E(u,v) \lor X(v)) = \neg E(b,v).$$
  Thus $G(\mset{1,2,3,5}, \underline{X(b)}, s)$ is acyclic and has maximum path length $1$.
  Furthermore we have~\mbox{$B_{\underline{X(b)}}^1 \iff \lambda u. \neg X(u) \land u \noeq b$}.
  Note that $I(D) = \mset{2}$ and the function $f: I(D) \to \N$ with~\mbox{$f(2) = 1$} is an acyclic purification annotation for $D$.
  We then have
  \begin{align*}
    \sigma_{\mathrm{fo}}(D, f) &= \tau_{\mathrm{ParMod}}\tau_{\mathrm{PurDel}_{\underline{X(b)}}}^1\tau_{\mathrm{ExtPurDel}_X^{-}}  \\
    &= \mathrm{id}[X \leftarrow \lambda u. \neg B_{\underline{X(b)}}^1(u)][X \leftarrow \lambda u. \bot] \\
    &\iff [X \leftarrow \lambda u. X(u) \lor u \oeq b][X \leftarrow \lambda u. \bot] \\
    &\iff [X \leftarrow \lambda u. u \oeq b].
  \end{align*}

  In this example SCAN would not have been able to produce this derivation, if we did not have equality reasoning since purification of $\underline{X(b)}$ would have added the resolvent~\mbox{$\neg E(b,v) \lor X(v)$} to the clause set which would then have to be resolved further with clause $(3)$ resulting in clauses that have an increasing number of $\neg E$-literals.
\end{example}

\subsection{First-Order Background Theories}

Let $T$ be a first-order theory, i.e., a set of closed first-order formulas.
We consider WSOQE with respect to the background theory $T$, i.e., 
the \emph{WSOQE problem with respect to $T$} is the problem of, given a first-order formula $\varphi$ and a tuple of predicate variables $\seq{X}$, finding a substitution $\sigma$ with $\mathrm{dom}(\sigma) \subseteq \seq{X}$ such that
$$
T \models \Exists{\seq{X}} \varphi \liff \varphi\sigma.
$$
In that case we call $\sigma$ a \emph{WSOQE-witness for $\Exists{\seq{X}} \varphi$ with respect to $T$}.
For finitely axiomatizable theories $T$ we can reduce this problem to the ordinary WSOQE problem.
\begin{proposition}
  \label{reduction-of-background-theory-solutions}
  Let $T$ be a finitely axiomatizable first-order theory, let $\varphi$ be a first-order formula and let $\sigma$ be a substitution with $\mathrm{dom}(\sigma) \subseteq \seq{X}$.
  Then the following are equivalent.
  \begin{enumerate}
  \item \label{reduction.theory-witness} $\sigma$ is a WSOQE-witness for $\Exists{\seq{X}} \varphi$ with respect to $T$.
  \item \label{reduction.and-witness} $\sigma$ is a WSOQE-witness for $\Exists{\seq{X}}(T \land \varphi)$.
  \end{enumerate}
\end{proposition}
\begin{proof}
  \ref{reduction.theory-witness} $\implies$ \ref{reduction.and-witness}: It suffices to show $\models \Exists{\seq{X}} (T \land \varphi) \limp (T \land \varphi)\sigma$.
  Since $\sigma$ is a WSOQE-witness for $\Exists{\seq{X}} \varphi$ with respect to $T$ we have $T \models \Exists{\seq{X}} \varphi \liff \varphi\sigma$.
  In particular, we get $T \models \Exists{\seq{X}} \varphi \limp \varphi\sigma$.
  By the deduction theorem and since $T$ is finite we get $\models T \limp (\Exists{\seq{X}} \varphi \limp \varphi\sigma)$.
  By propositional equivalences we then get $\models (T \land \Exists{\seq{X}} \varphi) \limp (T \land \varphi\sigma)$.
  Since $T$ does not contain any variables from $\seq{X}$ we further get $\models \Exists{\seq{X}}(T \land \varphi) \limp (T \land \varphi)\sigma$.

  \ref{reduction.and-witness} $\implies$ \ref{reduction.theory-witness}:
  It suffices to show $T \models \Exists{\seq{X}} \varphi \limp \varphi\sigma$.
  Since $\sigma$ is a WSOQE-witness for $\Exists{\seq{X}}(T \land \varphi)$ we get in particular $\models \Exists{\seq{X}}(T \land \varphi) \limp (T \land \varphi)\sigma$.
  Since $T$ does not contain any variables from~$\seq{X}$ this is equivalent to $\models (T \land \Exists{\seq{X}}\varphi) \limp (T \land \varphi\sigma)$.
  Thus we get $\models (T \land \Exists{\seq{X}}\varphi) \limp \varphi\sigma$ and by propositional equivalences we have
  $\models T \limp (\Exists{\seq{X}}\varphi \limp \varphi\sigma)$.
  Since $T$ consists of closed formulas we then get $T \models \Exists{\seq{X}}\varphi \limp \varphi\sigma$.\qedhere  
\end{proof}

\subsection{Graph reachability problems}
We introduce a class of background theories where WSCAN can find solutions using equality reasoning although finding a derivation can be challenging in practice.
Let~$G=(V,E)$ be a finite directed graph and $I, F \subseteq V$.
We want to find a subset~\mbox{$X \subseteq V$} of nodes which contains $I$ (the initial nodes), is closed under the edge relation of $G$ and avoids all nodes in $F$ (the fail nodes).
In this case we call $X$ a \emph{graph reachability solution for $(G, I, F)$}.

We can encode graph reachability problems in the following way.
Consider the set of vertices~\mbox{$V = \mset{1, \dots, n}$} and let $a_1, \dots, a_n$ be distinct fresh constant symbols standing for the nodes in the graph.
Furthermore, let $E$ be a binary predicate symbol denoting the edge relation of the graph $G$.
Then the theory of the graph $G$ can be axiomatized by a first-order clause set $T(G)$ given by
\begin{align*}
  &\mset{ a_i \noeq a_j \suchthat i \neq j \in V} \\
  \union &\mset{E(a_i, a_j) \suchthat \text{$i$ and $j$ are connected by an edge in $G$}} \\
  \union &\mset{\neg E(a_i,a_j) \suchthat \text{$i$ and $j$ are not connected by an edge in $G$}} \\
  \union &\mset{u \oeq a_1 \lor \dots \lor u \oeq a_n}. 
\end{align*}
To be a graph reachability solution, $X$ has to satisfy the following clause set.
\begin{align*}
  R(I,F) := \mset{X(a_i) \suchthat i \in I} \union \mset{\neg X(a_i) \suchthat i \in F} \union \mset{\neg X(u) \lor \neg E(u,v) \lor X(v)}.
\end{align*}
Then $(G, I, F)$ has a graph reachability solution if and only if $T(G) \models \Exists{X} R(I,F)$.
Furthermore, if $(G, I, F)$ has a graph reachability solution, then a $\WSOQE$-witness for~$\Exists{X} R(I,F)$ with respect to $T(G)$ gives us a first-order definition of a graph reachability solution for $(G, I, F)$.
By \Cref{reduction-of-background-theory-solutions} such a solution is a solution to the ordinary WSOQE problem for $\Exists{\seq{X}}(T(G) \land R(I,F))$.
We show how our method can be used to compute a graph reachability solution on a simple example.
This example also shows that paramodulation inferences are sometimes necessary for purified clause deletion to be possible.
\begin{example}
Consider a simple graph reachability problem. Let $G$ be the graph consisting of the three nodes $1, 2, 3$ where the only edge is from $1$ to $2$, the only initial node is $1$ (indicated by the double circle) and the only fail node is $3$ (indicated by the dashed circle):
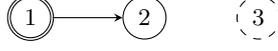
\begin{center}
  \begin{tikzpicture}[>=stealth, node distance=1.5cm,
  every edge/.style={->, thick}]

    \node[draw, circle, double] (1) {$1$};
    \node[draw, circle] (2) [right of=1] {$2$};
    \node[draw, circle, dashed] (3) [right of=2] {$3$};

    \draw[->] (1) to (2);
  \end{tikzpicture}
\end{center}
The only graph reachability solution is the set $\mset{1, 2}$.
We now give an $X$-eliminating $\mathcal{C}$-derivation from $T(G) \union R(\mset{1}, \mset{3})$ which uses several paramodulation inferences.
First, note that $\neg E(a_2,v)$ is derivable by paramodulation inferences from $T(G)$ via
\begin{prooftree}
  \AxiomC{$v \oeq a_1 \lor v \oeq a_2 \lor  v \oeq a_3$}
  \AxiomC{$\neg E(a_2,a_1)$}
  \BinaryInfC{$v \oeq a_2 \lor v \oeq a_3 \lor \neg E(a_2, v)$}
  \AxiomC{$\neg E(a_2,a_2)$}
  \BinaryInfC{$v \oeq a_3 \lor \neg E(a_2, v)$}
  \AxiomC{$\neg E(a_2, a_3)$}
  \BinaryInfC{$\neg E(a_2,v)$}
\end{prooftree}
From $T(G)$ and $\neg E(a_2,v)$ we can further derive $\neg E(a_1,u) \lor \neg E(u,v)$ via
\begin{prooftree}
  \AxiomC{$u \oeq a_1 \lor u \oeq a_2 \lor u \oeq a_3$}
  \AxiomC{$\neg E(a_2,v)$}
  \BinaryInfC{$u \oeq a_1 \lor u \oeq a_3 \lor \neg E(u, v)$}
  \AxiomC{$\neg E(a_1,a_1)$}
  \BinaryInfC{$u \oeq a_3 \lor \neg E(a_1, u) \lor \neg E(u,v)$}
  \AxiomC{$\neg E(a_1, a_3)$}
  \BinaryInfC{$\neg E(a_1, u) \lor \neg E(u,v)$}
\end{prooftree}
Denote by $N^{+}$ the clause set $T(G)$ together with the derived clauses given in the above two derivations.
The clause set $R(\mset{1}, \mset{3})$ is given by
\begin{align*}
  (1)\ &X(a_1), \\
  (2)\ &\neg X(a_3), \\
  (3)\ &\neg X(u) \lor \neg E(u,v) \lor X(v). \\
  \intertext{The resolvent between $1$ and $3$ after variable elimination is given by}
  (4)\ &\neg E(a_1, v) \lor X(v)
\end{align*}
Now consider the following $\mathcal{C}$-derivation from $T(G) \union R(\mset{1}, \mset{3})$.
$$
D=
\Axiom$T(G) \union \mset{1, 2, 3}\fCenter$
\RightLabel{$\mathrm{ParMod}^\ast$}
\UnaryInf$N^+ \union \mset{1, 2, 3}\fCenter$
\RightLabel{$\mathrm{Res}_{1.1,3.1}$}
\UnaryInf$N^+ \union \mset{1, 2, 3, 4}\fCenter$
\RightLabel{$\mathrm{PurDel}_{1.1}$}
\UnaryInf$N^+ \union \mset{2, 3, 4}\fCenter$
\RightLabel{$\mathrm{PurDel}_{4.2}$}
\UnaryInf$N^+ \union \mset{2, 3}\fCenter$
\RightLabel{$\mathrm{ExtPurDel}_X^{-}.$}
\UnaryInf$N^+\fCenter$
\DisplayProof
$$
From a semantic perspective it is not necessary to add the first-order clauses~$\neg E(a_2, v)$ and~\mbox{$\neg E(a_1, u) \lor \neg E(u,v)$} to the clause set as they are semantically entailed by $T(G)$ which does not contain $X$, i.e., $N^+$ is logically equivalent to $T(G)$.
However, they were necessary to ensure that $4.2$ is purified in the last $\mathrm{PurDel}_{4.2}$ step.

We now compute the witness corresponding to the derivation $D$.
Note that $D$ has length~$10$: There are $6$ paramodulation steps, one resolution step, two purified clause deletion steps and one extended purity deletion step.
Then $I(D) = \mset{8, 9}$.
We now look for an acyclic purification annotation $f$ for $D$.

Consider the purified clause deletion step $\mathrm{PurDel}_{1.1}$ with conclusion $N_8 := N^{+} \union \mset{2, 3, 4}$.
Note that the $1.1$-resolvable pointed clauses in $N_8$ are $R_{1.1}(N_8) = \mset{2.1, 3.1}$.
The resolvent between $1.1$ and $2.1$ is $a_1 \noeq a_3$ which is subsumed by $T(G)$.
The resolvent between $1.1$ and~$3.1$ is $a_1 \noeq u \lor \neg E(u,v) \lor X(v)$ which after variable elimination reduces to $\neg E(a_1,v) \lor X(v)$ which is subsumed by $4$.
Thus, a $1.1$-purification subsumption $s_8$ for $N_8$ is given by
\begin{align*}
  s_8(2.1) &= a_1 \noeq a_3 \in T(G) \\
  s_8(3.1) &= \neg E(a_1, v) \lor X(v).
\end{align*}
The corresponding purification subsumption graph is
\begin{center}
  \begin{tikzpicture}[>=stealth,
  every node/.style={draw=none},
  every edge/.style={->, thick},
  every edge quotes/.style={draw=none, fill=none, font=\footnotesize, inner sep=0pt}]

    \node (1) {$\neg X(a_3)$};
    \node (2) [right=3cm of 1] {$a_1 \noeq a_3$};
    \node (3) [below=0.5cm of 1] {$\neg X(u) \lor E(u,v) \lor X(v)$};
    \node (4) [below=0.5cm of 2] {$\neg E(a_1, v) \lor X(v)$};
    \node (5) [left=of 1] {$N^{+} \setminus \mset{a_1 \noeq a_3}$};

    \draw[->] (1) -- node[above, draw=none, inner sep=0pt,font=\footnotesize] {$\neg X(a_3)$} (2);
    \draw[->] (3) -- node[above, draw=none, inner sep=0pt,font=\footnotesize] {$\neg X(u)$} (4);
  \end{tikzpicture}
\end{center}
which is acyclic and has maximal length $1$.
Thus we set $f(8) = 1$.

Now consider the derivation step $\mathrm{PurDel}_{4.2}$ with conclusion~\mbox{$N_9 := N^{+} \union \mset{2, 3}$}.
Note that the $4.2$-resolvable clauses with $N_9$ are~\mbox{$R_{4.2}(N_9) = \mset{2.1, 3.1}$}.
The resolvent between $4.2$ and~$2.1$ is $\neg E(a_1, v) \lor v \noeq a_3$ which after variable elimination reduces to $\neg E(a_1, a_3)$ and which is subsumed by $T(G)$.
The resolvent between $4.2$ and $3.1$ is $u \noeq v' \lor \neg E(a_1, v') \lor \neg E(u,v) \lor X(v)$ which after variable elimination reduces to $\neg E(a_1,u) \lor \neg E(u, v) \lor X(v)$ which is subsumed by clause~\mbox{$\neg E(a_1, u) \lor E(u,v) \in N^+$}.
Thus, a $4.2$-purification subsumption $s_9$ for $N_9$ is given by
\begin{align*}
  s_9(2.1) &= \neg E(a_1, a_3) \in T(G) \\
  s_9(3.1) &= \neg E(a_1, u) \lor E(u,v) \in N^{+}.
\end{align*}
The corresponding purification subsumption graph is
\begin{center}
  \begin{tikzpicture}[>=stealth,
  every node/.style={draw=none},
  every edge/.style={->, thick},
  every edge quotes/.style={draw=none, fill=none, font=\footnotesize, inner sep=0pt}]

    \node (1) {$\neg X(a_3)$};
    \node (2) [right=2.5cm of 1] {$\neg E(a_1, a_3)$};
    \node (3) [below=0.5cm of 1] {$\neg X(u) \lor E(u,v) \lor X(v)$};
    \node (4) [below=0.5cm of 2] {$\neg E(a_1, u) \lor E(u,v)$};
    \node (5) [left=of 1] {$N^{+} \setminus \mset{\neg E(a_1, a_3), \neg E(a_1, u) \lor E(u,v)}$};

    \draw[->] (1) -- node[above, draw=none, inner sep=0pt,font=\footnotesize] {$\neg X(a_3)$} (2);
    \draw[->] (3) -- node[above, draw=none, inner sep=0pt,font=\footnotesize] {$\neg X(u)$} (4);
  \end{tikzpicture}
\end{center}
which is acyclic and has maximal length $1$.
Thus we set $f(9) = 1$.

We can now compute the corresponding witness as 
\begin{align*}
  \sigma(D, f) &= \tau_{\mathrm{PurDel}_{1.1}}^1\tau_{\mathrm{PurDel}_{4.2}}^1\tau_{\mathrm{ExtPurDel}_X^{-}} \\
  &= [X \leftarrow \neg B_{1.1}^1][X\leftarrow \neg B_{4.2}^1][X \leftarrow \lambda u. \bot] 
\end{align*}
Note that 
\begin{align*}
  B_{1.1}^1 &\iff \lambda u. \neg X(u) \land u \noeq a_1 \text{ and} \\
  B_{4.2}^1 &\iff \lambda u. \neg X(u) \land \Forall{v}(u \noeq v \lor \neg E(a_1, v)) \iff \lambda u. \neg X(u) \land \neg E(a_1, u).
\end{align*}
Therefore we get
\begin{align*}
  \sigma(D,f) &\iff [X \leftarrow \lambda u. X(u) \lor u \oeq a_1][X \leftarrow \lambda u. X(u) \lor E(a_1, u)][X \leftarrow \lambda u. \bot] \\
  &\iff [X \leftarrow \lambda u. X(u) \lor u \oeq a_1][X \leftarrow \lambda u. E(a_1, u)] \\
  &\iff [X \leftarrow \lambda u. E(a_1, u) \lor u \oeq a_1]
\end{align*}

Note that $T(G) \models E(a_1, u) \lor u \oeq a_1 \liff u \oeq a_1 \lor u \oeq a_2$ so the witness matches the graph reachability solution $\mset{1, 2}$.
\end{example}

\section{Implementation}
\label{sec.implementation}

The presented method is implemented in the Scala programming language~\cite{scala} as part of the General Architecture for Proof Theory (GAPT) software package~\cite{GAPT}, which provides useful data structures for developing algorithms in computational logic, including terms, formulas and clauses.
The prototype is available in version 2.19.0 of GAPT (\texttt{https://www.logic.at/gapt/release\_archive.html}).
Section 8.4 of the GAPT user manual details how to use the prototype.

To evaluate our witness construction method, we implemented the SCAN algorithm so that, in addition to finding a logically equivalent first-order clause set, it provides a~$\mathcal{C}$-derivation from which our method extracts a witness.

In the following, let $\seq{X}$ be the tuple of predicate variables to be eliminated.
Our implementation operates in two alternating phases.
First, we run a preprocessing phase where the clause set is simplified by eagerly applying variable elimination, subsumption deletion, tautology deletion and, if applicable, extended purity deletion.
As part of this phase we also add non-redundant constraint factors of all clauses in the clause set.

The second phase is the purification phase, where we pick a clause $C$ from the active clause set and an $\seq{X}$-literal $L \in C$ which induces the pointed clause $P = C[\underline{L}]$.
Denote by $N$ the rest of the active clause set.
We repeatedly add those constraint resolvents~$R$ between $P$ and $N$ such that $N \not\subsumesLvelim{L^\perp} R$.
During this phase we also eagerly apply variable elimination and subsumption deletion, but to avoid the issues of \Cref{ex:necessary-to-treat-tautological-clause-not-as-redundant} we do not perform tautology deletion.
This phase runs until the pointed clause $P$ is purified in the rest of the clause set, at which point it is removed by purified clause deletion.

Afterwards, the two phases are repeated until no more $\seq{X}$-literals occur in the active clause set.
During the saturation process all performed derivation steps are stored in a $\mathcal{C}$-derivation which is later used to construct a witness.

The choice of pointed clause during the purification phase determines the outcome of the saturation process.
Different choices can lead to termination or non-termination.
This also affects which $\mathcal{C}$-derivations are found.
To find all possible $\mathcal{C}$-derivations from a given clause set,
we implemented backtracking on top of the saturation process, which yields $\mathcal{C}$-derivations based on different choices of pointed clauses.
This backtracking can be performed on-demand and allows us to find multiple derivations and therefore multiple witnesses.

After an $\seq{X}$-eliminating $\mathcal{C}$-derivation is found, we compute a corresponding witness according to \Cref{def:first-order-witness}.
For the derivation steps $S$ which are not purified clause deletion the computation of $\tau_S$ is straightforward.
For purified clause deletion steps~\mbox{$N \union \mset{P} / N$} we calculate the minimal $k \in \N$ such that $P$ is $k$-acyclically purified in $N$ by iterating all $P$-purification subsumptions $s$ for $N$ and computing a maximal path length of $G(N,P,s)$ while also checking for acyclicity of $G(N,P,s)$.
If $P$ is not $k$-acyclically purified in $N$ for any $k \in \N$, we resort to computing $\mathrm{\ell Res}_P$ up to redundancy using the purification process from the SCAN implementation and then use the predicate substitution $\tau_{\mathrm{PurDel}_P}$ from \Cref{def:witness-transforming-purified-clause-deletion}.
This might not terminate, but there are cases where $\mathrm{\ell Res}_P$ is finite up to redundancy while $P$ is not $k$-acyclically redundant in $N$.
Since $\mathrm{\ell Res}_P$ is potentially infinite, we provide parameters that limit the number of derivation steps performed in its computation.

Our prototype implementation has been tested on 44 examples from the file
\texttt{examples/predicateEliminationProblems.scala} included in GAPT 2.19.0.
These examples were either created by us or selected from the literature.
The tests were run using the OpenJDK
21.0.4 Java Runtime with 16GB heap space on a MacBook Pro M2 with 32GB RAM.
We first run our implementation of the SCAN algorithm to get a derivation and measure the execution time and the length of the derivation.
We limit the search to derivations that eliminate the predicates in at most 50 derivation steps.
Furthermore, for each example we set a timeout of 10 seconds.
If a derivation could be computed, we run the witness computation based on that derivation and measure its execution time separately.
Furthermore, we measure input clause set size and the resulting witness size (if found).
We measure the size of a literal by the number of non-logical symbols it contains.
The size of a clause is given by the sum of the sizes of its literals and the size of a clause set is given by the sum of the sizes of its clauses.
The size of a witness is measured by the number of logical and non-logical symbols it contains.
See the \texttt{runWscanBenchmark} method in the file \texttt{testing/src/main/scala/testWscan.scala} for the exact test protocol.

The aggregate results are as follows: The input size ranged from 2 to 177 with an average of 17.705.
For 40 of the 44 examples (90.91\%) a derivation was found.
Their lengths ranged from 0 (this example had no second-order quantifiers) to 34 with an average of 7.975.
In the remaining 4 cases the timeout of 10 seconds was reached without finding a derivation.
For the 40 terminating examples the running times for SCAN ranged from 0.119 milliseconds to 6181.111
milliseconds with an average of~375.435~milliseconds.
For all of the 40 examples where SCAN terminated, a witness could be computed as well.
The running times for the witness computation given the derivation ranged from 0.009 milliseconds to 4.684 milliseconds with an average of~0.619~milliseconds.
The corresponding witness sizes ranged from 0 (this example had no second-order quantifiers) to 54 with an average of 6.900.
This shows that witness computation takes significantly less time than finding the SCAN derivation itself.

We have not optimized our implementation for large clause sets.
Furthermore, we are not aware of a sufficiently large database of examples to test our method on.
Thus, a more thorough empirical evaluation is left as future work.

\section{Discussion and Future Work}
\label{sec.discussion}

In this section we discuss several aspects of our work and possible directions for future work.

\subsection*{First-Order Witnesses for $\seq{X}$-Eliminating Derivations}

We have presented two methods for witness construction: one that produces a fixpoint witness for arbitrary $\seq{X}$-eliminating $\mathcal{C}$-derivations, and one that produces a first-order witness for $\seq{X}$-eliminating $\mathcal{C}$-derivations with acyclic purification annotations. We conjecture that it is possible to produce first-order witnesses whenever there is an $\seq{X}$-eliminating $\mathcal{C}$-derivation from a clause set $N$ without any further assumptions.
\begin{conjecture}
\label{first-order-witness-problem}
Let $N$ be a clause set.
If there is an $\seq{X}$-eliminating $\mathcal{C}$-derivation from $N$, then there is a first-order \WSOQE-witness for $\Exists{\seq{X}} N$.
\end{conjecture}

\subsection*{Recursion-Free Horn Clauses}
Program verification problems are often encoded as constrained Horn clauses \cite{Bjorner15Horn} where unknown predicates $X$ stand for loop invariants necessary to prove the correctness of a given program.
Solving constrained Horn clauses means finding an interpretation of the unknown predicates $X$.
This is a restriction of the \FEQ-problem: Given a formula $\varphi$ with predicate variables in $\seq{X}$, find a first-order predicate substitution $\sigma$ such that $\models \varphi\sigma$.
Solving constrained Horn clauses then means solving \FEQ\ where $\varphi$ is a constrained Horn clause set.
A particular fragment of constrained Horn clauses, namely \emph{recursion-free} Horn clause sets, has been widely studied, see, e.g., \cite{Ruemmer2015,Unno2015}.
\begin{definition}
  Let $N$ be a clause set and let $X, Y$ be predicate variables.
  We write $X \prec_N Y$ if there is a clause $C \in N$ such that $X$ occurs negatively in $C$ and $Y$ occurs positively in~$C$.
  Denote by $\prec_N^{+}$ the transitive closure of $\prec_N$.
  For a tuple of predicate variables $\seq{X}$ we say~\emph{$N$ is recursion-free in $\seq{X}$} if $\prec_N$ is acyclic on $\seq{X}$, i.e., there is no predicate variable $Y \in \seq{X}$ with~\mbox{$Y \prec_N^{+} Y$}.
\end{definition}

We suspect that SCAN terminates on this class of clause sets and that our method can produce a \WSOQE-witness.
\begin{conjecture}
  Let $N$ be a Horn clause and let $\seq{X}$ be a tuple of predicate variables such that~$N$ is recursion-free in $\seq{X}$.
  Then there is an $\seq{X}$-eliminating $\mathcal{C}$-derivation from $N$ and an acyclic purification annotation $f$ such that $\sigma_{\mathrm{fo}}(D, f)$ is a \WSOQE-witness for $\Exists{\seq{X}} N$.
  If~$\models N'$, then $\sigma_{\mathrm{fo}}(D, f)$ is an \FEQ-witness for $\Exists{\seq{X}} N$.
\end{conjecture}
One can then check whether the produced \WSOQE-witness is an \FEQ-witness by checking whether the first-order conclusion of the produced $\mathcal{C}$-derivation is valid using, e.g., a first-order theorem prover.
This conjecture would give a procedure for solving recursion-free constrained Horn clause sets different from those in \cite{Ruemmer2015,Unno2015}.
It may be possible to extend this procedure based on redundancy criteria from resolution theorem proving in order to extend the class where it terminates beyond that of recursion-free constrained Horn clause sets.

\subsection*{Constructing Multiple Witnesses}
Different $\mathcal{C}$-derivations can result in different, non-equivalent witnesses. Thus, our method can construct multiple witnesses for a given input clause set.
As mentioned above, our prototype obtains different $\mathcal{C}$-derivations by backtracking over different choices of pointed clauses in the purification process of SCAN.

In practice, certain witnesses might be preferred over others, e.g., based on size.
Accounting for additional constraints like these when constructing witnesses is an avenue for future work.
An additional avenue is to analyze the relative logical strength of different witnesses.

\subsection*{Limitations for Finding Witnesses}
\label{sec.limitations-for-finding-witnesses}
The original SCAN algorithm uses Skolemization if the input formula is not in clause form and after the saturation process tries to undo the Skolemization to find a logically equivalent first-order formula.
While this reverse Skolemization is not possible in general, it does work in some cases, for example, if the input formula has the form $\Exists{\seq{u}} \varphi$ where $\varphi$ is quantifier-free.
We discuss such an example where the SCAN algorithm terminates and finds a SOQE-solution, but where no WSOQE-witness exists.

The input formula is $\Phi = \Exists{X} \Exists{u} \Exists{v} ( X(u) \land \neg X(v))$.
Skolemization of the first-order part yields the clause set $N = \mset{X(a), \neg X(b)}$ with Skolem constants~$a$ and~$b$.
An $X$-eliminating $\mathcal{C}$-derivation from $N$ is given by
$$
  D = \Axiom$\mset{X(a), \neg X(b)}\fCenter$
  \RightLabel{$\mathrm{Res}_{\underline{X(a)}, \underline{\neg X(b)}}$}
  \UnaryInf$\mset{X(a), \neg X(b), a \noeq b}\fCenter$
  \RightLabel{$\mathrm{PurDel}_{\underline{X(a)}}$}
  \UnaryInf$\mset{\neg X(b), a \noeq b}\fCenter$
  \RightLabel{$\mathrm{ExtPurDel}_X^{-}$}
  \UnaryInf$\mset{a \noeq b}\fCenter$
  \DisplayProof
$$
The original SCAN algorithm now reverses the Skolemization on the resulting formula~$a \noeq b$, to get the logically equivalent formula $\Exists{u} \Exists {v} u \noeq v$, i.e., $\Phi \iff \Exists{u} \Exists {v} u \noeq v$.
Let us now compute a witness.
We use the method from \Cref{sec.resolution-witnesses} as it also produces a first-order witness in this case as $\underline{X(a)}$ is a one-sided pointed clause.
Note that~\mbox{$\mathrm{\ell Res}_{\underline{X(a)}} = \lambda u.\neg X(u) \land u \noeq a$} and thus $\neg \mathrm{\ell Res}_{\underline{X(a)}} \iff \lambda u. X(u) \lor u \oeq a$.
Now a witness is given by 
$$\sigma(D) = [X \leftarrow \lambda u. X(u) \lor u \oeq a][X \leftarrow \lambda u. \bot] \iff [X \leftarrow \lambda u. u \oeq a].$$
However, this witness contains a Skolem constant which is not in the language of the input formula.
In fact, no WSOQE-witness exists in the input language:
Assume $\Phi$ had a WSOQE-witness $\witness = \lambda u. \varphi(u)$.
Then $\witness$ is a first-order predicate in the empty language and in every structure $\mathcal{M}$ it defines the subset $\witness^{\mathcal{M}}$ given by~\mbox{$\mset{m \in M \suchthat \mathcal{M} \models \varphi(u)}$}.
For every two-element model~\mbox{$\mathcal{M} = \mset{m_1, m_2}$}
we have
$$\mathcal{M} \models \Exists{X} \Exists{u} \Exists{v} (X(u) \land \neg X(v)).$$
Since $\witness$ is a witness we get~\mbox{$\mathcal{M} \models \Exists{u} \Exists{v}(\varphi(u) \land \neg \varphi(v))$}.
Thus $\witness^{\mathcal{M}}$ is neither the empty set nor the whole set $M$.
However, in the empty language, the only first-order parameter-free definable sets in a structure
$\mathcal{M}$ are the empty set and $M$ itself, see~\cite[Theorem 2.1.2]{Hodges97Shorter}, which is a contradiction.

This shows a fundamental limitation for finding witnesses, even if SOQE-solutions exist.
For practical applications it would be useful to add Skolemization and reverse Skolemization steps to the implementation and investigate classes of formulas where witness construction succeeds in the presence of Skolemization.

\subsection*{Quantifier Alternations}

We have discussed the WSOQE problem in the context of existential quantifiers, but one can also consider the dual problem for universal quantifiers:
Given a formula $\Forall{\seq{X}} \varphi$ with first-order~$\varphi$, find first-order predicates $\seq{\witness}$ such that
\begin{equation*}
  \Forall{\seq{X}} \varphi \iff \varphi[\seq{X} \leftarrow \seq{\witness}].
\end{equation*}
One can reduce this problem to the existential case since for all first-order predicates~\mbox{$\seq{\witness}$} we have $\Forall{\seq{X}} \varphi \iff \varphi[\seq{X} \leftarrow \seq{\witness}]$ if and only if $\Exists{\seq{X}} \neg \varphi \iff \neg \varphi[\seq{X} \leftarrow \seq{\witness}]$ by writing $\Forall{\seq{X}}$ as $\neg \Exists{\seq{X}} \neg$.
This way, one can solve the problem for an arbitrary quantifier prefix by successively eliminating alternating blocks of existential and universal quantifiers starting with the innermost block.

While the reduction above works if one has a general WSOQE method that can handle any input formula,
our method only
works when the input is a clause set,
i.e., the domain variables are universally
quantified.
The above reduction introduces a negation on the input formula,
meaning that we would now have to deal with existential quantifiers on
the domain variables as well, e.g., via Skolemization.
However, as stated before there are fundamental limitations when introducing Skolemization.
For input formulas where the first-order part is quantifier-free our method is applicable though.

\subsection*{Improvement over Ackermann's Lemma}

Recall Ackermann's Lemma:
\begin{theorem}
  \label{ackermanns-lemma}
  Let $\varphi$, $\psi$ be first-order formulas where $X$ only occurs positively in~$\varphi$ and $X$ does not occur in $\psi$.
  Then
  \begin{align*}
    \Exists{X} (\varphi \land \Forall{\seq{u}} (X(\seq{u}) \limp \psi(\seq{u}, \seq{v}))) \iff \varphi[X \leftarrow \lambda \seq{u}. \psi(\seq{u}, \seq{v})].
  \end{align*}

  Let $\varphi$, $\psi$ be first-order formulas where $X$ only occurs negatively in~$\varphi$ and $X$ does not occur in $\psi$.
  Then
  \begin{align*}
    \Exists{X} (\varphi \land \Forall{\seq{u}} (\psi(\seq{u}, \seq{v}) \limp X(\seq{u}))) \iff \varphi[X \leftarrow \lambda \seq{u}. \psi(\seq{u}, \seq{v})].
  \end{align*}
\end{theorem}
\begin{proof}
  See, e.g.,~\cite[Lemma 6.1]{Gabbay08Second}.
\end{proof}
Note that in both cases Ackermann's Lemma provides a witness $\lambda \seq{u}. \psi(\seq{u}, \seq{v})$.
Therefore Ackermann's Lemma is a method for solving certain instances of WSOQE.
We show that our method extends Ackermann's Lemma on clause sets:
\begin{proposition}
  \label{improvement-of-ackermanns-lemma}
  Let $N'$ be a finite clause set where $X$ only occurs positively and let $C(\seq{u})$ be a clause where $X$ does not occur.
  Further set $N := N' \union \mset{\neg X(\seq{u}) \lor C(\seq{u})}$.
  Then there is an~$X$-eliminating derivation $D$ from $N$ such that $\sigma(D) \iff [X \leftarrow \lambda \seq{u}. C(\seq{u})]$, which is the witness produced by Ackermann's Lemma applied to $\Exists{X} N$.

  Let $N'$ be a finite clause set where $X$ only occurs negatively and let $C(\seq{u})$ be a clause where~$X$ does not occur.
  Further set $N := N' \union \mset{X(\seq{u}) \lor C(\seq{u})}$.
  Then there is an $X$-eliminating derivation $D$ from $N$ such that $\sigma(D) \iff [X \leftarrow \lambda \seq{u}. \neg C(\seq{u})]$, which is the witness produced by Ackermann's Lemma applied to $\Exists{X} N$.
\end{proposition}
\begin{proof}
  Let $N' = \mset{C_1, \dots, C_n}$.
  Informally, the derivation proceeds as follows:
  First, form the resolution closure $\resclosure{P}{N'}$ with respect to the $X$-pointed clause $P = \underline{\neg X(\seq{u})} \lor C(\seq{u})$, resulting in a clause set $N_{\mathrm{Res}_P}$.
  We show later that this process terminates.
  Then delete $P$ by purified clause deletion, resulting in a clause set $N'_{\mathrm{Res}_P}$.
  The clause set $N'_{\mathrm{Res}_P}$ contains $X$ only positively, so extended purity deletion is applicable, resulting in a clause set $N^\ast$.
  More formally, the $X$-eliminating derivation is given by
  $$
  D=\Axiom$N\fCenter$
    \RightLabel{$\mathrm{Res}^\ast$}
    \UnaryInf$N_{\mathrm{Res}_P}\fCenter$
    \RightLabel{$\mathrm{PurDel}_P$}
    \UnaryInf$N'_{\mathrm{Res}_P}\fCenter$
    \RightLabel{$\mathrm{ExtPurDel}_X^{+}$}
    \UnaryInf$N^\ast\fCenter$
    \DisplayProof
  $$
  Note that
  $$\mathrm{\ell Res}_{P} = \lambda \seq{v}. X(\seq{v}) \land \Forall{\seq{u}} (\seq{v} \noeq \seq{u} \lor C(\seq{u})) \iff \lambda \seq{v}. X(\seq{v}) \land C(\seq{v}).$$
  Thus we have 
  $$\sigma(D) \iff [X \leftarrow \mathrm{\ell Res}_P[X \leftarrow \lambda \seq{u}. \top]] = [X \leftarrow \lambda \seq{v}. \top \land C(\seq{v})] \iff [X \leftarrow \lambda \seq{u}. C(\seq{u})].$$

  To finish the proof it remains to show that $\resclosure{P}{N'}$ is finite and to show that the clause set $N'_{\mathrm{Res}_P}$ contains $X$ only positively.
  One can show by induction on $\resclosure{P}{N'}$ that every clause $C \in \resclosure{P}{N'}$ only contains $X$ positively and whenever
  \begin{prooftree}
    \AxiomC{$C$}
    \AxiomC{$P$}
    \BinaryInfC{$R$}
  \end{prooftree}
  is a constraint resolution inference between $C$ and $P$ upon $\neg X(\seq{u})$, then $R$ has one fewer positive occurrences of $X$ than $C$.
  This implies that the number of positive $X$-literals is a termination measure on $\resclosure{P}{N'}$, showing its finiteness.

  The other version can be proved by analogous arguments.
\end{proof}

Note that our method can solve \Cref{ex.main-example.first-witness} where Ackermann's Lemma is not applicable since the clause set contains a clause where $X$ occurs both positively and negatively.
Thus our method strictly extends Ackermann's Lemma.

\subsection*{Computing Witnesses with DLS(*)}
As this work extends a known algorithm for solving SOQE to solve the more general WSOQE problem, the question arises whether other algorithms like DLS \cite{Doherty97Computing} and DLS* \cite{Doherty98General,Nonnengart98Fixpoint} can be equally modified to solve WSOQE and how the resulting witnesses compare.

\subsection*{Application to Forgetting}
Methods for computing forgetting solutions and uniform
interpolants have been developed for description logics and modal
logics~\cite{KoopmannSchmidt13a,KoopmannSchmidt13c,KoopmannSchmidt15a,AlassafSchmidtSattler22}.
For applications such as knowledge processing, agent applications and
Boolean unification, the possibility of extending these to compute also
witnesses for eliminated variables would be attractive.
Using a WSOQE-solution instead of a SOQE-solution has the advantage that the structure of the
existing knowledge base is left largely intact.
Functionality to return witnesses could
provide a promising avenue to develop methods for such applications, but
such development is subject to future work.

\section{Conclusion}
\label{sec.conclusion}

We introduced WSCAN, an extension of the SCAN algorithm on clause sets that computes witnesses for existential second-order quantifiers and thereby solves the more general WSOQE-problem.
Given an $\seq{X}$-eliminating derivation $D$ from a clause set $N$, we can produce a witness for $\Exists{\seq{X}} N$.
This witness may be infinite in general. However, if all purified clause deletions of $D$ have a certain finiteness property, such as being acyclically purified, then we can guarantee that the witness is first-order.
To the best of our knowledge there are currently no other WSOQE-algorithms applicable to arbitrary clause sets.
This work paves the way for new applications of the SCAN algorithm in a variety of areas,
for example, modal correspondence theory, knowledge representation or verification.

On a more abstract level, we see this work as a contribution to bridging the gap between
second-order quantifier elimination (SOQE) and solving formula equations (FEQ), which will be of interest to the
respective communities investigating these problems.
SOQE has been studied mostly in the context of modal logic, description logics, knowledge representation and answer set programming.
On the other hand, FEQ is studied (usually in different formalisms and under different names)
in the verification and automated (inductive) theorem proving communities.
We believe that the complex of problems consisting of SOQE, WSOQE, and FEQ provides a suitable
{\em common logical foundation} for work that is done on all of these topics, as has also been suggested in~\cite{Wernhard17Boolean}.
This perspective not only allows us to relate these seemingly different problems to one another
on a logical level, but also stimulates cross-fertilization in terms of practical solution
ideas, techniques, and algorithms.

\section*{Declarations}

\subsection*{Funding}
This research was funded in part by the Austrian Science Fund (FWF) grant number 10.55776/P35787.

\subsection*{Competing interests}
The authors have no competing interests to declare that are relevant to the content of this article.

\bibliography{./main.bib}

\end{document}